\documentclass[12pt,letterpaper]{article}
\usepackage[utf8]{inputenc}
\usepackage{docmute}
\usepackage{appendix}
\usepackage{subcaption}
\usepackage{placeins}
\usepackage{comment}
\usepackage{pdflscape}

\usepackage{amsthm,amsmath,amssymb}
\usepackage{tikz}
\usetikzlibrary{intersections,calc,arrows.meta,shapes.callouts}
\usepackage{pxpgfmark} 
\usepackage[customcolors,shade,markings]{hf-tikz}
\usepackage{microtype}
\usepackage{letterspace}

\usepackage{graphicx}
\usepackage{setspace}
\usepackage{ascmac}
\usepackage{bm}

\usepackage{xcolor}
\usepackage{float}

\usepackage{xr}
\makeatletter
\newcommand*{\addFileDependency}[1]{
  \typeout{(#1)}
  \@addtofilelist{#1}
  \IfFileExists{#1}{}{\typeout{No file #1.}}
}
\makeatother

\usepackage{algorithm2e}
\usepackage{algorithmic}
\usepackage{enumitem}
\usepackage{latexsym}
\usepackage{threeparttable}

\def\qed{\hfill $\Box$}

\newcommand{\R}{\mathbb{R}}

\newcommand{\N}{\mathbb{N}}

\newcommand{\mB}{\mathcal{B}}

\newcommand{\mX}{\mathcal{X}}

\newcommand{\sumin}{\sum_{i=1}^n}

\newcommand{\pto}{\stackrel{P}{\to}}

\newcommand{\dto}{\stackrel{d}{\to}}
\newcommand{\deq}{\stackrel{d}{=}}

\newcommand{\eps}{\varepsilon}

\DeclareMathOperator{\Var}{Var}
\DeclareMathOperator{\Cov}{Cov}

\DeclareMathOperator{\diag}{diag}
\DeclareMathOperator{\trace}{trace}

\DeclareMathOperator*{\argmin}{arg\,min}

\newcommand{\1}{\mbox{1}\hspace{-0.25em}\mbox{l}}

\renewcommand{\footnotesize}{\normalsize}

\theoremstyle{plain}
\newtheorem{theorem}{Theorem}[section]
\newtheorem{lemma}{Lemma}[section]
\newtheorem{proposition}{Proposition}[section]

\newtheorem{assumption}{Assumption}

\theoremstyle{definition}

\newtheorem{remark}{Remark}

\usepackage[top=1.25in, bottom=1.25in, left=1.25in, right=1.25in]{geometry}

\numberwithin{equation}{section}

\usepackage{color}

\AtBeginDocument{
  \abovedisplayskip     =0.6\abovedisplayskip
  \abovedisplayshortskip=0.6\abovedisplayshortskip
  \belowdisplayskip     =0.6\belowdisplayskip
  \belowdisplayshortskip=0.6\belowdisplayshortskip}

\let\oldenumerate\enumerate
\renewcommand{\enumerate}{
\oldenumerate
\setlength{\itemsep}{1.5pt}
\setlength{\parskip}{1.5pt}
\setlength{\parsep}{1.5pt}
}

\usepackage{hyperref}
\hypersetup{
       colorlinks=false,
       citebordercolor=green,
       linkbordercolor=red,
       urlbordercolor=cyan,
}
\PassOptionsToPackage{linktocpage}{hyperref}
\usepackage{breakcites}
\usepackage[authoryear, round, longnamesfirst]{natbib}
	\title{A unified test for regression discontinuity designs}
	
	\author{Koki Fusejima\thanks{Hitotsubashi Institute for Advanced Study, Hitotsubashi University; 2-1 Naka, Kunitachi, Tokyo 186-8601, Japan; k.fusejima@r.hit-u.ac.jp},
	Takuya Ishihara\thanks{Graduate School of Economics and Management,
     Tohoku University; 27-1 Kawauchi, Aoba-ku, Sendai, Miyagi 980-8576, Japan; takuya.ishihara.b7@tohoku.ac.jp}, and 
    Masayuki Sawada\thanks{Corresponding Author. Institute of Economic Research,
  Hitotsubashi University; 2-1 Naka, Kunitachi,
Tokyo 186-8601, Japan; m-sawada@ier.hit-u.ac.jp}}
	
	\date{}
	
\begin{document}

	\maketitle
 
	\begin{abstract}
	Diagnostic tests for regression discontinuity design face a size-control problem. We document a massive over-rejection of the diagnostic restriction among empirical studies in the top five economics journals. At least one diagnostic test was rejected for $19$ out of $59$ studies, whereas less than 5\% of the collected $787$ tests rejected the null hypotheses. In other words, one-third of the studies rejected at least one of their diagnostic tests, whereas their underlying identifying restrictions appear plausible. Multiple testing causes this problem because the median number of tests per study was as high as $12$. Therefore, we offer unified tests to overcome the size-control problem. Our procedure is based on the new joint asymptotic normality of local polynomial mean and density estimates. In simulation studies, our unified tests outperformed the Bonferroni correction. We implement the procedure as an R package \href{https://github.com/SMasa11/rdtest}{rdtest} with two empirical examples in its vignettes.
	\end{abstract}

 \newpage
 
\section{Introduction}
Diagnostic tests are vital to regression discontinuity (RD) design. The typical procedure evaluates the testable restrictions which are expected to hold under a non-testable identifying restriction \citep{hahnIdentificationEstimationTreatment2001}. All such testable restrictions are expected to hold whenever identification holds. Hence, one must test the null hypothesis of \textit{all restrictions hold} against the alternative hypothesis of \textit{at least one restriction fails}.

However, these restrictions are often tested separately without appropriate size control. Throughout this study, we consider the \textit{continuity-based} framework. \footnote{In \citeauthor{Cattaneo_Idrobo_Titiunik_2019}(\citeyear{Cattaneo_Idrobo_Titiunik_2019,Cattaneo_Idrobo_Titiunik_2023}), the \textit{continuity-based} framework is a contrasting concept against the alternative \textit{local randomization} framework for a finite sample analysis in which explicit randomization is considered within a small range as studied in \cite{Cattaneo_Frandsen_Titiunik_2015, Cattaneo_Titiunik_Vazquez-Bare_2016, cattaneoComparingInferenceApproaches2017}.} In this framework, most empirical studies test these restrictions as null hypotheses for the density test \citep{mccraryManipulationRunningVariable2008} and the balance or placebo test \citep{leeRandomizedExperimentsNonrandom2008}. These tests are well-established separately but are rarely evaluated jointly. Hence, the multiple-testing problem plagues size control. Consequently, empirical studies may over-reject the underlying null hypothesis that all restrictions hold under identification.

In this study, we document the severe over-rejection of the diagnostic restriction in a meta-analysis of $59$ empirical RD studies published recently in the top five economics journals. \footnote{See Online Appendix B 
for the search criterion.} We found that $19$ out of $59$ studies rejected at least one testable restriction, leading to the underlying restriction being questioned in $32\%$ of the studies. 

Nevertheless, each test appeared valid. Among $787$ tests run separately, only $4.3\%$ of the restrictions were rejected, and their p-values uniformly distributed over $[0,1]$, suggesting that the test statistics are drawn from the null hypothesis that all testable restrictions hold. Hence, the underlying identification was plausible; nonetheless, we ended up with an incorrect conclusion for $19$ studies rather than three. 

We conclude that this over-rejection was caused by multiple testing problems. In $59$ studies, the median number of tests was 12 per study. Furthermore, we found only five studies out of $59$ reporting some type of joint testing. No study has reported a joint test that combines the density and balance or placebo tests because no test has been proposed to handle them jointly.


In this study, we propose a unified diagnostic test to resolve this over-rejection problem. This involves two major challenges. First, no tests consider the joint nature of the density and balance tests. Second, the test statistics are nonparametric estimates. Only one of the five studies considered the nonparametric nature of the test statistics. To address these two challenges, we derived the joint asymptotic normality of the nonparametric density estimation and nonparametric conditional mean estimation. We find that nonparametric density estimates and conditional mean estimates are asymptotically orthogonal. Based on this finding, we provide the first unified test with density and balance test statistics based on our new theoretical results from their nonparametric bias-corrected estimates.

Our unified test aggregates the vector of the density and balance test statistics; however, possible aggregation is not unique. An appropriate aggregation depends on the underlying alternative hypothesis. We considered two different aggregations: (1) a Wald (chi-squared) statistic of the sum of the squared statistics of the vector and (2) a Max test statistic of the maximum squared statistics of the vector. The two test statistics complement each other in detecting different types of alternative hypotheses. 

The Wald statistic detects the alternative hypothesis that a large fraction of the original vector deviates from the null hypothesis. The Max statistic detects the alternative hypothesis that only a small fraction of the vector deviates from the null hypothesis. Nevertheless, the conventional Wald statistic did not attain the nominal size, possibly because of poor inverse estimation of the variance matrix. We bypass this problem by standardizing only the diagonal elements instead of multiplying the inverse of the variance matrix and propose it as the \textit{standardized} Wald (hereafter, sWald) statistic. 

In numerical simulations, the sWald and Max statistics exhibit appropriate size control properties for a moderate number of covariates when the sample size is as large as $1000$. In particular, the sWald test is superior to other methods in size control performance for any number of covariates up to $25$ in cases where even the Bonferroni correction may fail to achieve size control. \footnote{We conjecture that the reason for this failure is that Bonferroni correction with $25$ covariates evaluates the thin tail of the nonparametric test statistic distribution, which could be imprecisely estimated.}

We verified the superior power properties of the proposed methods. The Max test is superior to the Bonferroni correction in its power against the alternative hypothesis that only one covariate sees a jump in the conditional mean. The sWald test is superior to the Bonferroni correction in its power against the alternative hypothesis that all covariates see jumps in their conditional means. In conclusion, we recommend the Max test if a user has a few particularly concerning covariates. Otherwise, we recommend the sWald test for better size control and power in other scenarios. \footnote{We also compare our methods with the resampling based method of \cite{romanoExactApproximateStepdown2005, romanoEfficientComputationAdjusted2016} using a \textit{stata} package \textit{rwolf2} \citep{clarkeRwolf2ImplementationFlexible2021} in Online Appendix D. We note that the applicability of \cite{romanoExactApproximateStepdown2005} is not necessarily guaranteed for the RD context, and \textit{rwolf2} can accept the balance tests only. Hence, our simulation comparisons with \textit{rwolf2} are limited to the balance tests only. In the simulation, the Romano and Wolf correction exhibits a bimodal performance: it achieves a good power for the large effect sizes with high correlations, but the power is much lower than our methods when the effect sizes are small. Hence, we continue to support our methods for the purpose of detecting possibly small violations in diagnostic testings.}

This study contributes to the literature on diagnostic procedures for RD design. \footnote{See the latest textbooks \citeauthor{Cattaneo_Idrobo_Titiunik_2019}(\citeyear{Cattaneo_Idrobo_Titiunik_2019,Cattaneo_Idrobo_Titiunik_2023}) for an extensive survey of RD literature.} \cite{mccraryManipulationRunningVariable2008} has developed the density test. Several theoretical studies have improved these procedures. \footnote{Several theoretical studies also provide other forms of testing procedures in regression discontinuity designs. For example, \cite{hsuTestingTreatmentEffect2019} and \cite{hsuTestingMonotonicityConditional2021} provide tests for the presence of treatment effect heterogeneity across covariate values and for monotonicity in the conditional treatment effects; \cite{Bertanha_Imbens_2020} provide tests for exogeneity and external validity in fuzzy designs.}
\citet{otsuEstimationInferenceDiscontinuity2013} have proposed an empirical likelihood-based density test, \citet{cattaneoSimpleLocalPolynomial2019} have considered a test based on a local polynomial estimation of the density function, and \cite{bugniTestingContinuityDensity2020a} have proposed an approximate sign test for the density test. Nevertheless, none of these studies have considered the joint procedure for these tests. Diagnostic test practices follow a similar pattern: separate multiple testing of these diagnostic tests or some joint procedure without considering the nonparametric nature of the test statistics. In this study, we provide unified testing for the null hypothesis that all these diagnostic restrictions hold, suggesting that identification is plausible. We also provide a statistical software \textit{rdtest}, which is an explicit wrapper of two standard packages, \textit{rdrobust} (\citealp{Calonico_Cattaneo_Farrell_Titiunik_2017}) and \textit{rddensity} \citep{Cattaneo_Jansson_Ma_2018}. Therefore, we provide an easy-to-implement unified diagnostic procedure for the widely used continuity-based framework.

Moreover, we contribute to the asymptotic theory of nonparametric estimates. We show the joint asymptotic normality of multiple local polynomial estimates for the conditional mean and density functions. \citet{calonicoRobustNonparametricConfidence2014a} show the asymptotic normality of a single local polynomial estimate for a conditional mean function at a boundary point after bias correction. \citet{cattaneoSimpleLocalPolynomial2019} also show the asymptotic normality of a local polynomial estimate of a density function at a boundary point. We developed the joint normality results based on these results. The marginal bias and variance expressions are identical to those in \cite{calonicoRobustNonparametricConfidence2014a} and \cite{cattaneoSimpleLocalPolynomial2019}, and we impose no restrictions on the choice of bandwidth apart from the original rate conditions. Consequently, we succeeded in accommodating the same mean squared error (MSE) optimal estimates for the balance and placebo tests \citep{calonicoRobustNonparametricConfidence2014a} \footnote{Recently, another bandwidth selector of coverage error optimal bandwidths has been proposed by \citet{Calonico_Cattaneo_Farrell_2020} along with theoretical background \citet{calonicoCoverageErrorOptimal2022}. In our numerical analyses and implementation, we used the default option for MSE optimal bandwidths.} and density tests \citep{cattaneoSimpleLocalPolynomial2019}.

The remainder of the paper proceeds as follows. Section \ref{sec:2} demonstrates the multiple testing problem in the meta-analysis, and
Section \ref{sec:3} explains the joint asymptotic normality of the nonparametric estimates. Subsequently, Section \ref{sec:4} presents our joint tests, Section \ref{sec:5} demonstrates the performance of our tests in the simulation, and Section \ref{sec:extension} considers the two extensions for our procedure: equivalence testing which flips the null and alternative hypotheses, and pretesting analysis in view of \cite{roth2022pretest}. Finally, Section \ref{sec:6} concludes the article by recommending practices and proposing future research questions. In Online Appendix A
, we provide all the formal results and the proofs.

\section{Current practices of diagnostic tests}\label{sec:2}

Most empirical studies diagnose their designs using two continuity restrictions: the density function of the running variable and the conditional mean functions of the pretreatment covariates. In this study, we focus on the two major practices that are routinely evaluated with well-established packages. There are other popular diagnoses such as the \textit{placebo cutoff} analysis. See \citet[Section 4]{cattaneoRegressionDiscontinuityDesigns2022} and \citet[Section 5]{Cattaneo_Idrobo_Titiunik_2019} for further diagnostic concepts.

 The logic behind the diagnoses is that ``the RD design offers an array of empirical methods that, under reasonable assumptions, can provide useful evidence about the plausibility of its assumptions.'' \citep{Cattaneo_Idrobo_Titiunik_2019} In the continuity-based framework, identification holds when
 the conditional mean function of the potential outcome is continuous in the conditioning running variable at the cutoff. In other words, the units near the cutoff must be similar in the unobserved characteristics of the potential outcomes. 
 
 However, units often know the cutoff value and hence they may manipulate the running variable to be eligible for treatment. A typical idea to assess such manipulation is that ``if units lack the ability to precisely manipulate the score value they receive, there should be no systematic differences between units with similar values of the score'' \citep{Cattaneo_Idrobo_Titiunik_2019}. This idea conveys a testable restriction that the observed characteristics must be balanced near the cutoff. 
 
 Similar idea is that ``if units do not have the ability to precisely manipulate the value of the score that they receive, the number of treated observations just above the cutoff should be approximately similar to the number of control observations below it'' \citep{Cattaneo_Idrobo_Titiunik_2019}. This latter idea conveys another testable restriction that the observed density function of the running variable must be continuous at the cutoff. The former restriction is often conducted as the balance or placebo test, and the latter restriction is evaluated as the density test. 
 
 All those testable restrictions stem from the ideas that support identification. Hence, if we claim the plausibility of identification, all those diagnostic restrictions should hold, and none of them must be violated. \footnote{Without additional restrictions, these restrictions have no direct connection with identification. For example, the density test is neither necessary nor sufficient for identification \citep{mccraryManipulationRunningVariable2008}. See \cite{Ishihara_Sawada_2022} for conditions under which the null hypothesis of the density test implies identification.} 
 
 In other words, we must test the null hypothesis that \textit{all testable restrictions hold} against the alternative hypothesis that \textit{at least one restriction fails} to diagnose the underlying designs. However, most studies have tested these restrictions separately as density tests \citep[for example]{mccraryManipulationRunningVariable2008, cattaneoSimpleLocalPolynomial2019} and balance or placebo tests \citep[for example]{leeRandomizedExperimentsNonrandom2008, leeRegressionDiscontinuityDesigns2010}. Running these tests separately would over-reject the underlying null hypothesis that all testable restrictions hold because of the multiple testing problem.

We quantified the size control problem by conducting a meta-analysis of RD studies published in the top five economics journals. Among $59$ papers that satisfied our search criteria up to November 2021, we collected $787$ reported test statistics. \footnote{Many studies aim to demonstrate their robustness with different specifications of the same tests with different bandwidths, order of polynomials, or kernels. We screen the \textit{main} specification for each test, and our reported results are the lower bounds of their over-rejections. See Online Appendix B 
for details.} Our meta-analysis found two facts: (1) the reported test statistics appear to conform to the null hypothesis that all restrictions hold; (2) nonetheless, $32\%$ ($19$) studies had at least one rejected test.

\subsection{Fact 1: Test statistics are likely drawn from the null distributions}

First, we document the reported test statistics likely to be drawn from null distributions. Table \ref{tab:rejectFracAmongAll} presents the percentage of tests rejected from the reported ones. Among $787$ tests reported in $59$ studies, we found that less than $5\%$ rejected null hypotheses. Figure \ref{fig:quantile_test} shows the quantile plot of the p-values. The quantile plot aligns with the $45$-degree line, and the distribution is likely to be uniform over $[0,1]$. \footnote{This is also true among the rejected studies. See Appendix Figure 7 for the histogram of p-values among the rejected studies.} Both lines of evidence support the null hypothesis that all restrictions hold and hence identification is plausible.
\begin{table}[H]
    \centering
    \caption{Percentage of rejected tests}
    \begin{tabular}{l*{3}{c}}
                &  Balance&   Density&       Total\\
    \hline
    Rejected      &        4.32\%&        4.35\%&        4.32\%\\
    Not         &        95.68\%&       95.65\%&       95.68\%\\
    \hline
    Number of tests     &         764&          23 &        787\\
    \end{tabular}
    
    \flushleft

    \begin{minipage}{380pt}
{\fontsize{10pt}{10pt}\selectfont\smallskip\textit{Notes}: Percentages of tests reported as \textit{rejected} and \textit{non-rejected} among $764$ balance tests, $23$ density tests, and $787$ tests in total. See Online Appendix B 
for details.}
\end{minipage}
  \label{tab:rejectFracAmongAll}
\end{table}

\begin{figure}[H]
  \centering
  \caption{Quantile plot of p-values of tests}
  \includegraphics[width=0.9\linewidth]{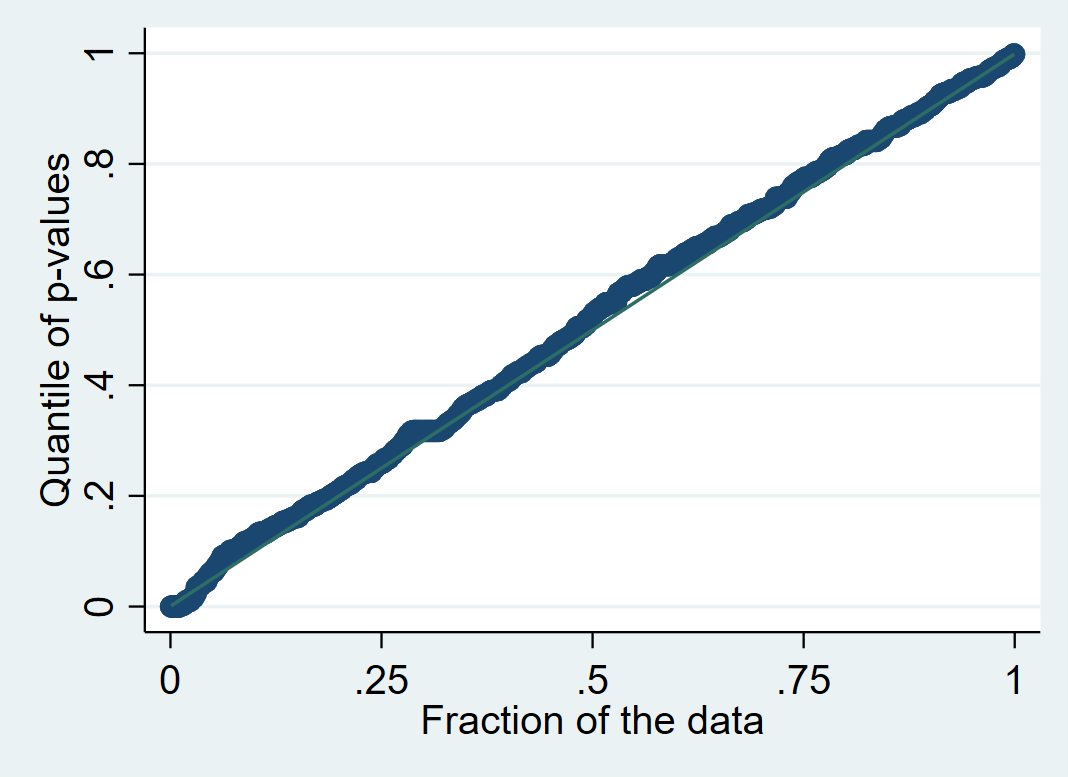} 

  \flushleft
 \begin{minipage}{380pt}
{\fontsize{10pt}{10pt}\selectfont\smallskip\textit{Notes}: The quantile plot of p-values that are reported or computed from t-statistics. Plot aligned with the solid 45 degree line implies similarity to the uniform distribution over $[0,1]$.}
\end{minipage}
  \label{fig:quantile_test}
\end{figure}

\subsection{Fact 2: The underlying null hypothesis is over-rejected because of multiple testing}

From Fact 1, approximately three studies out of $59$ should have shown any test rejection because the null hypotheses appear valid. However, Table \ref{tab:rejectProb} suggests that $32\%$ of the studies ($19$ of $59$) rejected at least one null hypothesis. Among $19$ studies, $8$ studies rejected more than once. As shown in Table \ref{tab:rejectProb} and Figure \ref{fig:histogram}, many studies ran more than $10$ diagnostic tests without appropriate size control. \footnote{Also, studies with many rejections tend to run excessively many tests. See Appendix Figure 8 for the scatter plot of the number of rejected tests and the number of tests.} Hence, we blame the multiple testing problem for over-rejecting the null hypothesis of \textit{all restrictions hold}.


\begin{table}[H]
  \centering
  \caption{Fraction of rejecting at least one hypothesis.}
\begin{tabular}{l*{1}{cccccc}}
            &        mean&      median &    sd\\
\hline
Number of tests per study &    14.20&   12 &  9.59\\
Fraction rejecting at least one test &         0.32& - &  0.47\\
Number of rejected tests per study&         0.69&  0 &  1.52\\
\hline
Number of studies      &          59 & \\
\end{tabular}

  \flushleft
   \begin{minipage}{380pt}
{\fontsize{10pt}{10pt}\selectfont\smallskip\textit{Notes}: The table shows the mean, median, and standard deviation for the number of tests per study, the fraction of rejecting at least one test per study, and the number of rejected tests per study out of $59$ studies.}
\end{minipage}

  \label{tab:rejectProb}
\end{table}
\begin{figure}[H]
  \centering
  \vspace{0.5cm}
    \caption{Histogram of the number of tests}
    \centering
    \includegraphics[width=0.8\hsize]{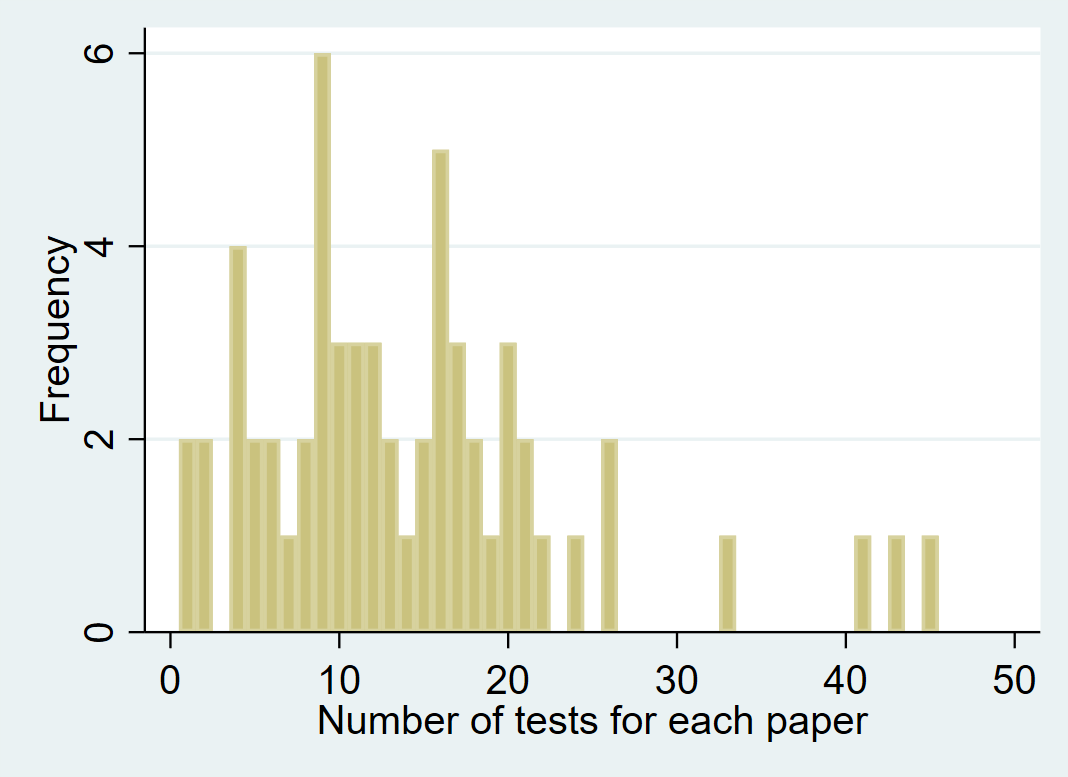}
    \label{fig:histogram}

  \flushleft
 \begin{minipage}{380pt}
{\fontsize{10pt}{10pt}\selectfont\smallskip\textit{Notes}: This figure is the histogram of the number of tests reported in each study.}
\end{minipage}
\end{figure}

\subsection{How does empirical studies respond to their rejections}

Upon rejection of their diagnostic tests, most studies defended their statistical significance in diagnostic tests by making some supporting arguments such as stating \textit{economical} insignificance of their rejected test statistics. Those supporting arguments are indeed desirable. However, the majority of the studies did not address the root cause of their rejections, the multiple testing problem.

Notably, only five studies reported joint testing for balance tests. All but one study did not consider the nonparametric nature of the test statistics. \footnote{We confirm some sort of joint testing in the following five studies: \cite{meyerssonIslamicRuleEmpowerment2014, abdulkadirogluEliteIllusionAchievement2014, johnsonRegulationShamingDeterrence2020, Fort_Ichino_Zanella_2020, jagerLaborBoardroom2021}. \cite{Fort_Ichino_Zanella_2020} use a nonparametric procedure by \cite{canayApproximatePermutationTests2018}.} Furthermore, no studies have reported joint testing of the balance and density tests because no method is available for the joint test.

We emphasize that only a few studies made explicit argument that \textit{testing many hypotheses} can result in their rejections of the diagnostic tests. Furthermore, their arguments were all informal without any multiple testing corrections. For example, \cite{tyreforshinnerichDemocracyRedistributionPolitical2014} stated that ``Only one of the 44 estimates is significant at the 5 percent level. However, that
is to be expected since, if 100 specifications are tested, it is likely that five will
be statistically significant by chance'' (page 981). With a unified diagnostic test, these studies could have formally concluded that their rejections were by chance in their data.

\section{Unified test based on joint asymptotic normality}\label{sec:3,4}
As shown in Section \ref{sec:2}, the multiple-testing problem causes the over-rejection of the underlying restriction which suggests plausible identification. We introduce a notation to formally state the problem. Let $(Z_{1,i},\ldots,Z_{d,i},X_i)'$, $i=1,\ldots,n$ denote the observed random sample, where $X_i$ is the running variable with density $f(x)$ with respect to the Lebesgue measure and $Z_i=(Z_{1,i},\ldots,Z_{d,i})'$ denote the pre-treatment covariates. 
Given a known threshold $\Bar{x}$, which we set to $\Bar{x}=0$ without loss of generality, let the running variable $X_i$ determine whether unit $i$ is assigned a treatment ($X_i\geq0$) against ($X_i<0$). In this notation, empirical applications frequently consider two different types of restrictions: the density function $f(x)$ is continuous at $x = 0$, and the conditional mean functions of $Z_k$, $\mu_{Z_k}(x)=E[Z_{k,i}|X_i=x]$ for $k=1,\ldots,d$, are continuous at $x = 0$. Specifically, the following $d+1$ restrictions must hold simultaneously
\[
 f_+ = f_-, \mu_{Z_1+} = \mu_{Z_1-}, \ldots, \mu_{Z_d+} = \mu_{Z_d-}
\]
where $f_+ = \lim_{x\to0^+} f(x)$, $f_- = \lim_{x\to0^-} f(x)$ and $\mu_{Z_k+}=\lim_{x\to0^+}\mu_{Z_k}(x)$, $\mu_{Z_k-}=\lim_{x\to0^-}\mu_{Z_k}(x)$.
Hence, one must conduct the following test:
\begin{equation*}
H_0: (\tau_f,\tau_{Z_1},\ldots,\tau_{Z_d})=0 \quad \text{vs} \quad
H_1: (\tau_f,\tau_{Z_1},\ldots,\tau_{Z_d})\neq0,
\end{equation*}
to determine whether identification is plausible where $\tau_f = f_+ - f_-$ and $\tau_{Z_k} = \mu_{Z_k+} - \mu_{Z_k-}$. However, if one runs $d+1$ tests separately:
\begin{align*}
 H_{0f}:& \tau_f = 0 \quad \text{vs} \quad H_{1f}: \tau_f \neq 0, \\
 H_{01}:& \tau_{Z_1} = 0 \quad \text{vs} \quad H_{11}: \tau_{Z_1} \neq 0, \\
 & \vdots \\
 H_{0d}:& \tau_{Z_d} = 0 \quad \text{vs} \quad H_{1d}: \tau_{Z_d} \neq 0,
\end{align*}
the null hypothesis $H_0$ is over-rejected as documented in Section \ref{sec:2}.

We resolved the size control problem by presenting a unified test in two steps. First, we present the theoretical basis for our unified test. Second, we propose two test statistics for the unified test against multiple testing.

\subsection{New joint normality result}\label{sec:3}
First, we derive the joint asymptotic normality of the nonparametric local polynomial estimators for $f$ and $(\mu_{Z_1}, \ldots, \mu_{Z_d})$ evaluated at the boundary of the support.

Following \cite{calonicoRobustNonparametricConfidence2014a} and \cite{cattaneoSimpleLocalPolynomial2019}, we use local polynomial regressions to approximate the unknown functions $f(x)$ and $\mu_{Z_k}(x)$ for $k=1,\ldots,d$ flexibly near the cutoff $\Bar{x}=0$; see \cite{fan1996local} for a review of local polynomial regressions. Based on the testing problem, we must obtain two distinct estimands for each function using the two subsamples $\{X_i: X_i\geq 0\}$ (approximation from the right) and $\{X_i: X_i< 0\}$ (approximation from the left) at a boundary point $x = 0$. These estimators are recommended owing to their excellent boundary properties. 

Using a polynomial expansion, we construct smooth local one-sided approximations of the sample average of $Z_{k,i}$ and obtain the estimators $\hat \mu_{Z_k+,l}(h_{k,n})$ and $\hat \mu_{Z_k-,l}(h_{k,n})$ as intercepts in the local polynomial regression. Specifically, we estimate $\tau_{Z_k}$ using the following $l$th-order local polynomial estimators: for some $l\geq1$ and for each $k=1,\ldots,d$,
\begin{equation*}
\hat\tau_{Z_k,l}(h_{k,n})=\hat \mu_{Z_k+,l}(h_{k,n})-\hat \mu_{Z_k-,l}(h_{k,n}),
\end{equation*}
\begin{eqnarray}\label{eq:1}
\hat \mu_{Z_k+,l}(h_{k,n}) &=& e_0'\hat{\beta}_{Z_{k}+,l}(h_{k,n}), \\ 
\hat \mu_{Z_k-,l}(h_{k,n}) &=& e_0'\hat{\beta}_{Z_{k}-,l}(h_{k,n}),
\end{eqnarray}
where
\begin{eqnarray}\label{eq:2}
\hat{\beta}_{Z_k+,l}(h_{k,n}) &=& \argmin_{\beta\in\R^{l+1}}
\sumin1\{X_i \geq 0\} (Z_{k,i}-r_l(X_i)'\beta)^2K(X_i/h_{k,n})/h_{k,n}, \\
\hat{\beta}_{Z_k-,l}(h_{k,n}) &=& \argmin_{\beta\in\R^{l+1}}
\sumin1\{X_i < 0\} (Z_{k,i}-r_l(X_i)'\beta)^2K(X_i/h_{k,n})/h_{k,n},
\end{eqnarray}
$e_0=(1,0,\ldots,0)'$ is the first $(l+1)$-dimensional unit vector, $r_l(x)=(1,x,\ldots,x^l)'$ is a $l$th order polynomial expansion, $K(\cdot)$ is a kernel function, and $h_{k,n}$ is a positive-bandwidth sequence.

Similarly, using a polynomial expansion, we construct smooth local one-sided approximations of the empirical distribution of $X_i$ and obtain the estimators $\hat f_{+,p}(h_{f,n})$ and $\hat f_{-,p}(h_{f,n})$ as the slope coefficients in the local polynomial regression. Specifically, we estimate $\tau_{f}$ using the following $p$th-order local polynomial estimator. For some $p\geq2$,
\begin{equation*}
\hat\tau_{f,p}(h_{f,n})=\hat f_{+,p}(h_{f,n}) - \hat f_{-,p}(h_{f,n}),
\end{equation*}
\begin{eqnarray}\label{eq:f1}
\hat f_{+,p}(h_{f,n}) &=& e_1'\hat{\beta}_{f+,p}(h_{f,n}), \\
\hat f_{-,p}(h_{f,n}) &=& e_1'\hat{\beta}_{f-,p}(h_{f,n}),
\end{eqnarray}
where
\begin{eqnarray}
\hat{\beta}_{f+,p}(h_{f,n})=\argmin_{\beta\in\R^{p+1}}
\sumin1\{X_i \geq 0\} (\tilde{F}(X_i)-r_p(X_i)'\beta)^2K(X_i/h_{f,n})/h_{f,n},\label{eq:f2.1} \\
\hat{\beta}_{f-,p}(h_{f,n})=\argmin_{\beta\in\R^{p+1}}
\sumin1\{X_i < 0\} (\tilde{F}(X_i)-r_p(X_i)'\beta)^2K(X_i/h_{f,n})/h_{f,n},\label{eq:f2.2}
\end{eqnarray}
$e_1=(0,1,0,\ldots,0)'$ is the second $(p+1)$-dimensional unit vector; $\tilde{F}(x) \equiv n^{-1}\sumin$\\$ 1\{X_i \leq x\}$ is the empirical distribution function of $X_i$, and $h_{f,n}$ is the positive-bandwidth sequence.

We derive the asymptotic distributions of the local polynomial estimators based on the assumptions in \cite{calonicoRobustNonparametricConfidence2014a} and \cite{cattaneoSimpleLocalPolynomial2019}, as follows.
\begin{assumption}\label{asp:1}
For some $\kappa_0 > 0$, the following holds true for $(-\kappa_0,0)\cup(0,\kappa_0)$ around the cutoff $\Bar{x}=0$.
\begin{enumerate}
\item[(a)] $f(x)$ is $R$ times continuously differentiable for some $R \geq 1$ and bounded away from zero.
\item[(b)] For $k=1,\ldots,d$, $\mu_{Z_k}(x)$ is $S$ times continuously differentiable for some $S \geq 1$, $\Var(Z_{k,i}|X_i=x)$ is continuous and bounded away from zero, and $E[|Z_{k,i}|^4|X_i=x]$ is bounded.
\end{enumerate}
\end{assumption}
Assumption \ref{asp:1} is a set of basic regularity conditions for the data generation process.
Assumption \ref{asp:1} (a) ensures that the density of the running variable is well defined and sufficiently smooth and that the observed values are arbitrarily close to the cutoff for large samples.
Assumption \ref{asp:1} (b) ensures that the conditional expectations and variances of the pre-treatment covariates are well defined and have sufficient smoothness with the existence of the fourth moments.

\begin{assumption}\label{asp:2}
For some $\kappa > 0$, the kernel function $K(\cdot):\R\to\R$ is symmetric, bounded, and non-negative; it is zero outside the support, and positive and continuous on $(-\kappa,\kappa)$.
\end{assumption}
Assumption \ref{asp:2} is a set of standard conditions for the kernel function and is satisfied for kernels such as triangular and uniform kernels, commonly used in empirical work. Our results can be extended to accommodate kernels with unbounded supports or asymmetric kernels with more complex notations.

We provide a joint asymptotic distribution for the local polynomial estimators. In the following, let $l,p\in\N$; $h_n=(h_{1,n},\ldots,h_{d,n})'$; $h_{max,n}=\max\{h_{1,n},\ldots,h_{d,n}\}$; $h_{min,n}=\min\{h_{1,n},\ldots,h_{d,n}\}$;
we assume that all the limits are $n\to\infty$ (and $h_{k,n}, h_{f,n}\to0$), unless stated otherwise.
Whenever there was no confusion, we dropped the sample size $n$ sub-index notation.
For derivatives $\mu_{Z_k}^{(s)}(x)$ 
and $f^{(s)}(x)$, let 
$\mu_{Z_k+}^{(s)}=\lim_{x\to0^+}\mu_{Z_k}^{(s)}(x)$, $\mu_{Z_k-}^{(s)}=\lim_{x\to0^-}\mu_{Z_k}^{(s)}(x)$,
$f_+^{(s)} = \lim_{x\to0^+} f^{(s)}(x)$, 
and $f_-^{(s)} = \lim_{x\to0^-} f^{(s)}(x)$.
\begin{theorem}\label{thm:1}
Suppose that Assumptions \ref{asp:1} and \ref{asp:2} hold, with $R \geq p + 1$ and $S\geq l+2$. If $n h_{min}\to\infty$, $nh_f^2 \to\infty$, $n (h_{max})^{2l+5}\to0$ and $n (h_f)^{2p+3} \to0$, then
\[
\begin{split}
\sqrt{n} 
\begin{bmatrix}
\sqrt{h_{1}}(\hat\tau_{Z_1,l}(h_1) - \tau_{Z_1} - (h_1)^{1+l}B_{Z_1,l,l+1}(h_1)) \\
\vdots \\
\sqrt{h_{d}}(\hat\tau_{Z_d,l}(h_d) - \tau_{Z_d} - (h_d)^{1+l}B_{Z_d,l,l+1}(h_d)) \\
\sqrt{h_{f}}(\hat\tau_{f,p}(h_f) - \tau_f - h_f^{p} B_{f,p,p+1}(h_f))
\end{bmatrix}
\dto& N_{d+1}(0,V_{l,p}),
\end{split}
\]
where $N_{d+1}(\mu,\Sigma)$ denotes the multivariate normal distribution with $(d+1)$-dimensional mean vector $\mu$ and
$(d+1)\times(d+1)$ covariance matrix $\Sigma$,
\[
\begin{split}
B_{Z^k,l,r}(h_k) 
=& B_{Z^k+,l,r}(h_k) - B_{Z^k-,l,r}(h_k) \\
=& \frac{\mu^{(r)}_{Z_k+}}{r!}\mB_{+,0,l,r}(h_k)
- \frac{\mu^{(r)}_{Z_k-}}{r!}\mB_{-,0,l,r}(h_k),
\end{split}
\]
\[
\begin{split}
B_{f,p,q}(h_f) 
=& B_{f,+,p,q}(h_f) - B_{f,-,p,q}(h_f) \\
=& \frac{f_+^{(q-1)}}{(q-1)!} \mB_{+,1,p,q}(h_f)
- \frac{f_-^{(q-1)}}{(q-1)!} \mB_{-,1,p,q}(h_f),
\end{split}
\]
and
\[
V_{l,p} = 
\begin{pmatrix}
    V_{Z,l} & 0\\
    0& V_{f,p}
\end{pmatrix}.
\]
The exact forms of $\mB_{+,s,l,r}(h)$ and $\mB_{-,s,l,r}(h)$ are provided in Online Appendix A.2. 
The exact forms of $V_{Z,l}$ and $V_{f,p}$ are provided in Online Appendix A.4.
\end{theorem}

Theorem \ref{thm:1} shows that the joint asymptotic distribution of $\hat\tau_{f,p}(h_f)$ and $\hat\tau_{Z_k,l}(h_k)$ is a multivariate normal distribution under the same set of conditions as in \cite{calonicoRobustNonparametricConfidence2014a} for $\hat\tau_{Z_k,l}(h_k)$ and \cite{cattaneoSimpleLocalPolynomial2019} for $\hat\tau_{f,p}(h_f)$. Hence, the proposed joint distributional approximation is a natural generalization of the distributional approximation of each local polynomial estimator in these studies.

We emphasize that the covariance matrix of $(\hat\tau_{Z_1,l}(h_1),\ldots,\hat\tau_{Z_d,l}(h_d))'$ and the variance of $\hat\tau_{f,p}(h_f)$ are asymptotically orthogonal. Orthogonality simplifies the estimation of the overall asymptotic covariance matrix because the asymptotic covariance of $\hat\tau_{f,p}(h_f)$ and $\hat\tau_{Z_k,l}(h_k)$ is known to be zero. We derive this orthogonality from the well-known fact that $\Cov(Z_{k,i} - \mu_{Z_k}(X_i) ,g(X_i))=0$ for any function $g$.

Orthogonality affects the nature of the multiple testing problem in various ways. When the test statistics are independent, the multiple testing problem is likely worse. Hence, the size control problem in the naive test procedure becomes worse. In turn, a typical multiple test correction such as Bonferroni correction may become less conservative. For our procedure, we have one less parameter to estimate, and the orthogonality may help stabilizing our procedure.

\begin{remark}[Increasing the polynomial order for bias correction]
We use bias-corrected estimators for the following analyses. In general, bias correction is critical in empirical contexts \citep{hyytinenWhenDoesRegression2018}. As highlighted in \cite{calonicoRobustNonparametricConfidence2014a} and \cite{cattaneoSimpleLocalPolynomial2019}, bias correction is important for inference.
As provided in Online Appendix A.2
, $\mB_{+,s,l,r}(h)$ and $\mB_{-,s,l,r}(h)$ are the observed quantities, and the bias-corrected estimators are constructed by replacing $\mu^{(r)}_{Z_k+}$, $\mu^{(r)}_{Z_k-}$, $f_+^{(q-1)}$, and $f_-^{(q-1)}$ with their local polynomial estimators.
Nevertheless, as shown in \cite{calonicoRobustNonparametricConfidence2014a}, if the same bandwidth and kernel function are used for estimating $\mu^{(r)}_{Z_k+}$ and $\mu^{(r)}_{Z_k-}$, the bias-corrected estimator is numerically equivalent to $\hat\tau_{Z_k,l+1}(h_k)$ (with no bias-correction).
A bias correction based on similar results has been adopted by \cite{cattaneoSimpleLocalPolynomial2019} for their local polynomial density estimators. This same bandwidth choice is shown to be optimal for the uniform kernel \citep{calonicoCoverageErrorOptimal2022}.
Following their approach, we implement bias correction by increasing the order of the local polynomial estimators. \footnote{See Online Appendix A.3 
for details.}
\end{remark}

\begin{remark}[Bandwidth selection]
We allow the bandwidths to differ for each estimator as long as the conditions stated in Theorem \ref{thm:1} are satisfied and the marginal bias and variance expressions are identical to those in \cite{calonicoRobustNonparametricConfidence2014a} and \cite{cattaneoSimpleLocalPolynomial2019}. Hence, we may apply the original data-driven bandwidth selection provided in \cite{calonicoRobustNonparametricConfidence2014a} for $h_1,\ldots,h_d$ and \cite{cattaneoSimpleLocalPolynomial2019} for $h_f$ to our joint distributional approximation. As described above, we used bias-corrected local polynomial estimators. Specifically, $h_k$ is the MSE-optimal bandwidth for $\hat\tau_{Z_k,l}(h_k)$, and $h_f$ is the MSE-optimal bandwidth for $\hat\tau_{f,p}(h_f)$. We construct the test statistics based on $\hat\tau_{Z_k,l+1}(h_k)$ and $\hat\tau_{f,p+1}(h_f)$. The latter implementation via $\hat\tau_{f,p+1}(h_f)$ is the default for \textit{rddensity}, based on \cite{cattaneoSimpleLocalPolynomial2019}.
\end{remark}

\begin{remark}[Asymptotic covariance matrix]
We estimate the asymptotic covariance matrix as follows.
\[
\hat V_{l,p}(h) = 
\begin{pmatrix}
    \hat V_{Z,l}(h) & 0\\
    0& \hat V_{f,p}(h_f)
\end{pmatrix},
\]
where
\[
\hat V_{Z,l}(h)\pto V_{Z,l},\quad \hat V_{f,p}(h_f)\pto V_{f,p}.
\]
We employed the estimators proposed by
\cite{calonicoRobustNonparametricConfidence2014a} and \cite{cattaneoSimpleLocalPolynomial2019} for the asymptotic variances of $\hat\tau_{Z_k,l}(h_k)$ and $\hat\tau_{f,p}(h_f)$ for each diagonal element of the asymptotic covariance matrix $\hat V_{Z,l}(h)$ and $\hat V_{f,p}(h_f)$. 
For the covariance estimators of $\hat\tau_{Z_j,l}(h_j)$ and $\hat\tau_{Z_k,l}(h_k)$, we follow \cite{calonicoRobustNonparametricConfidence2014a}. Specifically, we propose to estimate the elements of $\hat V_{Z,l}(h)$ based on the nearest-neighbor estimation, which may be more robust than plugging in the corresponding residuals of nonparametric regressions in finite samples.
We provide the exact forms of $\hat V_{Z,l}(h)$ and $\hat V_{f,p}(h_f)$ in Online Appendix  A.4.
\end{remark}

\subsection{Proposed unified tests}\label{sec:4}
Given the joint asymptotic normality result, we introduce a unified test for the plausibility of identification. We constructed the test statistics based on the bias-corrected local polynomial estimators $\hat\tau_{Z_k,l+1}(h_k)$ and $\hat\tau_{f,p+1}(h_f)$, where the bandwidth was chosen as MSE-optimal for $\hat\tau_{Z_k,l}(h_k)$ and $\hat\tau_{f,p}(h_f)$. Our unified test has two forms: one aggregates the vector of the density and balance test statistics with the $L^2$-norm, and the other aggregates them via the $L^\infty$-norm. 

The Wald (chi-squared) test statistic with the $L^2$-norm has the following form.
\begin{equation*}
\hat \chi^2_{l+1,p+1}(h) = ||(\hat V_{l+1,p+1}(h))^{-1/2}\hat T_{l+1,p+1}(h)||^2_2,
\end{equation*}
where
\begin{equation*}
\hat T_{l,p}(h)=
\sqrt{n} [\sqrt{h_{1}}\hat \tau_{Z_1,l}(h_1),\ldots,\sqrt{h_{d}}\hat \tau_{Z_d,l}(h_d),\sqrt{h_{f}}\hat\tau_{f,p}(h_f)]'.
\end{equation*}
This conventional test statistic suffers from severe size distortions in finite samples. In particular, the inverse of the asymptotic covariance matrix estimator becomes unstable for a sufficiently large dimension $d$. 
 
Instead, we propose a \textit{sWald} test statistic that modifies the Wald test statistic. Specifically, we standardize only the diagonal elements instead of multiplying the entire matrix by the inverse of the asymptotic covariance matrix. The sWald test statistic has the following form:
\begin{equation*}
\tilde \chi^2_{l+1,p+1}(h) = ||\tilde T_{l+1,p+1}(h)||^2_2,
\end{equation*}
where
\begin{equation*}
\tilde T_{l,p}(h) = \diag\{(\hat V_{Z_1,l}(h_1))^{-1/2},\ldots,(\hat V_{Z_d,l}(h_d))^{-1/2},(\hat V_{f,p}(h_f))^{-1/2}\} \hat T_{l,p}(h),
\end{equation*}
and $\diag(a_1,\ldots,a_n)$ is the $(n\times n)$ diagonal matrix with diagonal elements $a_1,\ldots,a_n$.


We obtain the asymptotic distribution of $\tilde \chi^2_{l+1,p+1}(h)$ from the joint asymptotic distribution for the local polynomial estimators, with their diagonal elements standardized by their standard errors $(\hat V_{f,p+1}(h_f))^{1/2}$ and $(\hat V_{Z_k,l+1}(h_k))^{1/2}$. We construct the standard errors of $\tau_{Z_k}$ as $\hat V_{Z_k,l+1}(h_k)\pto V_{Z_k,l+1}$, where $V_{Z_k,l+1}$ is the asymptotic variance of $\sqrt{n h_{k}}(\hat\tau_{Z_k,l+1}(h_k)-\tau_{Z_k})$.
The asymptotic covariance matrix of the standardized local polynomial estimators is the asymptotic correlation matrix of the local polynomial estimators for $\tau_f$ and $\tau_{Z_k}$. The correlation matrix $V^*_{l+1,p+1}$ is
\[
V^*_{l+1,p+1} = 
\begin{pmatrix}
    V^*_{Z,l+1} & 0\\
    0& 1
\end{pmatrix},
\]
where $V^*_{Z,l}=\diag(V_{Z_1,l}^{-1/2},\ldots,V_{Z_d,l}^{-1/2})V_{Z,l}\diag(V_{Z_1,l}^{-1/2},\ldots,V_{Z_d,l}^{-1/2})$; our estimator of $V^*_{l+1,p+1}$ is
\[
\hat V^*_{l+1,p+1}(h) = 
\begin{pmatrix}
    \hat V^*_{Z,l+1}(h) & 0\\
    0& 1
\end{pmatrix},
\]
where
\begin{align*}
 & \hat V^*_{Z,l}(h) = D_{\hat{V}}(h) \hat V_{Z,l}(h) D_{\hat{V}}(h)\mbox{ with } D_{\hat{V}}(h) = \diag\{(\hat V_{Z_1,l}(h_1))^{-1/2},\ldots,(\hat V_{Z_d,l}(h_d))^{-1/2}\}.
\end{align*}

In the following proposition, we show that the finite-sample distribution of $\tilde \chi^2_{l+1,p+1}(h)$ is approximated by the distribution of the sum of the squares of the $(d+1)$-dimensional multivariate normal random vector.
\begin{proposition}\label{prop:1.5}
Suppose Assumptions \ref{asp:1} and \ref{asp:2} hold with $R \geq p+2$ and $S\geq l+3$, and the asymptotic covariance matrix $V_{l+1,p+1}$ is nonsingular. If $n h_{min}\to\infty$, $nh_f^2 \to\infty$, $n (h_{max})^{2l+5}\to0$, and $n (h_f)^{2p+3} \to0$, then
\begin{eqnarray*}
\text{Under }H_0 &:& \lim_{n\to\infty} 
P(\tilde \chi^2_{l+1,p+1}(h)\geq \hat c_{l+1,p+1}(\alpha)) = \alpha, \\
\text{Under }H_1 &:& \lim_{n\to\infty} 
P(\tilde \chi^2_{l+1,p+1}(h)\geq \hat c_{l+1,p+1}(\alpha)) = 1,
\end{eqnarray*}
where $\hat c_{l+1,p+1}(\alpha)$ is the $(1 - \alpha)$-quantile of the distribution of $||N_{d+1}(0,\hat V^*_{l+1,p+1}(h))||^2_2$.
\end{proposition}

Proposition \ref{prop:1.5} establishes the asymptotic validity and consistency of the $\alpha$-level testing procedure, which rejects $H_0$ if $\hat c_{l+1,p+1}(\alpha)$ exceeds $\tilde \chi^2_{l+1,p+1}(h)$.
The critical value $\hat c_{l+1,p+1}(\alpha)$ is obtained numerically by generating $(d+1)$-dimensional standard normal random vectors using Monte Carlo simulation.

The \textit{Max} test statistic with $L^{\infty}$ norm takes the following form.
\begin{equation*}
\hat M_{l+1,p+1}(h) = ||\tilde T_{l+1,p+1}(h)||^2_\infty.
\end{equation*}
Similar to the sWald statistic, the Max statistic is constructed by standardizing diagonal elements. The following statement shows that the finite sample distribution of $\hat M_{l+1,p+1}(h)$ is approximated by the distribution of the maximum value of the $(d+1)$-dimensional multivariate normal random vector. \footnote{This construction via simulation is similar to the plug-in sup-t implementation in \cite{oleaSimultaneousConfidenceBands2019}.}
\begin{proposition}\label{prop:2}
Suppose Assumptions \ref{asp:1} and \ref{asp:2} hold with $R \geq p+2$ and $S\geq l+3$, and the asymptotic covariance matrix $V_{l+1,p+1}$ is nonsingular. If $n h_{min}\to\infty$, $nh_f^2 \to\infty$, $n (h_{max})^{2l+5}\to0$ and $n (h_f)^{2p+3} \to0$, 
\begin{eqnarray*}
\text{Under }H_0 &:& \lim_{n\to\infty} 
P(\hat M_{l+1,p+1}(h)\geq \hat m_{l+1,p+1}(\alpha)) = \alpha, \\
\text{Under }H_1 &:& \lim_{n\to\infty} 
P(\hat M_{l+1,p+1}(h)\geq \hat m_{l+1,p+1}(\alpha)) = 1,
\end{eqnarray*}
where $\hat m_{l+1,p+1}(\alpha)$ is the $(1 - \alpha)$-quantile of the distribution $||N_{d+1}(0,\hat V^*_{l+1,p+1}(h))||^2_\infty$.
\end{proposition}

Proposition \ref{prop:2} establishes the asymptotic validity and consistency of the $\alpha$-level testing procedure, which rejects $H_0$ if $\hat M_{l+1,p+1}(h)$ exceeds $\hat m_{l+1,p+1}(\alpha)$. The critical value $\hat m_{l+1,p+1}(\alpha)$ is obtained numerically by generating $(d+1)$-dimensional standard normal random vectors using Monte Carlo simulation.

\begin{remark}[Max test against Bonferroni correction]
The Max test is more powerful than the Bonferroni correction. The Bonferroni correction conducts asymptotic $t$ tests separately with an altered size of $\alpha/(d+1)$. The Bonferroni correction rejects $H_0$ if
\[
\hat M_{l+1,p+1}(h) \geq z^2(\alpha/2(d+1)),
\]
where $z(\alpha/2(d+1))$ is the $\left(1-\frac{\alpha}{2(d+1)}\right)$-quantile of $N(0,1)$, suggesting that the same $\hat M_{l+1,p+1}(h)$ is used for the Bonferroni correction and Max test. For any correlation matrix $V$, we obtain
\begin{eqnarray*}
P\left( \| N_{d+1}(0,V) \|_{\infty}^2 \geq z^2(\alpha/2(d+1)) \right) \leq \alpha.
\end{eqnarray*}
Furthermore, $z^2(\alpha/2(d+1))$ must be greater than $\hat m_{l+1,p+1}(\alpha)$
because $\hat m_{l+1,p+1}(\alpha)$ is the $(1 - \alpha)$-quantile of distribution $||N_{d+1}(0,\hat V^*_{l+1,p+1}(h))||^2_\infty$. Therefore, the rejection probability of the Max test is always higher than that of the Bonferroni correction.
\end{remark}

\begin{remark}
 In general, neither Max test nor sWald test might not dominate the other in their power properties. On the one hand, the sWald (and Wald) test has the rejection region which is more sensitive to multiple small jumps than that of the Max test. On the other hand, the rejection region of the Max test is more sensitive to a single large jump than that of the sWald test. Since both cases are sensible in assessing the plausibility of the RD design, it may be difficult to conclude which tests are superior. In the simulation in Section \ref{sec:5}, we confirmed this feature in the power analysis.
\end{remark}

\begin{remark}
    Multiple testing problem for multiple outcomes is also a serious concern. Our \textit{rdtest} package handles a special case of the test for the balance tests only, and the proposed tests are readily applicable to a joint testing for multiple outcomes.
\end{remark}

\section{Simulation evidences for the unified tests}\label{sec:5}

Given the established theoretical properties of the proposed tests, their performance was evaluated using the Monte Carlo experiment. We conducted $3000$ replications to generate a random sample 
$\{(X_i, Z_{1,i}, \ldots, Z_{d,i})':i=1,\ldots,n\}$ with size $n= \{500, 1000\}$ for each of them. \footnote{The detailed simulation data generating process is in Online Appendix C.}

We specify the distribution of running variable $X_i$ as the weighted average of two truncated normal distributions, where $\bar{p} \geq 0.5$ determines the weights. The density of $X_i$ is shown in Figure \ref{fig:fx}. The density $f(x)$ is continuous at $x=0$ if and only if $\bar{p} = 0.5$ but jumps at $x=0$ when $\bar{p} > 0.5$.
\begin{figure}[H]
\centering
\scalebox{1.4}{
\begin{tikzpicture}
 \coordinate[label=below left:\scriptsize{O}] (O) at (2,0); 
 \coordinate (XS) at (-0.5,0); 
 \coordinate (XL) at (4.5,0); 
 \coordinate (YS) at (2,-0.5); 
 \coordinate (YL) at (2,4); 
 \draw[->,>=stealth] (XS)--(XL) node[right] {\scriptsize{$x$}}; 
 \draw[name path=y,->,>=stealth] (YS)--(YL) 
 node[right] {\scriptsize{$y$}}; 
 \draw[green,semithick,samples=100,domain=-0.5:4.5] 
 plot(\x,{exp(-pow((\x-2)/10,2)/(2*pow(0.12,2)))/sqrt(2*pow(0.12,2)*pi)});
 \draw[name path=fp,red,semithick,samples=100,domain=2:4.5] 
 plot(\x,{exp(-pow((\x-2)/10,2)/(2*pow(0.12,2)))/sqrt(2*pow(0.12,2)*pi)*(1-0.48)/0.5});
 \draw[name path=fm,red,semithick,samples=100,domain=-0.5:2] 
 plot(\x,{exp(-pow((\x-2)/10,2)/(2*pow(0.12,2)))/sqrt(2*pow(0.12,2)*pi)*0.48/0.5});
 
 \draw[] (1,2.5) node[left] {\scriptsize{$f(x)$}};
 \path[name intersections={of= y and fp, by={FP}}];
 \path[name intersections={of= y and fm, by={FM}}];
 \draw[<-,>=stealth] (FP) arc (130:90:1.15); 
 \draw (FP)++(130:-1.15)++(90:1.15)
 node[right]{\scriptsize{$f_+$}};
 \draw[<-,>=stealth] (FM) arc (60:20:1.15); 
 \draw (FM)++(67:-1.15)++(27:1.15)
 node[below]{\scriptsize{$f_-$}};
 \draw[green,semithick] (4,3.5)--(4.75,3.5) node[black,right] {\scriptsize{$\bar{p}=0.5$}};
 \draw[red,semithick] (4,3)--(4.75,3) node[black,right] {\scriptsize{$\bar{p}>0.5$ }};
 \draw[] (3.9,2.75) rectangle (6.2,3.75);
\end{tikzpicture}
}
 \caption{Graph of $f(x)$}\label{fig:fx}
\end{figure}

For covariates $Z_i$, we consider two specifications: first, one of the $d$ covariates sees a jump; second, all $d$ covariates see a jump, with each jump size divided by $d$. For the first specification, we have $\mu_{Z_k}(x)=E[Z_{k,i}|X_i=x], k = 1\ldots, d$ as follows.
\[
\mu_{Z_k}(x) = \lambda(x) \text{ for }k=1,\ldots,d-1\quad \text{and} \quad
\mu_{Z_d}(x) =
\begin{cases}
\lambda(x)  & \text{if }x<0\\
\lambda(x) + a & \text{if }x\geq 0,
\end{cases}
\]
where only the $d$th covariate sees a jump, and the functional form of $\lambda(x)$ is obtained from the simulated data close to the data in \cite{leeRandomizedExperimentsNonrandom2008} following \cite{calonicoRobustNonparametricConfidence2014a}, where $a\geq0$ is the level of discontinuity of $\mu_{Z_d}(x)$ at $x=0$. For the second specification, we have $\mu_{Z_k}(x)=E[Z_{k,i}|X_i=x], k=1,\ldots,d$ as follows.
\[
\mu_{Z_k}(x) =
\begin{cases}
\lambda(x)  & \text{if }x<0\\
\lambda(x) + a/d & \text{if }x\geq 0,
\end{cases}
\]
where $a/d$ denotes the discontinuity level for each $\mu_{Z_k}(x)$ at $x=0$,
For both cases, let $1>\rho\geq0$ denote the conditional correlation between $Z_{j,i}$ and $Z_{k,i}$, given $X_i=x$. In this data-generating process, the null hypothesis for the joint manipulation test, $H_0: (\tau_f,\tau_{Z_1},\ldots,\tau_{Z_d})=0$ is true if and only if $\bar{p} = 0.5$ and $a = 0$; $\tau_f > 0$ when $\bar{p} > 0.5$; $\tau_{Z_d}=a>0$ when $a>0$.

We conducted the tests with size $\alpha=0.05$, employing local polynomial estimators $(\hat{\tau}_{Z_{1},2}(h_1),….,\hat{\tau}_{Z_{d},2}(h_d), \hat{\tau}_{f,3}(h_f))$ described in Section \ref{sec:3} with $K(u)=(1-|u|)\vee0$, where the bandwidths are chosen to be MSE-optimal for $\hat\tau_{Z_k,1}(h_k)$ and $\hat\tau_{f,2}(h_f)$.
The selection of the kernel function and polynomial order follows \cite{calonicoRobustNonparametricConfidence2014a} and \cite{cattaneoSimpleLocalPolynomial2019}. \footnote{As mentioned in the previous sections, the polynomial orders are chosen so that the estimates are bias-corrected for their valid inference. This choice of polynomial order is the default for \textit{rddensity} and available as $\rho = 1$ option for \textit{rdrobust.}}
Below, we report the empirical sizes and powers of the conventional and proposed testing methods.

We compared five testing methods: (i) naive testing; (ii) Bonferroni correction of naive testing; (iii) Wald (chi-square) test; (iv) Max test proposed in Section \ref{sec:4}; (v) sWald test proposed in Section \ref{sec:4}. Naive testing separately conducts $d+1$ asymptotic $t$ tests, as described in \cite{calonicoRobustNonparametricConfidence2014a} and \cite{cattaneoSimpleLocalPolynomial2019} with a size $\alpha$ for each of the $d+1$ restrictions of the null hypothesis, $\tau_f=0, \tau_{Z_1}=0,\ldots,\tau_{Z_d}=0$. Naive testing rejects the null hypothesis if at least one test rejects it. The Bonferroni correction conducts asymptotic $t$ tests separately, but the size is changed to $\alpha/(d+1)$ to restrict type 1 error of the null hypothesis that all $d+1$ restrictions hold. \footnote{We do not benefit from rejecting more than one test for evaluating the underlying identifying restriction. Thus, Holm correction for Bonferroni test does not improve the following results because the probability of rejecting at least one null hypothesis remains the same in two corrections.}

In Tables \ref{tab:5.1}–\ref{tab:5.3}, we compare the empirical size of each test under the null hypothesis (where $\bar{p} = 0.5$ and $a = 0$) with different correlation coefficients $\rho=0,0.5,0.9$. Notably, naive testing suffers from size distortion, and over-rejection worsens when $d$ is large. Similarly, the Wald test without correction suffers from size distortion as $d$ increases. As discussed in Section \ref{sec:4}, the effective sample size may be too small compared to a large $d$.
When $d$ is large, the covariance matrix increases, and inverse matrix estimation becomes unstable when the effective sample size is limited.

\begin{table}[H]
\caption{Empirical size of tests (i)--(v) ($\rho = 0$)}\label{tab:5.1}
\centering
\scalebox{0.9}{
\begin{tabular}{c|c c c c c|c c c c c} 
\multicolumn{1}{c|}{} & \multicolumn{5}{|c|}{n: 500} & \multicolumn{5}{|c}{n: 1000} \\ 
dim & {naive} & bonfe & Wald & {Max} & sWald & {naive} & bonfe & Wald & {Max} & sWald\\ \hline
1 & 0.095 & 0.050 & 0.047 & 0.060 & 0.051 & 0.081 & 0.035 & 0.039 & 0.048 & 0.046 \\ 3 & 0.185 & 0.056 & 0.068 & 0.064 & 0.049 & 0.173 & 0.049 & 0.054 & 0.054 & 0.043 \\ 5 & 0.289 & 0.070 & 0.112 & 0.077 & 0.055 & 0.271 & 0.052 & 0.073 & 0.057 & 0.044 \\ 10 & 0.464 & 0.086 & 0.271 & 0.088 & 0.055 & 0.425 & 0.059 & 0.132 & 0.060 & 0.041 \\ 25 & 0.787 & 0.114 & 0.896 & 0.116 & 0.034 & 0.748 & 0.074 & 0.517 & 0.071 & 0.032
\end{tabular}
}

\vspace{0.5cm}
\caption{Empirical size of tests (i)--(v) ($\rho = 0.5$)}\label{tab:5.2}
\centering
\scalebox{0.9}{
\begin{tabular}{c|c c c c c|c c c c c}
\multicolumn{1}{c|}{} & \multicolumn{5}{|c|}{n: 500} & \multicolumn{5}{|c}{n: 1000} \\ 
dim & {naive} & bonfe & Wald & {Max}  & sWald & {naive} & bonfe & Wald & {Max}  & sWald\\ \hline
3 & 0.177 & 0.053 & 0.072 & 0.061 & 0.060 & 0.167 & 0.045 & 0.054 & 0.051 & 0.050 \\5 & 0.252 & 0.058 & 0.107 & 0.071 & 0.060 & 0.234 & 0.051 & 0.072 & 0.060 & 0.055 \\10 & 0.362 & 0.076 & 0.248 & 0.092 & 0.068 & 0.327 & 0.051 & 0.119 & 0.062 & 0.054 \\25 & 0.545 & 0.076 & 0.862 & 0.097 & 0.060 & 0.500 & 0.058 & 0.468 & 0.073 & 0.051
\end{tabular}
}

\vspace{0.5cm}
\caption{Empirical size of tests (i)-(v) ($\rho =0.9$)}\label{tab:5.3}
\centering
\scalebox{0.9}{
\begin{tabular}{c|c c c c c|c c c c c}
\multicolumn{1}{c|}{} & \multicolumn{5}{|c|}{n: 500} & \multicolumn{5}{|c}{n: 1000} \\ 
dim & {naive} & bonfe & Wald & {Max}  & sWald & {naive} & bonfe & Wald & {Max}  & sWald\\ \hline
3 & 0.132 & 0.042 & 0.066 & 0.055 & 0.060 & 0.124 & 0.033 & 0.051 & 0.049 & 0.053 \\5 & 0.150 & 0.036 & 0.090 & 0.059 & 0.048 & 0.149 & 0.035 & 0.068 & 0.054 & 0.050 \\10 & 0.186 & 0.037 & 0.216 & 0.072 & 0.070 & 0.169 & 0.023 & 0.104 & 0.050 & 0.050 \\25 & 0.226 & 0.024 & 0.820 & 0.060 & 0.056 & 0.204 & 0.018 & 0.406 & 0.047 & 0.047
\end{tabular}
}

\vspace{0.5cm}
\begin{minipage}{380pt}
{\fontsize{10pt}{10pt}\selectfont\smallskip\textit{Notes}: (i) ``naive'': empirical size of naive testing; (ii) ``bonfe'': empirical size of Bonferroni correction; (iii) ``Max'': empirical size of the Max test; (iv) ``Wald'': empirical size of the Wald (chi-square) test; (v) ``sWald'': standardized Wald (chi-square) test; ``dim'': dimension of the pre-treatment covariates $Z_i=(Z_{1,i},\ldots,Z_{d,i})'$; Columns under ``n: 500'' and ``n: 1000'' report the results obtained with size $n=500$ and $n=1000$, respectively.}
\end{minipage}
\end{table}

The performances of the other methods vary according to sample size,  dimension of covariates, and correlation coefficients. For the relatively small sample size of $n = 500$, the superiority of these methods is ambiguous. For a relatively small number of covariates $dim \in \{1,3,5\}$, the sWald test controls for size by $6\%$, whereas the Max test and Bonferroni correction control for size by approximately $7\%$. For a larger number of covariates, the Bonferroni correction and Max test may fail to control size by $10\%$, whereas the sWald test controls size by $7\%$ throughout. For a larger sample size of $n = 1000$, the sWald test dominates the Bonferroni correction for size control. The sWald test controls size by $5.5\%$ for any correlation and number of covariates up to $25$. \footnote{We consider that the failure in the Bonferroni correction comes from challenges in evaluating the tail of the nonparametric test statistics.} For a moderate number of covariates up to $dim = 5$ or $10$, the Max test can control the size by $6\%$ and is not substantially worse than the Bonferroni correction. Therefore, when the sample size is not small, the sWald test can control the size for any dimension up to dim = $25$; the Max test and Bonferroni correction can control the size up to dim = $5$ or $10$.

Given the results of size control, we evaluate their empirical power properties. We compare the empirical power of the (ii) Bonferroni correction, (iv) Max test, and (v) sWald test under $\bar{p}=0.575$ and $\tau_{Z_d}(=a)=0,0.5,1,1.5,2$ for different dimensions $d=1,3,5$.
Figure \ref{fig:5.1} presents the results for the scalar covariate. Figures \ref{fig:5.3.last}, and \ref{fig:5.5.last} present the results when one of the $d$ covariates sees a jump. Further, Figures \ref{fig:5.3.all} and \ref{fig:5.5.all} show the results when all $d$ covariates see a jump. However, each jump size is divided by $d$ for $d=3,5$.
In Figures \ref{fig:5.3.last}–\ref{fig:5.5.all}, each figure presents the results for different correlations $\rho = 0.5, 0.9$.\footnote{In Online Appendix D
, we present the results for weaker correlations and results containing negative correlations.}

In all cases, the Max test exhibits higher power than the Bonferroni correction. When one of the covariates jumps (Figures \ref{fig:5.3.last} and \ref{fig:5.5.last}), the Max test outperforms the Bonferroni correction and sWald test, and the gap among the three methods expands as the correlation and dimension increase. While the sWald test has relatively less power to detect small jumps in the covariate, the power of the sWald test catches up quickly as the jump size increases. 

When all covariates have seen jumps (see Figure \ref{fig:5.3.all} and \ref{fig:5.5.all}), the sWald test outperforms the Max test and Bonferroni correction in most cases. When the correlation is weak, the sWald test is substantially superior to the other two methods; however, the Max test approaches the sWald test as the correlation and dimensions increase. In Online Appendix D
, we report qualitatively similar simulation results for other specifications, including lower or negative correlation coefficients.

To summarize, we recommend the Max test when one suspects that a few covariates exhibit jumps but not the others among a moderate number of covariates. We recommend the sWald test for other cases because it has superior size control and power properties for detecting jumps among many covariates.

\begin{figure}[H]
 \centering
  \includegraphics[width=\hsize]{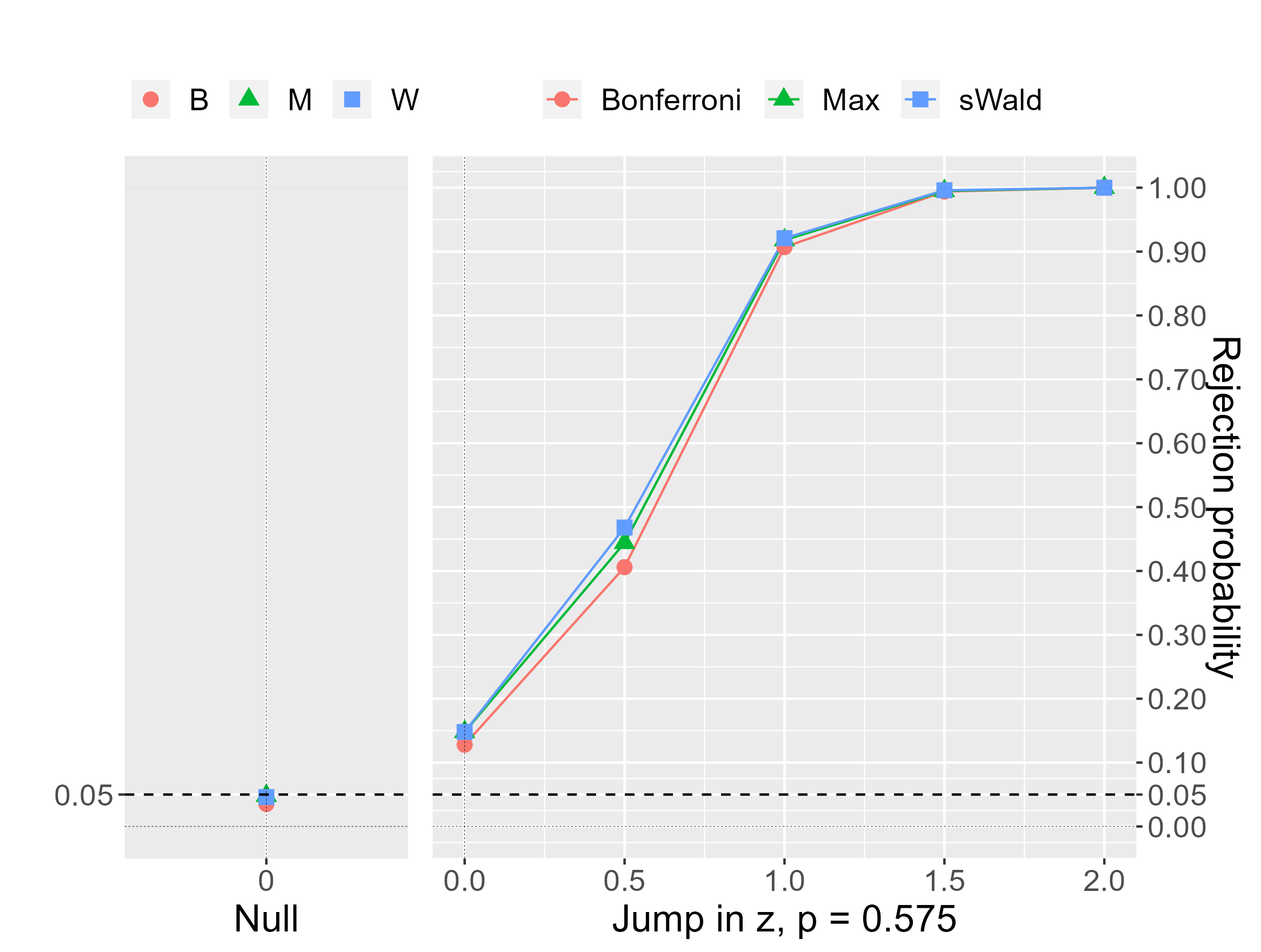}
  \caption{A jump in a scalar $Z$ and a density test, $n = 1000$}\label{fig:5.1}
  \begin{minipage}{380pt}
{\fontsize{10pt}{10pt}\selectfont\smallskip\textit{Notes}: 
Jump in $z$ is relative to its standard deviation of the error term. The density has the discontinuity with the density ratio of $0.575/0.425 = 1.353$.
}
\end{minipage}
\end{figure}

\begin{figure}[H]
 \centering
 \begin{minipage}[b]{0.8\hsize}
   \includegraphics[width=\hsize]{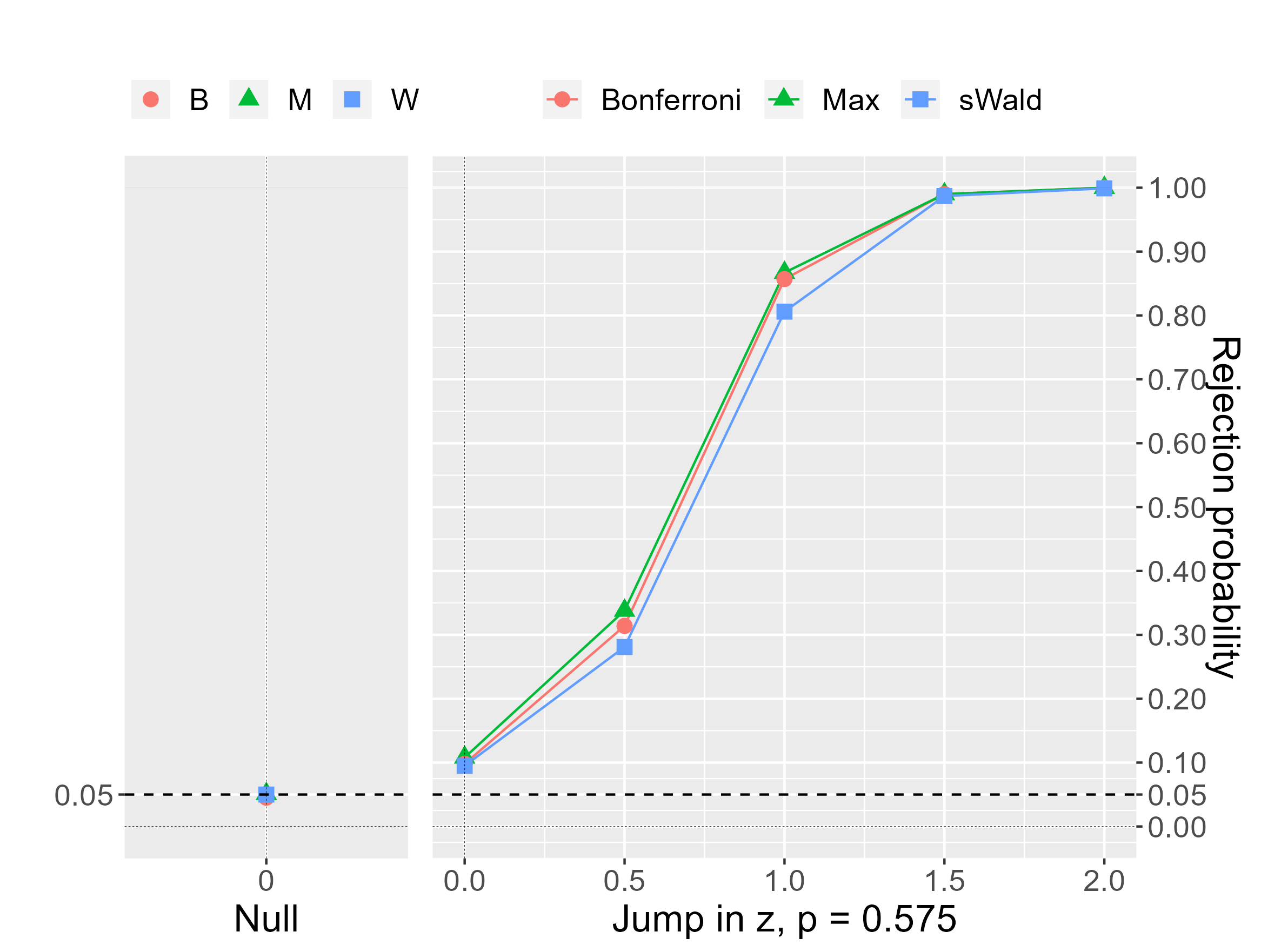}
   \subcaption{Covariates have the same pairwise correlation coefficient of $0.5$}
 \end{minipage} \\
 \begin{minipage}[b]{0.8\hsize}
   \includegraphics[width=\hsize]{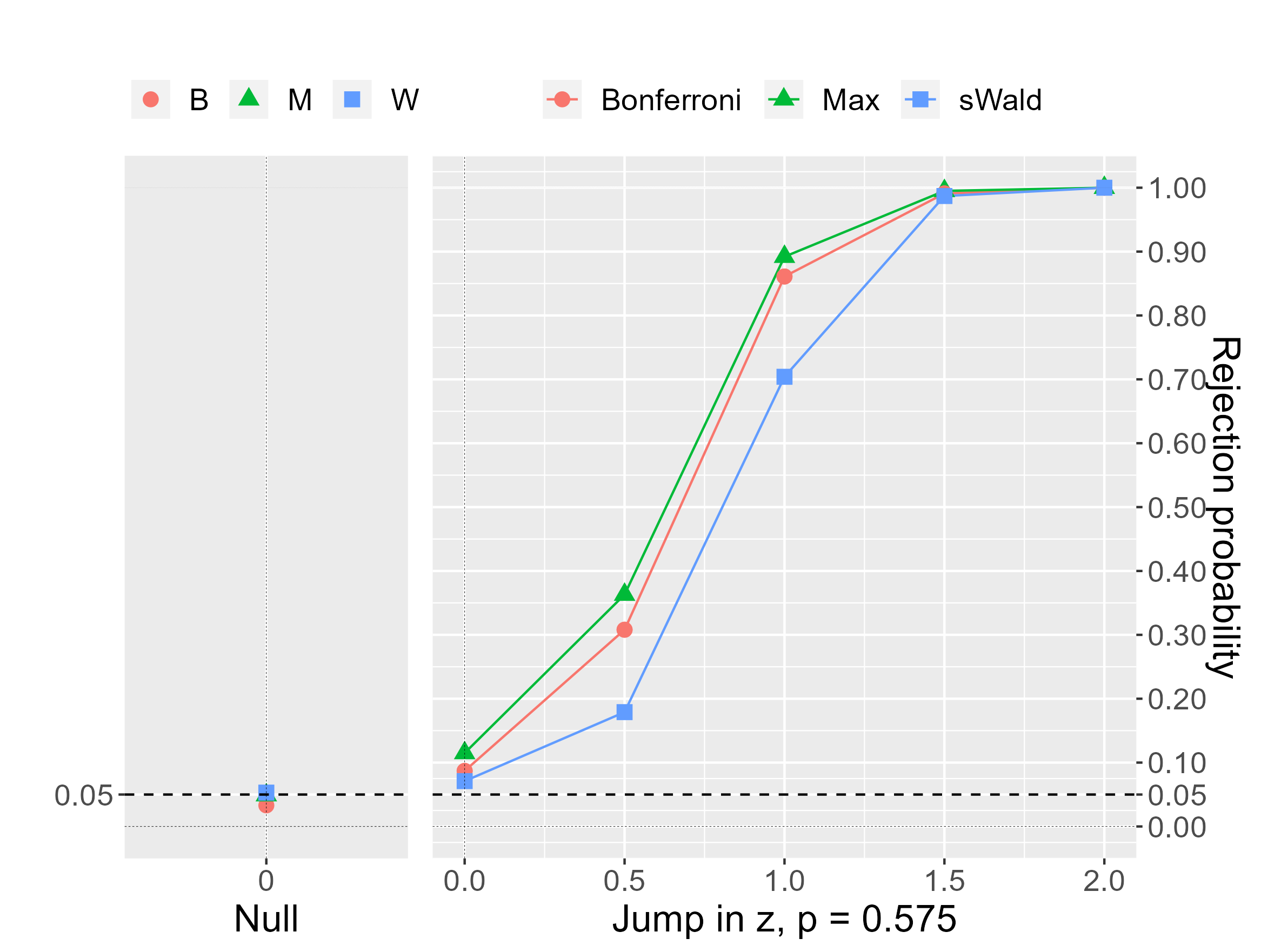}
  \subcaption{Covariates have the same pairwise correlation coefficient of $0.9$}
 \end{minipage} 
  \caption{A jump in the three covariates with a density, $n = 1000$}\label{fig:5.3.last} 
\begin{minipage}{380pt}
{\fontsize{10pt}{10pt}\selectfont\smallskip\textit{Notes}: 
Jump in $z$ is relative to its standard deviation of the error term. The density has the discontinuity with the density ratio of $0.575/0.425 = 1.353$.
}
\end{minipage}
\end{figure}

\begin{figure}[H]
 \centering
 \begin{minipage}[b]{0.8\hsize}
   \includegraphics[width=\hsize]{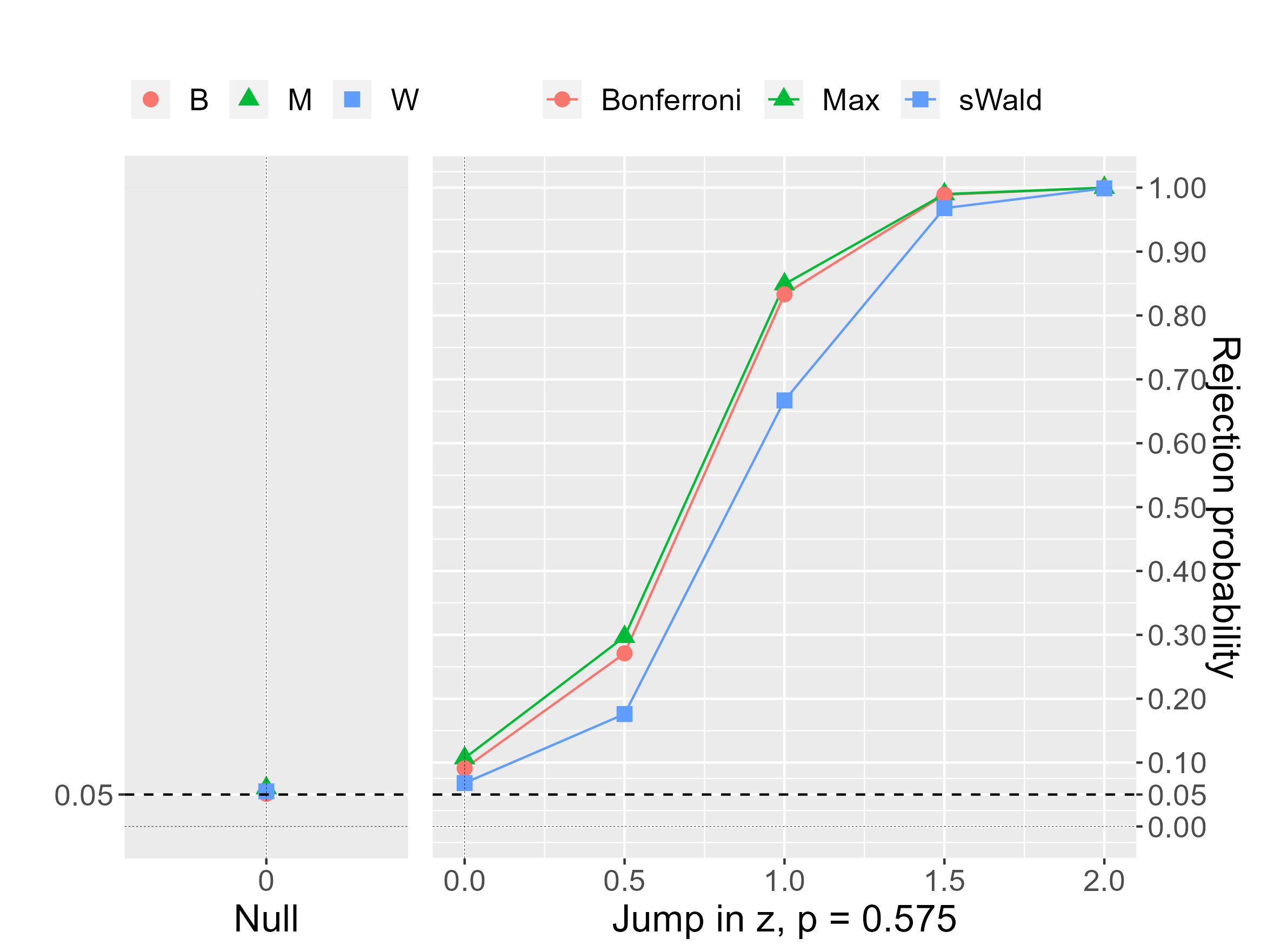}
   \subcaption{Covariates have the same pairwise correlation coefficient of $0.5$}
 \end{minipage} \\
 \begin{minipage}[b]{0.8\hsize}
   \includegraphics[width=\hsize]{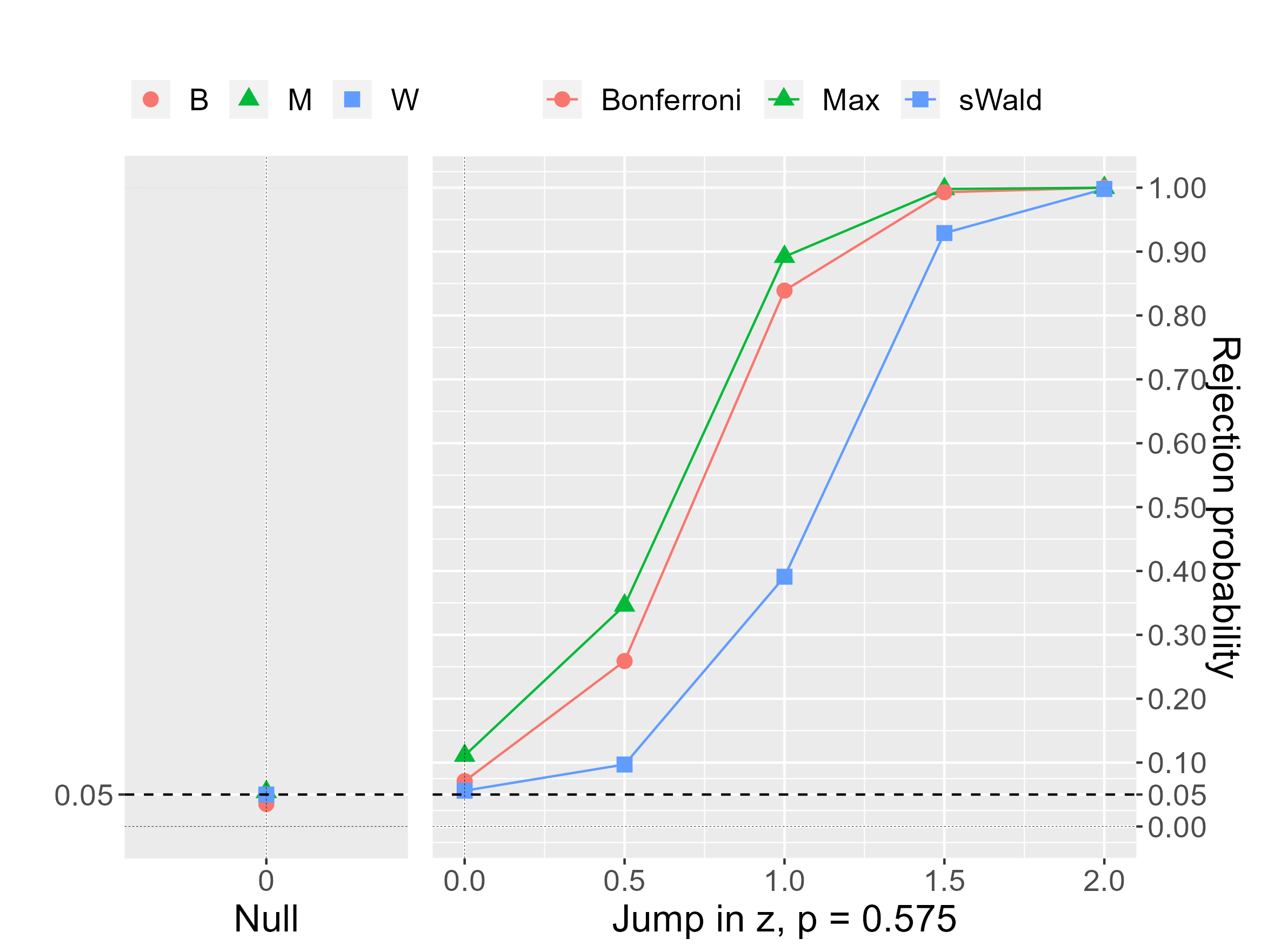}
  \subcaption{Covariates have the same pairwise correlation coefficient of $0.9$}
 \end{minipage} 
  \caption{A jump in the five covariates with a density, $n = 1000$}\label{fig:5.5.last} 
\begin{minipage}{380pt}
{\fontsize{10pt}{10pt}\selectfont\smallskip\textit{Notes}: 
Jump in $z$ is relative to its standard deviation of the error term. The density has the discontinuity with the density ratio of $0.575/0.425 = 1.353$.
}
\end{minipage}
\end{figure}

\begin{figure}[H]
 \centering
 \begin{minipage}[b]{0.8\hsize}
   \includegraphics[width=\hsize]{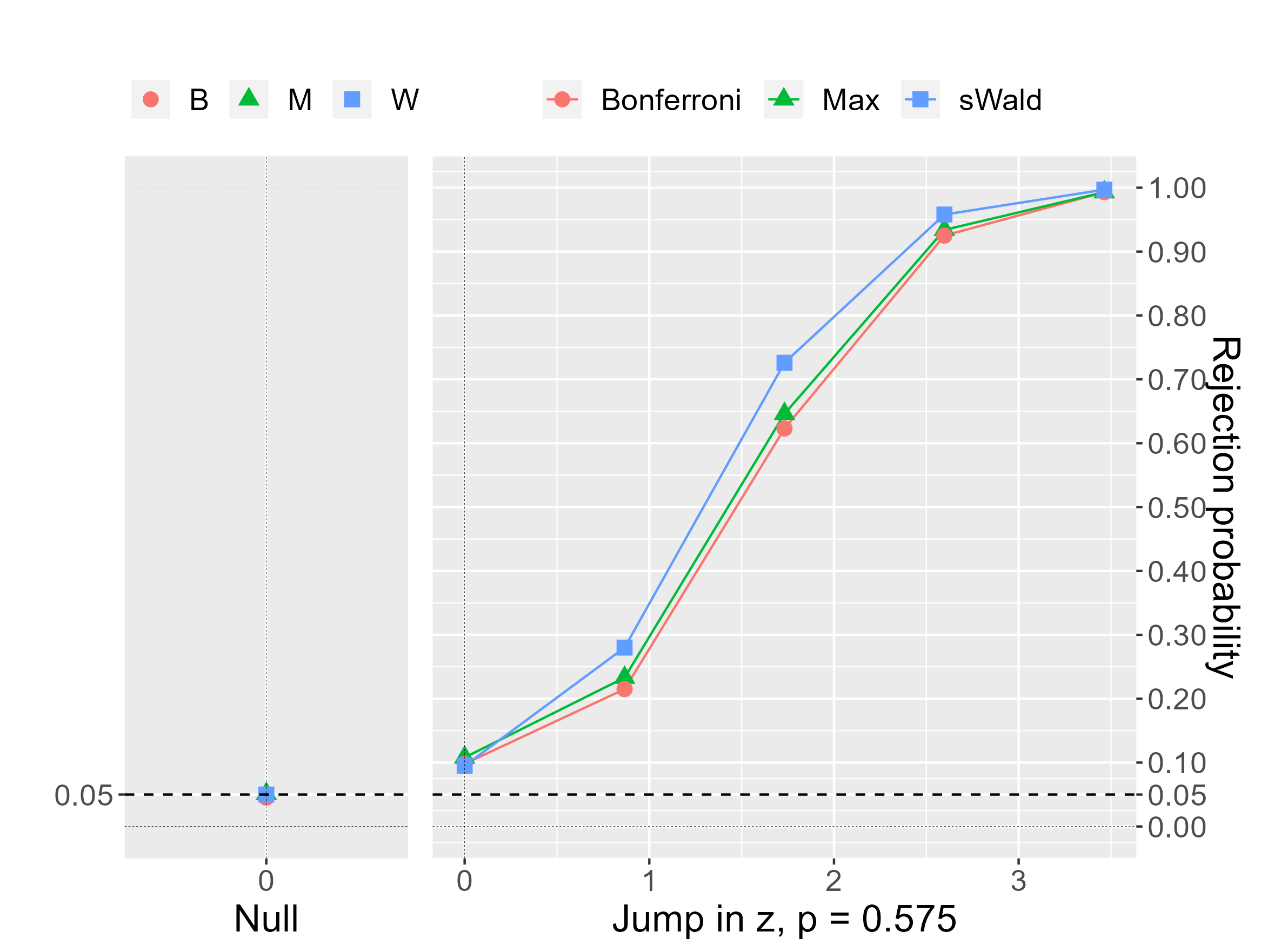}
   \subcaption{Covariates have the same pairwise correlation coefficient of $0.5$}
 \end{minipage} \\
 \begin{minipage}[b]{0.8\hsize}
   \includegraphics[width=\hsize]{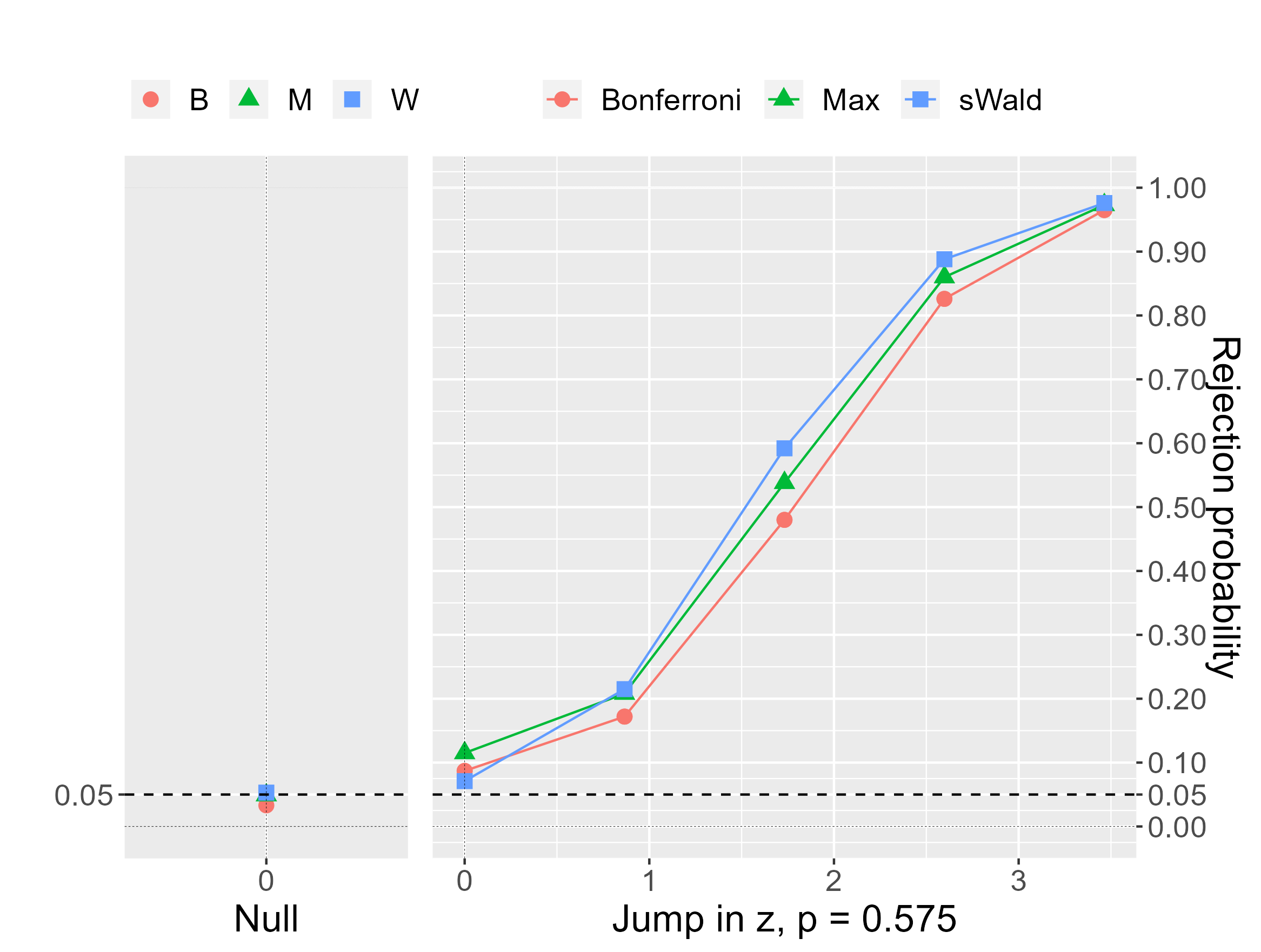}
  \subcaption{Covariates have the same pairwise correlation coefficient of $0.9$}
 \end{minipage} 
  \caption{$1/3$ of jumps in all the three covariates with density, $n = 1000$}\label{fig:5.3.all} 
  \begin{minipage}{380pt}
{\fontsize{10pt}{10pt}\selectfont\smallskip\textit{Notes}: 
Jump in $z$ divided by the number of covariates is relative to its standard deviation of the error term. The density has the discontinuity with the density ratio of $0.575/0.425 = 1.353$.
}
\end{minipage}
\end{figure}

\begin{figure}[H]
 \centering
 \begin{minipage}[b]{0.8\hsize}
   \includegraphics[width=\hsize]{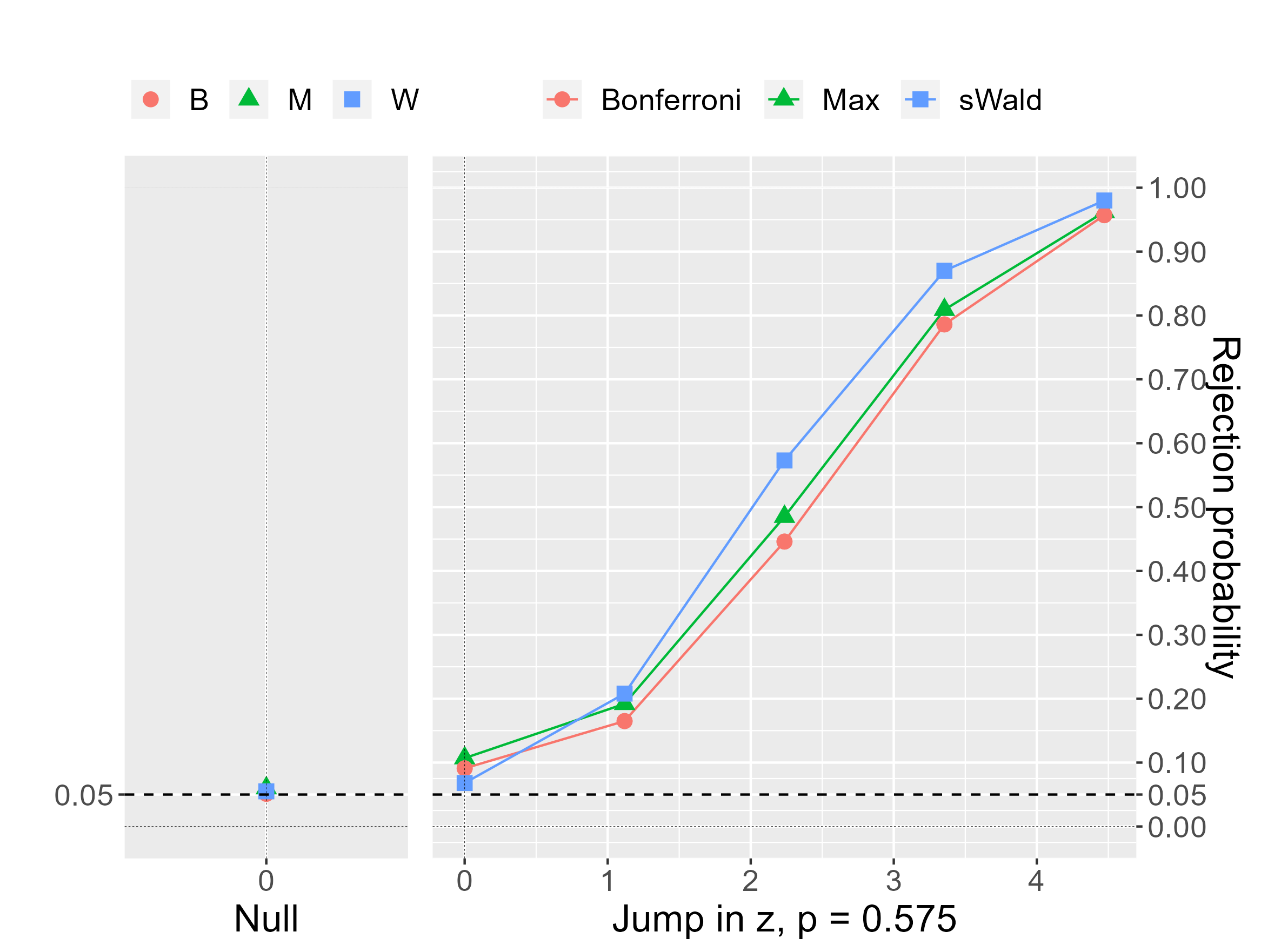}
   \subcaption{Covariates have the same pairwise correlation coefficient of $0.5$}
 \end{minipage} \\
 \begin{minipage}[b]{0.8\hsize}
   \includegraphics[width=\hsize]{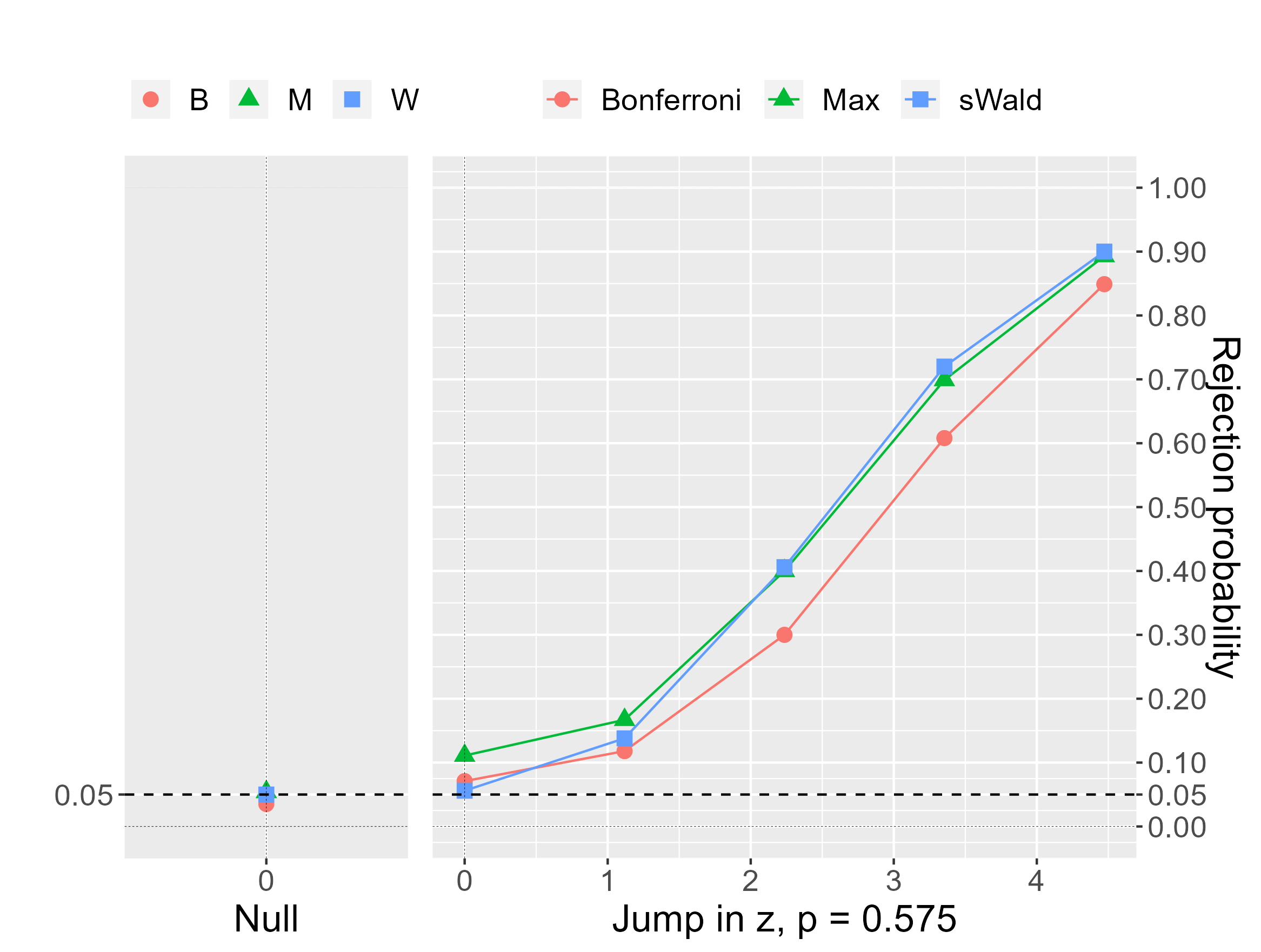}
  \subcaption{Covariates have the same pairwise correlation coefficient of $0.9$}
 \end{minipage} 
  \caption{$1/5$ of jumps in all the five covariates with density, $n = 1000$}\label{fig:5.5.all} 
  \begin{minipage}{380pt}
{\fontsize{10pt}{10pt}\selectfont\smallskip\textit{Notes}: 
Jump in $z$ divided by the number of covariates is relative to its standard deviation of the error term. The density has the discontinuity with the density ratio of $0.575/0.425 = 1.353$.
}
\end{minipage}
\end{figure}

\section{Extension} \label{sec:extension}

Our unified test solves the over-rejection problem prevalent in empirical studies. Below, we consider two extensions of our analysis to supplement our unified test.

\subsection{Equivalence Testing} \label{subsec:equivalence}
The diagnostic test evaluates the null hypothesis that suggests the plausibility of the underlying identification. Hence, the diagnostic procedure tests whether the \textit{desired} null hypothesis can be rejected, and we control the probability of falsely rejecting the \textit{desired} null hypothesis of no imbalance or discontinuity. In other words, the roles of type I and type II errors are flipped.

Several studies propose \textit{equivalence} or \textit{non-inferiority} testing as alternatives to flipped testing for the desirable null hypothesis. See, for example, \cite{wellek2010equivalence} and \cite{hartman2018equivalence} for further details. In RD designs, \cite{hartman2021equivalence} propose equivalence tests for diagnostic testing of each covariate and density. As before, the existing equivalence tests are uni-variate, and a unified alternative is necessary. Hence, we developed a unified equivalent test for multivariate balance tests and a density test in terms of the max test based on \cite{wellek2010equivalence} with a rectangular equivalence domain.

Recall that we want to verify whether the following $d+1$ restrictions hold simultaneously
\[
 f_+ = f_-,\quad \mu_{Z_1+} = \mu_{Z_1-}, \quad \ldots, \quad \mu_{Z_d+} = \mu_{Z_d-}.
\]
An appropriate set of hypotheses for the equivalence test is as follows:
\begin{equation}
\begin{split}
&H_0: \min\{\tau_f,\tau_{Z_1},\ldots,\tau_{Z_d}\} \leq -\varepsilon \quad \text{or} \quad
\max\{\tau_f,\tau_{Z_1},\ldots,\tau_{Z_d}\} \geq \varepsilon \\
\text{vs} \quad
&H_1: -\varepsilon < \tau_f < \varepsilon \quad \text{and} \quad -\varepsilon < \tau_{Z_k} < \varepsilon
\quad\text{for all} \quad k=1,\ldots,d,
\end{split}
\end{equation}
where $\varepsilon > 0$ determines the range within which the imbalance or discontinuity is sufficiently small. 

To conduct an equivalence test, the researcher must select an appropriate equivalence range $(-\varepsilon,\varepsilon)$. We follow \cite{hartman2018equivalence} and \cite{hartman2021equivalence} to focus on the equivalence confidence interval, which is the shortest equivalence range $(-\varepsilon,\varepsilon)$ such that the alternative hypothesis with the equivalence range would not be detected at the prespecified level $\alpha$ for the observed data.
Hence, we achieve an equivalence confidence interval without manually selecting $\epsilon$.

Denote 
\[
\begin{split}
&\hat \psi_{1,l,p}(h) \\
=& \max\left\{ \frac{\sqrt{n}\sqrt{h_{f}}(\hat\tau_{f,p}(h_f)-\varepsilon)}{\hat V_{f,p}(h_f)}, 
\frac{\sqrt{n}\sqrt{h_{1}}(\hat\tau_{Z_1,l}(h_1)-\varepsilon)}{\hat V_{Z_1,l}(h_1)}, \ldots,
\frac{\sqrt{n}\sqrt{h_{d}}(\hat\tau_{Z_d,l}(h_d)-\varepsilon)}{\hat V_{Z_d,l}(h_d)}
\right\}
\end{split}
\]
and
\[
\begin{split}
&\hat \psi_{2,l,p}(h) \\
=&\min\left\{ \frac{\sqrt{n}\sqrt{h_{f}}(\hat\tau_{f,p}(h_f)+\varepsilon)}{\hat V_{f,p}(h_f)}, 
\frac{\sqrt{n}\sqrt{h_{1}}(\hat\tau_{Z_1,l}(h_1)+\varepsilon)}{\hat V_{Z_1,l}(h_1)}, \ldots,
\frac{\sqrt{n}\sqrt{h_{d}}(\hat\tau_{Z_d,l}(h_d)+\varepsilon)}{\hat V_{Z_d,l}(h_d)}
\right\}.
\end{split}
\]
The following proposition demonstrates the asymptotic validity of the proposed equivalence test for a prespecified level $\alpha$.
\begin{proposition}\label{prop:5.1}
Suppose Assumptions \ref{asp:1} and \ref{asp:2} hold with $R \geq p+2$ and $S\geq l+3$, and the asymptotic covariance matrix $V_{l+1,p+1}$ is nonsingular.
If $n h_{min}\to\infty$, $nh_f^2 \to\infty$, $n (h_{max})^{2l+5}\to0$ and $n (h_f)^{2p+3} \to0$, 
\begin{eqnarray*}
\text{Under }H_0 &:& \limsup_{n\to\infty} P\left( \hat \psi_{1,l+1,p+1}(h) \leq z_{\alpha/2} \text{ or }
\hat \psi_{2,l+1,p+1}(h) \geq z_{1 - \alpha/2} \right) \leq \alpha, \\
\text{Under }H_1 &:& \lim_{n\to\infty} P\left( \hat \psi_{1,l+1,p+1}(h) \leq z_{\alpha/2} \text{ or }
\hat \psi_{2,l+1,p+1}(h) \geq z_{1 - \alpha/2} \right)  = 1,
\end{eqnarray*}
where $z_{\alpha/2}$ and $z_{1 - \alpha/2}$ are the $\alpha/2$ and $(1 - \alpha/2)$ quantiles of $N(0,1)$, respectively.
\end{proposition}

\subsection{Pretesting Analysis} \label{subsec:pretest}

Following \cite{roth2022pretest}, we analyze the bias of the treatment effect estimator induced by conducting a pretest. Let $\hat{\tau}_{Y}, \hat{\tau}_{Z}, \hat{\tau}_{f}$ be the local polynomial estimators of $\lim_{x \to 0+} E[Y_i|X_i=x] - \lim_{x \to 0-} E[Y_i|X_i=x], (\mu_{Z_1+}-\mu_{Z_1-}, \ldots, \mu_{Z_d+}-\mu_{Z_d-})', f_{+}-f_{-}$, respectively.\footnote{To apply Theorem \ref{thm:1}, we use the bias-corrected estimators as $\hat{\tau}_{Y}, \hat{\tau}_{Z}, \hat{\tau}_{f}$.} For simplicity, suppose that $d=1$ and $(\hat{\tau}_{Y}, \hat{\tau}_{Z}, \hat{\tau}_{f})'$ is multivariate normally distributed; that is,
\[
\left(
\begin{array}{c}
\hat{\tau}_{Y} \\
\hat{\tau}_{Z} \\
\hat{\tau}_{f} 
\end{array}
\right) \ \sim \ N \left( \left(
\begin{array}{c}
\tau_{Y} \\
0 \\
0 
\end{array}
\right) + \left(
\begin{array}{c}
\delta_{Y} \\
\delta_{Z} \\
\delta_{f} 
\end{array}
\right), \left(
\begin{array}{ccc}
\sigma_{YY} & \sigma_{YZ}& 0 \\
\sigma_{YZ} & \sigma_{ZZ}& 0 \\
0 & 0 & \sigma_{ff} 
\end{array}
\right)  \right),
\]
where $\tau_{f}$ is the treatment effect and $\delta = (\delta_Y, \delta_Z, \delta_f)'$ is the unconditional bias from the imbalance at the cutoff. Because Theorem \ref{thm:1} implies that the asymptotic covariance between $(\hat{\tau}_Y, \hat{\tau}_Z)'$ and $\hat{\tau}_f$ is zero, we assume that $\hat{\tau}_f$ is independent of $\hat{\tau}_Y$ and $\hat{\tau}_Z$.

Let $\mathcal{T}(\Sigma) \subset \mathbb{R}^{2}$ be the acceptance region, which potentially depends on the covariance matrix $\Sigma = (\sigma_{ZZ},\sigma_{ff})$. From Proposition 1 in \cite{roth2022pretest}, the expectation of $\hat{\tau}_Y$ conditional on passing the pretest can be written as
\begin{eqnarray}
E\left[ \hat{\tau}_{Y} | (\hat{\tau}_Z, \hat{\tau}_{f})' \in \mathcal{T}(\Sigma) \right] \ = \ \tau_{Y} + \delta_{Y} +\frac{\sigma_{YZ}}{\sigma_{ZZ}} \left\{ E\left[ \hat{\tau}_{Z} | (\hat{\tau}_Z, \hat{\tau}_{f})' \in \mathcal{T}(\Sigma) \right] - \delta_{Z} \right\}. \label{bias_pretest}
\end{eqnarray}
Equation (\ref{bias_pretest}) implies that the expectation of $\hat{\tau}_{Y}$ conditional on passing the pretest is the sum of the treatment effect $\tau_{Y}$, the unconditional bias $\delta_{Y}$, and the pretest bias $\frac{\sigma_{YZ}}{\sigma_{ZZ}} \left\{ E\left[ \hat{\tau}_{Z} | (\hat{\tau}_Z, \hat{\tau}_{f})' \in \mathcal{T}(\Sigma) \right] - \delta_{Z} \right\}$. When considering the max test, the acceptance region can be written as
\begin{equation*}
\mathcal{T}(\Sigma) \ = \ \{ (t_Z, t_f)' \in \mathbb{R}^{2} : |t_Z| \leq c_{Z} \ \text{and} \ |t_f| \leq c_{f} \},
\end{equation*}
where $c_Z$ and $c_f$ are positive constants that depend on $\Sigma$. Because $\hat{\tau}_f$ is independent of $\hat{\tau}_Z$, (\ref{bias_pretest}) becomes
\begin{eqnarray}
E\left[ \hat{\tau}_{Y} | (\hat{\tau}_Z, \hat{\tau}_{f})' \in \mathcal{T}(\Sigma) \right] \ = \ \tau_{Y} + \delta_{Y} + \frac{\sigma_{YZ}}{\sigma_{ZZ}} \left\{ E\left[ \hat{\tau}_{Z} | - c_Z \leq \hat{\tau}_Z \leq c_Z \right] - \delta_{Z} \right\}. \label{bias_pretest_max}
\end{eqnarray}
Hence, when we use the Max test as the pretest, the density test statistic $\hat{\tau}_{f}$ does not affect the pretest bias, which disappears if $\delta_Z=0$. 

To discuss the sign of the pretest bias, we consider the following model:
\begin{equation*}
Y_i \ = \ \tau_{Y} 1\{X_i \geq 0\} + \mu_Y(X_i) + \beta Z_i + \epsilon_i,
\end{equation*}
where $\mu_Y$ is a continuous function, and $\epsilon_i$ is an error term satisfying $E[\epsilon_i|X_i,Z_i]=0$. Then, we have
\[
\lim_{x\to0^+} E[Y_i|X_i=x] - \lim_{x\to0^+} E[Y_i|X_i=x] \ = \ \tau_Y + \beta \delta_Z,
\]
where we use $\lim_{x\to0^+} E[Z_i|X_i=x] - \lim_{x\to0^+} E[Z_i|X_i=x] = \delta_Z$. This implies that the unconditional bias becomes $\delta_Y = \beta \delta_Z$. Hence, if $\beta > 0$ and $\sigma_{YZ} >0$, then the sign of the pretest bias is opposite to that of the unconditional bias $\delta_Y$. Therefore, in this case, the pretest can reduce bias. In Online Appendix E, we conduct a numerical analysis using simulated and calibrated datasets. The numerical results also support our theoretical analysis that pretesting may reduce bias.

An analogous analysis for the variance implies that the variance of the treatment effect estimate is likely to be smaller after pretesting. Similar to Proposition 3 of \cite{roth2022pretest}, we obtain
\begin{equation}
    V\left[ \hat{\tau}_{Y} | (\hat{\tau}_Z, \hat{\tau}_{f})' \in \mathcal{T}(\Sigma) \right] \ = \ \sigma_{YY}  + \left(\frac{\sigma_{YZ}}{\sigma_{ZZ}}\right)^2 \left\{ V\left[ \hat{\tau}_{Z} | (\hat{\tau}_Z, \hat{\tau}_{f})' \in \mathcal{T}(\Sigma) \right] - \sigma_{ZZ} \right\}. \label{variance_pretest}
\end{equation}
If the max test is conducted, the acceptance region is convex. Hence, Proposition 4 in \cite{roth2022pretest} implies that 
\[
 V\left[ \hat{\tau}_{Z} | (\hat{\tau}_Z, \hat{\tau}_{f})' \in \mathcal{T}(\Sigma) \right] - \sigma_{ZZ} \leq 0.
\]
Hence, the variance is likely to be smaller after pretesting. \footnote{The above analysis indicates that pretesting can reduce the bias and variance of the treatment effect estimates. Its implication for the power of the test for the treatment effect is indefinite. For example, pretesting reduces the variance of estimator when the bias is zero $\delta_Y = \delta_Z = 0$. The smaller variance may increase the power when $\tau_Y$ is away from $0$ and decrease the power when $\tau_Y$ is close to $0$.}
\FloatBarrier

\section{Conclusion}\label{sec:6}

Diagnostic tests are vital for applied studies of RD designs; however, their inappropriate use frustrates their purpose. The diagnostic procedure for RD design aims to test the underlying null hypothesis that all pre-determined covariates balance and the density function is continuous. Current practices test these restrictions separately without appropriate size control---they over-reject the null hypothesis that the identifying restriction for the RD design is plausible.

We demonstrate an extraordinary over-rejection of the null hypothesis, where each testable restriction is valid. Among the $59$ studies published in the top five economics journals, $32\%$ ($19$ studies) rejected at least one testable restriction. Nevertheless, less than $5\%$ of the $787$ tests were rejected. Hence, we conclude that the plausibility of identification is highly over-rejected where the underlying null hypotheses appear to hold. We blame the multiple testing problem for the over-rejection because most studies run more than $10$ covariates balance tests separately.


We provide a solution to this over-rejection problem. We demonstrate the joint asymptotic normality of RD design estimators and provide alternative unified testing procedures. We show two different unified tests: the standardized Wald (sWald) test and the Max test. In the numerical simulations, the sWald and Max statistics achieved appropriate size-control properties for a moderate number of covariates when the sample size was $1000$. The sWald test is particularly superior to other methods in size control, as it can control for any dimension of covariates up to $25$ where the Bonferroni correction may fail to control the size. The power properties vary according to the underlying alternative hypotheses. The Max test is superior to the Bonferroni correction when only one covariate exhibits a jump in the conditional mean. The sWald test is superior to the Bonferroni correction when all the covariates jump in their conditional means. 

In conclusion, we recommend the Max test for a design with a moderate number of covariates when only a few covariates which critically matter for the plausibility of identification such as placebo outcomes. We recommend the sWald test for other designs. This test is appropriate for most practical designs where there is no ex-ante concerning covariate and testing many covariate balancing is the norm.

Our unified test provides a formal judgment for whether the rejected diagnostic tests were spurious and by luck or not. In RD designs, rejections in the diagnostic tests can be genuine signals for their designs being implausible while they may be spurious. Hence, our procedure offers a transparent decision rule that is essential to the users if the diagnostic tests are part of the assessment in RD designs.

Furthermore, our procedure is complementary with the other current practices. For example, practitioners often conducts an additional inspection for a covariate with an apparent graphical evidence for a jump. We continue to recommend graphical analyses and those careful inspections because such procedure is complementary with our unified tests. With our unified test, practitioners may have stronger argument for a plausible design in the following combination of arguments: our unified test suggests that the rejections were spurious; furthermore, other anecdotes conclude that the spurious rejections were immaterial to the plausibility of the design. Hence, our procedure shall be an essential support for effectively justifying designs with a few rejected diagnostic tests when the underlying design is genuinely valid.

Based on our findings, we consider several topics for future research. First, finding another use for testable restrictions may be preferable. To validate the plausibility of identification, we must focus on the underlying null hypothesis rather than detecting which restrictions fail. Nevertheless, if we detect a particular subgroup of covariates that threats the plausibility of identification, such a finding may provide valuable information. Diagnostic procedures using such information may improve design validation. Second, although we focused on the statistical inference of testable restrictions, a graphical analysis of testable restrictions is also important. Recent studies, such as \citet{Calonico_Cattaneo_Titiunik_2015_JASA}, have considered the optimal procedure for visualization, and \citet{Korting_Lieberman_Matsudaira_Pei_Shen_2023} have demonstrated that tuning parameters in visualization is critical. We recommend our unified testing along with graphical analysis, with multiple testing problems in graphical analyses as a future research topic. Finally, evaluating the joint distribution of testing statistics that are nonparametric estimates is challenging. We will consider these issues as future research questions.


\FloatBarrier

\section*{Acknowledgements}

The study was supported by JSPS KAKENHI Grant Numbers JP20J20046, JP23K18799 (Fusejima), JP22K13373 (Ishihara), and JP21K13269 (Sawada). This study is derived from a chapter of Koki Fusejima's dissertation. We thank Eiji Kurozumi, Yukitoshi Matsushita, Katsumi Shimotsu, and Kohei Yata, the anonymous referees and associate editor as well as seminar participants at Hitotsubashi University, The University of Tokyo, Kansai Keiryo Keizaigaku Kenkyukai, Japanese Joint Statistical Meeting, International Conference on Econometrics and Statistics, and the Asian Meeting of the Econometric Society in East and South-East Asia. Moreover, we thank Masaki Oguni and Yuri Sugiyama for their excellent research assistance.


\begin{appendix}
\section{Formal results}\label{sec:a}
Appendix \ref{sec:a} provides formal theoretical results and proofs.

\subsection{Setup and Notation}\label{sec:a.1}
Let $|\cdot|$ denote the Euclidean matrix norm, that is, $|A|^2=\trace(A'A)=\sum_{j=1}^n\sum_{i=1}^m a_{ij}^2$ for $(m\times n)$ matrix $A=\{a_{ij}\}_{i=1,\ldots,m}^{j=1,\ldots,n}$, and $|a|^4=|aa'|^2=\sum_{j=1}^n\sum_{i=1}^n a_{i}^2a_{j}^2$ for $(n\times 1)$ vector $a=(a_1,\ldots,a_n)'$.
Let $a_n\precsim b_n$ denote $a_n\leq Cb_n$ for a positive constant $C$, which is independent of $n$.

The estimands of the local polynomial estimators are
\begin{equation*}\label{eq:19}
\beta_{Z_k+,l}=
\left[\mu_{Z_k+},\frac{\mu_{Z_k+}^{(1)}}{1!},\ldots,\frac{\mu_{Z_k+}^{(l)}}{l!}\right]',\,\,
\beta_{Z_k-,l}=
\left[\mu_{Z_k-},\frac{\mu_{Z_k-}^{(1)}}{1!},\ldots,\frac{\mu_{Z_k-}^{(l)}}{l!}\right]',
\end{equation*}
\begin{equation*}\label{eq:f19}
\beta_{f+,p}=
\left[{F}_+,\frac{f_+}{1!},\frac{f^{(1)}_+}{2!},\ldots,\frac{f^{(p-1)}_+}{p!}\right]',\,\,
\beta_{f-,p}=
\left[{F}_-,\frac{f_-}{1!},\frac{f^{(1)}_-}{2!},\ldots,\frac{f^{(p-1)}_-}{p!}\right]',
\end{equation*}
where $F(x)=E[1\{X \leq x\}]$, ${F}_+ = \lim_{x \to0^+} F(x)$, ${F}_- = \lim_{x \to0^-} F(x)$. \\
We set $Z_k=(Z_{k,1},\ldots,Z_{k,n})'$,
$\tilde{F}=(\tilde{F}(X_1),\ldots,\tilde{F}(X_n))'$,
$K_h(u)=K(u/h)/h$, $H_p(h)=\diag(1,h,\ldots,h^{p})$,
$X_p(h)=[r_p(X_1/h),\ldots,r_p(X_n/h)]'$,
\begin{equation*}\label{eq:5}
W_+(h)=\diag(1\{X_1 \geq 0\}K_{h}(X_1),\ldots,1\{X_n \geq 0\}K_{h}(X_{n})),
\end{equation*}
\begin{equation*}
W_-(h)=\diag(1\{X_1 < 0\}K_{h}(X_1),\ldots,1\{X_n < 0\}K_{h}(X_{n})),
\end{equation*}
\begin{equation*}\label{eq:6}
\Gamma_{+,p}(h) = \frac{1}{n} X_p(h)'W_+(h)X_p(h),\quad
\Gamma_{-,p}(h) = \frac{1}{n} X_p(h)'W_-(h)X_p(h),
\end{equation*}
The local polynomial estimators are as follows.
\begin{equation*}\label{eq:10}
\hat{\beta}_{Z_k+,l}(h_k)=
\frac{1}{n} H_l^{-1}(h_k)\Gamma_{+,l}^{-1}(h_k)X_l(h_k)'W_+(h_k)Z_k,
\end{equation*}
\begin{equation*}
\hat{\beta}_{Z_k-,l}(h_k)=
\frac{1}{n} H_l^{-1}(h_k)\Gamma_{-,l}^{-1}(h_k)X_l(h_k)'W_-(h_k)Z_k,
\end{equation*}
and
\[
\hat{\beta}_{f+,p}(h_{f}) = 
\frac{1}{n} H_p^{-1}(h_f)\Gamma_{+,p}^{-1}(h_f) X_p(h_f)'W_+(h_f) \tilde{F},
\]
\[
\hat{\beta}_{f-,p}(h_{f}) = 
\frac{1}{n} H_p^{-1}(h_f)\Gamma_{-,p}^{-1}(h_f) X_p(h_f)'W_-(h_f) \tilde{F}.
\]
Let
\[
\Delta \hat \mu_{Z_k+,l}(h_k) \equiv \sqrt{n h_{k}}[\hat{\mu}_{Z_k+,l}(h_k)-{\mu}_{Z_k+}],
\quad \Delta \hat \mu_{Z_k-,l}(h_k) \equiv \sqrt{n h_{k}}[\hat{\mu}_{Z_k-,l}(h_k)-{\mu}_{Z_k-}],
\]
\[
\Delta \hat f_{+,p}(h_f) \equiv \sqrt{n h_f}[\hat{f}_{+,p}(h_f) - f_+],
\quad \Delta \hat f_{-,p}(h_f) \equiv \sqrt{n h_f}[\hat{f}_{-,p}(h_f) - f_-].
\]

Letting $\mX_n=(X_1,\ldots,X_n)'$, we decompose $\Delta \hat \mu_{Z_k+,l}(h_k)$ and $\Delta \hat \mu_{Z_k-,l}(h_k)$ as
\[
\Delta \hat \mu_{Z_k+,l}(h_k)
=\Delta^M \hat \mu_{Z_k+,l}(h_k)+\Delta^B \hat \mu_{Z_k+,l}(h_k),
\]
\[
\Delta \hat \mu_{Z_k-,l}(h_k)
=\Delta^M \hat \mu_{Z_k-,l}(h_k)+\Delta^B \hat \mu_{Z_k-,l}(h_k),
\]
where
\[
\Delta^M \hat \mu_{Z_k+,l}(h_k) \equiv \sqrt{n h_{k}}\{\hat{\mu}_{Z_k+,l}(h_k)-E[\hat{\mu}_{Z_k+,l}(h_k)|\mX_n]\},
\]
\[
\Delta^M \hat \mu_{Z_k-,l}(h_k) \equiv \sqrt{n h_{k}}\{\hat{\mu}_{Z_k-,l}(h_k)-E[\hat{\mu}_{Z_k-,l}(h_k)|\mX_n]\}
\]
and
\[
\Delta^B \hat \mu_{Z_k+,l}(h_k) \equiv \sqrt{n h_{k}}\{E[\hat{\mu}_{Z_k+,l}(h_k)|\mX_n]-{\mu}_{Z_k+}\},
\]
\[
\Delta^B \hat \mu_{Z_k-,l}(h_k) \equiv \sqrt{n h_{k}}\{E[\hat{\mu}_{Z_k-,l}(h_k)|\mX_n]-{\mu}_{Z_k-}\}.
\]
$\Delta^B \hat \mu_{Z_k+,l}(h_k)$ and $\Delta^B \hat \mu_{Z_k-,l}(h_k)$ represent the smoothing bias. We consider the distributional approximations of $\Delta^M \hat \mu_{Z_k+,l}(h_k)$ and $\Delta^M \hat \mu_{Z_k-,l}(h_k)$.

We define $\tilde f_{+,p}(h_{f})$ and $\tilde f_{-,p}(h_{f})$ as slope coefficients in the following local polynomial regressions.
\begin{eqnarray*}
\tilde f_{+,p}(h_{f}) = e_1'\tilde{\beta}_{f+,p}(h_{f}), \quad
\tilde f_{-,p}(h_{f}) = e_1'\tilde{\beta}_{f-,p}(h_{f}),
\end{eqnarray*}
where
\begin{eqnarray*}
\tilde{\beta}_{f+,p}(h_{f})=\argmin_{\beta\in\R^{p+1}}
\sumin1\{X_i \geq 0\} ({F}(X_i)-r_p(X_i)'\beta)^2{K}_{h_f}(X_i), \\
\tilde{\beta}_{f-,p}(h_{f})=\argmin_{\beta\in\R^{p+1}}
\sumin1\{X_i < 0\} ({F}(X_i)-r_p(X_i)'\beta)^2{K}_{h_f}(X_i).
\end{eqnarray*}
$\tilde f_{+,p}(h_{f})$ and $\tilde f_{-,p}(h_{f})$ are obtained from $p$th-order local polynomial one-sided approximations of the true distribution of $X_i$.
We decompose $\Delta \hat f_{+,p}(h_f)$ and $\Delta \hat f_{-,p}(h_f)$ as
\[
\Delta \hat f_{+,p}(h_f)
=\Delta^{M_0} \hat f_{+,p}(h_f) + \Delta^{B_0} \hat f_{+,p}(h_f), \quad
\Delta \hat f_{-,p}(h_f)
=\Delta^{M_0} \hat f_{-,p}(h_f) + \Delta^{B_0} \hat f_{-,p}(h_f),
\]
where
\[
\Delta^{M_0} \hat f_{+,p}(h_f) \equiv \sqrt{n h_f}\{\hat f_{+,p}(h_f) - \tilde f_{+,p}(h_{f})\},
\]
\[
\Delta^{M_0} \hat f_{-,p}(h_f) \equiv \sqrt{n h_f}\{\hat f_{-,p}(h_f) - \tilde f_{-,p}(h_{f})\}
\]
and
\[
\Delta^{B_0} \hat f_{+,p}(h_f) \equiv \sqrt{n h_f}\{\tilde f_{+,p}(h_{f}) - f_+\}, \quad
\Delta^{B_0} \hat f_{-,p}(h_f) \equiv \sqrt{n h_f}\{\tilde f_{-,p}(h_{f}) - f_-\}.
\]
$\Delta^{B_0} \hat f_{+,p}(h_f)$ and $\Delta^{B_0} \hat f_{-,p}(h_f)$ represent smoothing bias.
We set \\$F=({F}(X_1),\ldots,{F}(X_{n}))'$, $\iota=(1,\ldots,1)'$, which are $n$-dimensional.
\[
U_{+,p}(h_f) \equiv \frac{1}{n(n-1)}\sum_{i,j:i\neq j}^{n} \psi_{+,p,h_f}(X_i, X_j),
\]
\[
U_{-,p}(h_f) \equiv \frac{1}{n(n-1)}\sum_{i,j:i\neq j}^{n} \psi_{-,p,h_f}(X_i, X_j), 
\]
where
\[
\psi_{+,p,h_f}(X_i, X_j) \equiv 1\{X_i \geq 0\}r_p(X_i/h_f)(1\{X_j \leq X_i\} - F(X_i)){K}_{h_f}(X_i),
\]
\[
\psi_{-,p,h_f}(X_i, X_j) \equiv 1\{X_i < 0\}r_p(X_i/h_f)(1\{X_j \leq X_i\} - F(X_i)){K}_{h_f}(X_i).
\]
We further decompose $\Delta^{M_0} \hat f_{+,p}(h_f)$ and $\Delta^{M_0} \hat f_{-,p}(h_f)$ as follows.
\[
\Delta^{M_0} \hat f_{+,p}(h_f)
=\Delta^{M_1} \hat f_{+,p}(h_f) + \Delta^{B_1} \hat f_{+,p}(h_f),
\]
\[
\Delta^{M_0} \hat f_{-,p}(h_f)
=\Delta^{M_1} \hat f_{-,p}(h_f) + \Delta^{B_1} \hat f_{-,p}(h_f),
\]
where
\[
\Delta^{M_1} \hat f_{+,p}(h_f) \equiv \left(1-\frac{1}{n}\right)\sqrt{\frac{n}{h_{f}}} e_1' \Gamma_{+,p}^{-1}(h_f) U_{+,p}(h_f),
\]
\[
\Delta^{M_1} \hat f_{-,p}(h_f) \equiv \left(1-\frac{1}{n}\right)\sqrt{\frac{n}{h_{f}}} e_1' \Gamma_{-,p}^{-1}(h_f) U_{-,p}(h_f),
\]
and
\[
\Delta^{B_1} \hat f_{+,p}(h_f) \equiv \frac{1}{\sqrt{n h_{f}}}\frac{1}{n} e_1' \Gamma_{+,p}^{-1}(h_f) X_p(h_f)' W_+(h_f) (\iota - {F}),
\]
\[
\Delta^{B_1} \hat f_{-,p}(h_f) \equiv \frac{1}{\sqrt{n h_{f}}}\frac{1}{n} e_1' \Gamma_{-,p}^{-1}(h_f) X_p(h_f)' W_-(h_f) (\iota - {F}).
\]
$U_{+,p}(h_f)$ and $U_{-,p}(h_f)$ are considered second-order U-statistics (where the kernels $\psi_{+,p,h_f}(X_i, X_j)$ and $\psi_{+,p,h_f}(X_i, X_j)$ are $(p+1)$-dimensional and depend on the bandwidth $h_f$), and
$\Delta^{B_1} \hat f_{+,p}(h_f)$ and $\Delta^{B_1} \hat f_{-,p}(h_f)$ represent the leave-in bias.
These terms arise because the empirical distribution function $\tilde{F}(X_i)$ is approximated, which leads to double summation.
Note that $E[\psi_{+,p,h_f}(X_i, X_j)]=0$ and $E[\psi_{-,p,h_f}(X_i, X_j)]=0$. U-statistics can be decomposed as
\[
U_{+,p}(h_f)=U_{1,+,p}(h_f) + U_{2,+,p}(h_f), \quad
U_{-,p}(h_f)=U_{1,-,p}(h_f) + U_{2,-,p}(h_f)
\]
with
\[
U_{1,+,p}(h_f) \equiv \frac{1}{n} \sumin E[\psi_{+,p,h_f}(X_j, X_i)|X_i], \quad
U_{1,-,p}(h_f) \equiv \frac{1}{n} \sumin E[\psi_{-,p,h_f}(X_j, X_i)|X_i]
\]
for $j\neq i$ and
\[
U_{2,+,p}(h_f) \equiv \frac{1}{n(n-1)}\sum_{i,j:i\neq j}^{n} \varphi_{+,p,h_f}(X_i, X_j), 
\]
\[
U_{2,-,p}(h_f) \equiv \frac{1}{n(n-1)}\sum_{i,j:i\neq j}^{n} \varphi_{-,p,h_f}(X_i, X_j),
\]
where
\[
\begin{split}
\varphi_{+,p,h_f}(X_i, X_j) \equiv& \psi_{+,p,h_f}(X_i, X_j) + \psi_{+,p,h_f}(X_j, X_i) \\
&- E[\psi_{+,p,h_f}(X_i, X_j)|X_j] - E[\psi_{+,p,h_f}(X_j, X_i)|X_i],
\end{split}
\]
\[
\begin{split}
\varphi_{-,p,h_f}(X_i, X_j) \equiv& \psi_{-,p,h_f}(X_i, X_j) + \psi_{-,p,h_f}(X_j, X_i) \\
&- E[\psi_{-,p,h_f}(X_i, X_j)|X_j] - E[\psi_{-,p,h_f}(X_j, X_i)|X_i].
\end{split}
\]
This decomposition is often referred to as the Hoeffding decomposition (see, e.g., \cite{serfling2009approximation}).
$U_{2,+,p}(h_f)$ and $U_{2,-,p}(h_f)$ are degenerate U-statistics that are later shown to be asymptotically negligible under our assumptions.
Let
\[
\Delta^{B_2} \hat f_{+,p}(h_f) \equiv \sqrt{\frac{n}{h_{f}}} e_1' \Gamma_{+,p}^{-1}(h_f) U_{2,+,p}(h_f),
\]
\[
\Delta^{B_2} \hat f_{-,p}(h_f) \equiv \sqrt{\frac{n}{h_{f}}} e_1' \Gamma_{-,p}^{-1}(h_f) U_{2,-,p}(h_f),
\]
\[
\Delta^{B_3} \hat f_{+,p}(h_f) \equiv -\frac{1}{\sqrt{n h_{f}}} e_1' \Gamma_{+,p}^{-1}(h_f) U_{+,p}(h_f),
\]
\[
\Delta^{B_3} \hat f_{-,p}(h_f) \equiv -\frac{1}{\sqrt{n h_{f}}} e_1' \Gamma_{-,p}^{-1}(h_f) U_{-,p}(h_f).
\]
We can decompose $\Delta \hat f_{+,p}(h_f)$ and $\Delta \hat f_{-,p}(h_f)$ into
\[
\Delta \hat f_{+,p}(h_f)
=\Delta^M \hat f_{+,p}(h_f) + \Delta^B \hat f_{+,p}(h_f),\quad
\Delta \hat f_{-,p}(h_f)
=\Delta^M \hat f_{-,p}(h_f) + \Delta^B \hat f_{-,p}(h_f),
\]
where
\[
\Delta^M \hat f_{+,p}(h_f) \equiv 
\sqrt{\frac{n}{h_{f}}} e_1' \Gamma_{+,p}^{-1}(h_f) U_{1,+,p}(h_f),\quad
\Delta^M \hat f_{-,p}(h_f) \equiv 
\sqrt{\frac{n}{h_{f}}} e_1' \Gamma_{-,p}^{-1}(h_f) U_{1,-,p}(h_f)
\]
and
\[
\Delta^B \hat f_{+,p}(h_f) \equiv 
\sum_{j=0}^3 \Delta^{B_j} \hat f_{+,p}(h_f), \quad
\Delta^B \hat f_{-,p}(h_f) \equiv 
\sum_{j=0}^3 \Delta^{B_j} \hat f_{-,p}(h_f).
\]
We consider distributional approximations of $\Delta^M \hat f_{+,p}(h_f)$ and $\Delta^M \hat f_{-,p}(h_f)$.

To characterize the bias and variance of the local polynomial estimators, we employ the following notation.
\begin{equation*}\label{eq:0.2}
\sigma^{2}_{Z_k+}=\lim_{x\to0^+}\sigma^2_{Z_k}(x), \quad
\sigma^{2}_{Z_k-}=\lim_{x\to0^-}\sigma^2_{Z_k}(x), \quad
\sigma^2_{Z_k}(x)=\Var(Z_k|X=x),
\end{equation*}
\begin{equation*}\label{eq:0.3}
\sigma_{Z_j Z_k+}=\lim_{x\to0^+}\sigma_{Z_j Z_k}(x), \quad
\sigma_{Z_j Z_k-}=\lim_{x\to0^-}\sigma_{Z_j Z_k}(x), \quad
\sigma_{Z_j Z_k}(x)=\Cov(Z_j,Z_k|X=x),
\end{equation*}
$\eps_{Z_k}=(\eps_{Z_k,1},\ldots,\eps_{Z_k,n})'$ with $\eps_{Z_k,i}=Z_{k,i}-\mu_{Z_k}(X_i)$,
\[
\Sigma_{Z_k}=E[\eps_{Z_k}\eps_{Z_k}'|\mX_n]=
\diag(\sigma^2_{Z_k}(X_1),\ldots,\sigma^2_{Z_k}(X_n)),
\]
\[
\Sigma_{Z_j Z_k}=E[\eps_{Z_j}\eps_{Z_k}'|\mX_n]=
\diag(\sigma_{Z_j Z_k}(X_1),\ldots,\sigma_{Z_j Z_k}(X_n)),
\]
\[
S_p(h)=[(X_1/h)^p,\ldots,(X_n/h)^p]',
\]
\begin{equation*}\label{eq:7}
\vartheta_{+,p,q}(h)=X_p(h)'W_+(h)S_q(h)/n,\quad
\vartheta_{-,p,q}(h)=X_p(h)'W_-(h)S_q(h)/n,
\end{equation*}
\begin{equation*}\label{eq:8}
\Psi_{Z_j Z_k+,p,q}(h_j,h_k)=
X_p(h_j)'W_+(h_j)\Sigma_{Z_j Z_k}W_+(h_k)X_q(h_k)/n,
\end{equation*}
\begin{equation*}
\Psi_{Z_j Z_k-,p,q}(h_j,h_k)=
X_p(h_j)'W_-(h_j)\Sigma_{Z_j Z_k}W_-(h_k)X_q(h_k)/n,
\end{equation*}
\begin{equation*}\label{eq:9}
\Psi_{Z_k+,p}(h)=\Psi_{Z_k Z_k+,p,p}(h,h)=
X_p(h)'W_+(h)\Sigma_{Z_k}W_+(h)X_p(h)/n,
\end{equation*}
\begin{equation*}
\Psi_{Z_k-,p}(h)=\Psi_{Z_k Z_k-,p,p}(h,h)=
X_p(h)'W_-(h)\Sigma_{Z_k}W_-(h)X_p(h)/n,
\end{equation*}
\begin{equation*}\label{eq:13}
\begin{split}
&\tilde \Psi_{Z_j Z_k+,p,q}(h_j,h_k)\\
=&\int_0^\infty K\left(\frac{h_{jk}u}{h_j}\right)
K\left(\frac{h_{jk}u}{h_k}\right)
r_p\left(\frac{h_{jk}u}{h_j}\right)
r_q\left(\frac{h_{jk}u}{h_k}\right)'
\sigma^2_{Z_j Z_k}(uh_{jk})f(uh_{jk})du, 
\end{split}
\end{equation*}
\begin{equation*}
\begin{split}
&\tilde \Psi_{Z_j Z_k-,p,q}(h_j,h_k)\\
=&\int_0^\infty K\left(\frac{h_{jk}u}{h_j}\right)
K\left(\frac{h_{jk}u}{h_k}\right)
r_p\left(\frac{h_{jk}u}{h_j}\right)
r_q\left(\frac{h_{jk}u}{h_k}\right)'
\sigma^2_{Z_j Z_k}(-uh_{jk})f(-uh_{jk})du, 
\end{split}
\end{equation*}
\begin{equation*}\label{eq:14}
\tilde \Psi_{Z_k+,p}(h)=\tilde \Psi_{Z_k Z_k+,p,p}(h,h)
=\int_0^\infty K(u)^2r_p(u)r_p(u)'\sigma^2_{Z_k}(uh)f(uh)du,
\end{equation*}
\begin{equation*}
\tilde \Psi_{Z_k-,p}(h)=\tilde \Psi_{Z_k Z_k-,p,p}(h,h)
=\int_0^\infty K(u)^2r_p(u)r_p(u)'\sigma^2_{Z_k}(-uh)f(-uh)du,
\end{equation*}
\begin{equation*}\label{eq:15}
\Gamma_{p}=
\int_0^\infty K(u)r_p(u)r_p(u)'du,\quad
\vartheta_{p,q}=
\int_0^\infty K(u)u^q r_p(u)du,
\end{equation*}
where $h_{jk}=(h_j h_k)^{1/2}$,
\begin{equation*}\label{eq:17}
\begin{split}
&\Psi_{jk,p,q}\\
&=\begin{cases}
\int_0^\infty K\left((c_{jk})^{1/2}u\right)
K\left((c_{jk})^{-1/2}u\right)
r_p\left((c_{jk})^{1/2}u\right)
r_q\left((c_{jk})^{-1/2}u\right)'du & \text{if }
\rho_{jk}\to c_{jk} \\
0 & \text{otherwise}
\end{cases} \\
&=\begin{cases}
(c_{jk})^{1/2}\int_0^\infty K\left(c_{jk}u\right)
K\left(u\right)
r_p\left(c_{jk}u\right)
r_q\left(u\right)'du & \text{if }
\rho_{jk}\to c_{jk}\\
0 & \text{otherwise},
\end{cases} 
\end{split}
\end{equation*}
for $c_{jk} \in(0,\infty)$,
\begin{equation*}\label{eq:18}
\Psi_{p,q}=\Psi_{kk,p,q}
=\int_0^\infty K(u)^2r_p(u)r_q(u)'du,
\end{equation*}
\begin{equation*}\label{eq:f21}
\Psi_{f,p} = \int_0^{\infty} \int_{0}^{\infty} (u \wedge v) r_p(u) r_p(v)'K(u)K(v) du dv.
\end{equation*}

\subsection{Preliminary lemmas}\label{sec:a.2}
First, we present the asymptotic bias, variance, and distribution of the local polynomial estimators.
We provide the proofs of the lemmas in Appendix \ref{sec:a.4}.

The following lemma describes the asymptotic behavior of the bias terms for the local polynomial estimators.

\begin{lemma}\label{lem:sa3.b}
Suppose Assumptions 1-2 hold with $S\geq l+2$ and $R \geq p+1$.
\begin{enumerate}
\item[(a)] If $n h_{min}\to\infty$ and $h_{max}\to0$, then
\begin{equation}\label{eq:b1}
\begin{split}
&\Delta^B \hat \mu_{Z_k+,l}(h_k) \\
=&\sqrt{n}(h_k)^{3/2+l}B_{Z_k+,l,l+1}(h_k)
+ \sqrt{n}(h_k)^{5/2+l}B_{Z_k+,l,l+2}(h_k)[1+o_p(1)],
\end{split}
\end{equation}
\begin{equation}
\begin{split}
&\Delta^B \hat \mu_{Z_k-,l}(h_k) \\
=&\sqrt{n}(h_k)^{3/2+l}B_{Z_k-,l,l+1}(h_k)
+ \sqrt{n}(h_k)^{5/2+l}B_{Z_k-,l,l+2}(h_k)[1+o_p(1)],
\end{split}
\end{equation}
where
\begin{equation}\label{eq:b2}
B_{Z^k+,l,r}(h_k)=\frac{\mu^{(r)}_{Z_k+}}{r!}
\mB_{+,0,l,r}(h_k),\quad
B_{Z^k-,l,r}(h_k)=\frac{\mu^{(r)}_{Z_k-}}{r!}
\mB_{-,0,l,r}(h_k),
\end{equation}
\begin{equation}\label{eq:b3}
\mB_{+,s,l,r}(h_k)=s!e_s'\Gamma_{+,l}^{-1}(h_k) \vartheta_{+,l,r}(h_k)=s!e_s'\Gamma_{l}^{-1} \vartheta_{l,r}+o_p(1),
\end{equation}
\begin{equation}
\mB_{-,s,l,r}(h_k)=s!e_s'\Gamma_{-,l}^{-1}(h_k) \vartheta_{-,l,r}(h_k)=(-1)^{s+r}s!e_s'\Gamma_{l}^{-1} \vartheta_{l,r}+o_p(1).
\end{equation}
\item[(b)] If $nh_f \to\infty$ and $h_f \to 0$, then
\begin{equation}\label{eq:bf1}
\begin{split}
&\Delta^{B_0} \hat f_{+,p}(h_f) \\
=&\sqrt{n}h^{p+1/2}_f B_{f,+,p,p+1}(h_f)
+\sqrt{n}h^{p+3/2}_f B_{f,+,p,p+2}(h_f)[1 + o_p(1)], \\ 
&\Delta^{B_1} \hat f_{+,p}(h_f) = O_p((n h_f)^{-1/2})
=o_p(1),
\end{split}
\end{equation}
\begin{equation}
\begin{split}
&\Delta^{B_0} \hat f_{-,p}(h_f) \\
=&\sqrt{n}h^{p+1/2}_f B_{f,-,p,p+1}(h_f)
+\sqrt{n}h^{p+3/2}_f B_{f,-,p,p+2}(h_f)[1 + o_p(1)], \\
&\Delta^{B_1} \hat f_{-,p}(h_f) =o_p(1),
\end{split}
\end{equation}
where
\begin{equation}\label{eq:bf2}
B_{f,+,p,q}(h_f) = \frac{f_+^{(q-1)}}{(q-1)!} \mB_{+,1,p,q}(h_f), \quad
B_{f,-,p,q}(h_f) = \frac{f_-^{(q-1)}}{(q-1)!} \mB_{-,1,p,q}(h_f).
\end{equation}
\item[(c)] If $nh_f^2 \to\infty$ and $h_f \to 0$, then
\begin{equation}\label{eq:bf3}
\Delta^{B_2} \hat f_{+,p}(h_f) = O_p((n h_f^2)^{-1/2}) = o_p(1), \quad
\Delta^{B_2} \hat f_{-,p}(h_f) = o_p(1).
\end{equation}
\end{enumerate}
\end{lemma}

Observe that 
\[
\Delta^M \hat \mu_{Z_k+,l}(h_k) = \sqrt{\frac{h_{k}}{n}} e_0' \Gamma_{+,l}^{-1}(h_k)X_l(h_k)'W_+(h_k)\eps_{Z_k},
\]
\[
\Delta^M \hat \mu_{Z_k-,l}(h_k) = \sqrt{\frac{h_{k}}{n}} e_0' \Gamma_{-,l}^{-1}(h_k)X_l(h_k)'W_-(h_k)\eps_{Z_k}.
\]
The following lemma is used to derive the asymptotic variance of the local polynomial estimators.
\begin{lemma}\label{lem:sa3.v}
Suppose Assumptions 1-2 hold with $S\geq l+2$ and $R \geq p+1$.
\begin{enumerate}
\item[(a)] If $n h_{min}\to\infty$ and $h_{max}\to0$, then
\[
\begin{split}
&\Cov\left(\sqrt{\frac{h_{j}}{n}}X_l(h_j)'W_+(h_j)\eps_{Z_j}, \sqrt{\frac{h_{k}}{n}}X_l(h_k)'W_+(h_k)\eps_{Z_k} \right) \\
=& f_+ \sigma_{Z_j Z_k+} \Psi_{j k,l,l} + o(1)
\end{split}
\]
and
\[
\begin{split}
&\Cov\left(\sqrt{\frac{h_{j}}{n}}X_l(h_j)'W_-(h_j)\eps_{Z_j}, \sqrt{\frac{h_{k}}{n}}X_l(h_k)'W_-(h_k)\eps_{Z_k} \right) \\
=& f_- \sigma_{Z_j Z_k-} H_l(-1)\Psi_{j k,l,l}H_l(-1) + o(1).
\end{split}
\]

\item[(b)] If $nh_f^2 \to\infty$ and $h_f \to 0$, then
\begin{eqnarray*}
&&\Var\left(\sqrt{\frac{n}{h_{f}}} U_{1,+,p}(h_f) \right) \\ 
&=& \frac{1}{h_f} f_+^2 (F_+ - F_+^2) \Gamma_{p}e_0e_0'\Gamma_{p} + f_+^3 \Psi_{f,p} \\
&&+ \{- f_+^3 F_+ + f_+ f'_+ (F_+ - F_+^2) \} \Gamma_{p}(e_1e_0'+e_0e_1')\Gamma_{p} + o(1)
\end{eqnarray*}
and
\begin{eqnarray*}
&&\Var\left(\sqrt{\frac{n}{h_{f}}} U_{1,-,p}(h_f) \right) \\
&=& \frac{1}{h_f} f_-^2 (F_- - F_-^2) H_p(-1)\Gamma_{p}e_0e_0'\Gamma_{p}H_p(-1) + f_-^3 H_p(-1)\Psi_{f,p}H_p(-1) \\
&&- \{- f_-^3 (F_- - 1) + f_- f'_- (F_- - F_-^2) \} H_p(-1)\Gamma_{p}(e_1e_0'+e_0e_1')\Gamma_{p}H_p(-1) \\
&&+ o(1).
\end{eqnarray*}

\item[(c)]
\[
\Cov\left(X_l(h_k)'W_+(h_k)\eps_{Z_k}, U_{1,+,p}(h_f) \right) = 0
\]
and
\[
\Cov\left(X_l(h_k)'W_-(h_k)\eps_{Z_k}, U_{1,-,p}(h_f) \right) = 0.
\]
\end{enumerate}
\end{lemma}

We set 
\[
C_{Z_j Z_k+,l} = \frac{\sigma^2_{Z_j Z_k+}}{f_+}
e_0'\Gamma_{l}^{-1}
\Psi_{j k,l,l}\Gamma_{l}^{-1}e_0,\quad
C_{Z_j Z_k-,l} = \frac{\sigma^2_{Z_j Z_k-}}{f_-}
e_0'\Gamma_{l}^{-1}
\Psi_{j k,l,l}\Gamma_{l}^{-1}e_0,
\]
\[
V_{Z_k+,l} = C_{Z_k Z_k+,l}
=\frac{\sigma^2_{Z_k+}}{f_+}e_0'\Gamma_{l}^{-1}
\Psi_{l}\Gamma_{l}^{-1}e_0,\quad
V_{Z_k-,l} = C_{Z_k Z_k-,l}
=\frac{\sigma^2_{Z_k-}}{f_-}e_0'\Gamma_{l}^{-1}
\Psi_{l}\Gamma_{l}^{-1}e_0,
\]
\[
V_{f,+,p} = f_+ e_1'\Gamma^{-1}_p\Psi_{f,p}\Gamma^{-1}_pe_1,\quad
V_{f,-,p} = f_- e_1'\Gamma^{-1}_p\Psi_{f,p}\Gamma^{-1}_pe_1.
\]

We further set
\[
\tilde \Delta \hat \mu_{Z_k+,l}(h_k) \equiv \Delta \hat \mu_{Z_k+,l}(h_k) - \sqrt{n}(h_k)^{3/2+l}B_{Z_k+,l,l+1}(h_k),
\]
\[
\tilde \Delta \hat \mu_{Z_k-,l}(h_k) \equiv \Delta \hat \mu_{Z_k-,l}(h_k) - \sqrt{n}(h_k)^{3/2+l}B_{Z_k-,l,l+1}(h_k),
\]
\[
\tilde \Delta \hat f_{+,p}(h_f) \equiv \Delta \hat f_{+,p}(h_f) - \sqrt{n}h^{p+1/2}_f B_{f,+,p,p+1}(h_f),
\]
\[
\tilde \Delta \hat f_{-,p}(h_f) \equiv \Delta \hat f_{-,p}(h_f) - \sqrt{n}h^{p+1/2}_f B_{f,-,p,p+1}(h_f),
\]

Using Lemmas \ref{lem:sa3.b} and \ref{lem:sa3.v}, the following provides an asymptotic distribution for the local polynomial estimators.

\begin{lemma}\label{lem:sa3.d}
Suppose Assumptions 1–2 hold with $S\geq l+2$ and $R \geq p+1$. If $n h_{min}\to\infty$, $nh_f^2 \to\infty$, $n (h_{max})^{2l+5}\to0$, and $n (h_f)^{2p+3} \to0$, then
\begin{equation}\label{eq:d1.5}
[\tilde \Delta \hat \mu_{Z_1+,l}(h_1),\ldots,\tilde \Delta \hat \mu_{Z_d+,l}(h_d),\tilde \Delta \hat f_{+,p}(h_f)]'
\dto N_{d+1}(0,V_{+,l,p}), 
\end{equation}
\begin{equation}
[\tilde \Delta \hat \mu_{Z_1-,l}(h_1),\ldots,\tilde \Delta \hat \mu_{Z_d-,l}(h_d),\tilde \Delta \hat f_{-,p}(h_f)]'
\dto N_{d+1}(0,V_{-,l,p}), 
\end{equation}
where 
\[
V_{+,l,p} = 
\begin{pmatrix}
    V_{Z_+,l} & 0\\
    0& V_{f,+,p}
\end{pmatrix},\quad V_{Z+,l}=\{C_{Z_j Z_k+,l}\}_{j,k=1,\ldots,d},
\]
\[
V_{-,l,p} = 
\begin{pmatrix}
    V_{Z_-,l} & 0\\
    0& V_{f,-,p}
\end{pmatrix},\quad V_{Z-,l}=\{C_{Z_j Z_k-,l}\}_{j,k=1,\ldots,d}.
\]
\end{lemma}


\begin{remark}
In Lemma \ref{lem:sa3.v} (b), the results for the estimators on the left and right sides are not analogous. The term
\[
\{- f_+^3 F_+ + f_+ f'_+ (F_+ - F_+^2) \} \Gamma_{p}(e_1e_0'+e_0e_1')\Gamma_{p}
\]
for the right-hand-side estimator changes to
\[
-\{- f_-^3 (F_- - 1) + f_- f'_- (F_- - F_-^2) \} H_p(-1)\Gamma_{p}(e_1e_0'+e_0e_1')\Gamma_{p}H_p(-1)
\]
for the left-hand-side estimator. Nevertheless, we observe the following.
\[
e_1'\Gamma_{p}^{-1}\Gamma_{p}(e_1e_0'+e_0e_1')\Gamma_{p}\Gamma_{p}^{-1}e_1
= e_1'(e_1e_0'+e_0e_1')e_1 =0,
\]
and the difference has no effect on the asymptotic variances $V_{f,+,p}$ and $V_{f,-,p}$.
We outline the proof of Lemma \ref{lem:sa3.v} (b) for the left-side estimator in Remark \ref{rem:sa3.v.b}.
\end{remark}

\subsection{Result for implementation}\label{sec:a.2.5}
As mentioned in the main text, we implemented bias correction by increasing the order of the local polynomial estimators constructed from the following (not bias-corrected) distributional approximation.
\begin{theorem}\label{thm:1'}
Suppose that Assumptions 1 and 2 hold, with $R \geq p+2$ and $S\geq l+3$. If $n h_{min}\to\infty$, $nh_f^2 \to\infty$, $n (h_{max})^{2l+5}\to0$, and $n (h_f)^{2p+3} \to0$, then:
\[
\begin{split}
\sqrt{n} [&\sqrt{h_{1}}(\hat\tau_{Z_1,l+1}(h_1)-\tau_{Z_1}),\ldots,\sqrt{h_{d}}(\hat \tau_{Z_d,l+1}(h_d)-\tau_{Z_d}),\sqrt{h_{f}}(\hat\tau_{f,p+1}(h_f)-\tau_f)]'\\
\dto& N_{d+1}(0,V_{l+1,p+1}).
\end{split}
\]
\end{theorem}

\subsection{Asymptotic covariance matrix and standard error estimation}\label{sec:a.3}
Using the asymptotic covariance matrices of $(\hat \mu_{Z_1+,l}(h_1),\ldots,\hat \mu_{Z_d+,l}(h_d),\hat f_{+,p}(h_f))'$ and $(\hat \mu_{Z_1-,l}(h_1),\ldots,\hat \mu_{Z_d-,l}(h_d),\hat f_{-,p}(h_f))'$, the asymptotic covariance matrix of $(\hat\tau_{Z_1,l}(h_1),\ldots,\hat\tau_{Z_d,l}(h_d),\hat\tau_{f,p}(h_f))'$ is 
\[
V_{l,p} = V_{+,l,p} + V_{-,l,p} =
\begin{pmatrix}
    V_{Z,l} & 0\\
    0& V_{f,p}
\end{pmatrix},
\]
where $V_{Z,l} = V_{Z+,l} + V_{Z-,l} =
\{C_{Z_j Z_k+,l}+C_{Z_j Z_k-,l}\}_{j,k=1,\ldots,d}$ and $V_{f,p} = V_{f,+,p} + V_{f,-,p}$.
We propose the following estimator for the asymptotic covariance matrix:
\[
\hat V_{l,p}(h) = \hat V_{+,l,p}(h) + \hat V_{-,l,p}(h) =
\begin{pmatrix}
    \hat V_{Z,l}(h) & 0\\
    0& \hat V_{f,p}(h_f)
\end{pmatrix},
\]
where $\hat V_{Z,l}(h) = \hat V_{Z+,l}(h) + \hat V_{Z-,l}(h) =
\{\hat C_{Z_j Z_k+,l}(h_j,h_k) + \hat C_{Z_j Z_k-,l}(h_j,h_k)\}_{j,k=1,\ldots,d}$ and $\hat V_{f,p}(h_f) = \hat V_{f,+,p}(h_f) + \hat V_{f,-,p}(h_f)$.
The diangonal elements of $\hat V_{Z,l}(h)$ are $\hat V_{Z_k,l}(h_k) = \hat V_{Z_k+,l}(h_k) + \hat V_{Z_k-,l}(h_k)$ for $k=1,\ldots,d$.

For $\hat C_{Z_j Z_k+,l}(h_j,h_k)$, $\hat C_{Z_j Z_k-,l}(h_j,h_k)$, $\hat V_{Z_k+,l}(h)$, and $\hat V_{Z_k-,l}(h)$, we follow an approach similar to the standard error estimators of \cite{calonicoRobustNonparametricConfidence2014a} based on nearest-neighbor estimation.
We define
\[
\hat C_{Z_j Z_k+,l}(h_j,h_k) = h_{jk}e_0'\Gamma_{+,l}^{-1}(h_j)
\hat \Psi_{Z_j Z_k+,l,l}(h_j,h_k)\Gamma_{+,l}^{-1}(h_k)e_0,
\]
\[
\hat C_{Z_j Z_k-,l}(h_j,h_k) = h_{jk}e_0'\Gamma_{-,l}^{-1}(h_j)
\hat \Psi_{Z_j Z_k-,l,l}(h_j,h_k)\Gamma_{-,l}^{-1}(h_k)e_0,
\]
\[
\hat V_{Z_k+,l}(h) = \hat C_{Z_k Z_k+,l}(h ,h) = h e_0'\Gamma_{+,l}^{-1}(h)
\hat \Psi_{Z_k+,l}(h)\Gamma_{+,l}^{-1}(h)e_0,
\]
\[
\hat V_{Z_k-,l}(h) = \hat C_{Z_k Z_k-,l}(h ,h) = h e_0'\Gamma_{-,l}^{-1}(h)
\hat \Psi_{Z_k-,l}(h)\Gamma_{-,l}^{-1}(h)e_0
\]
with
\[
\hat \Psi_{Z_j Z_k+,l,l}(h_j,h_k) = X_l(h_j)'W_+(h_j)\hat \Sigma_{Z_j Z_k+}W_+(h_k)X_l(h_k)/n,
\]
\[
\hat \Psi_{Z_j Z_k-,l,l}(h_j,h_k) = X_l(h_j)'W_-(h_j)\hat \Sigma_{Z_j Z_k-}W_-(h_k)X_l(h_k)/n,
\]
\[
\hat \Psi_{Z_k+,l}(h) = \hat \Psi_{Z_k Z_k+,l}(h,h) = X_l(h)'W_+(h)\hat \Sigma_{Z_k+}W_+(h)X_l(h)/n,
\]
\[
\hat \Psi_{Z_k-,l}(h) = \hat \Psi_{Z_k Z_k-,l}(h,h) = X_l(h)'W_-(h)\hat \Sigma_{Z_k-}W_-(h)X_l(h)/n,
\]
\[
\hat \Sigma_{Z_j Z_k+}=
\diag(\hat \sigma^2_{Z_j Z_k+}(X_1),\ldots,\hat \sigma^2_{Z_j Z_k+}(X_n)),
\]
\[
\hat \Sigma_{Z_j Z_k-}=
\diag(\hat \sigma^2_{Z_j Z_k-}(X_1),\ldots,\hat \sigma^2_{Z_j Z_k-}(X_n)),
\]
\[
\hat \Sigma_{Z_k+}=
\diag(\hat \sigma^2_{Z_k+}(X_1),\ldots,\hat \sigma^2_{Z_k+}(X_n)),
\]
\[
\hat \Sigma_{Z_k-}=
\diag(\hat \sigma^2_{Z_k-}(X_1),\ldots,\hat \sigma^2_{Z_k-}(X_n)),
\]
\[
\hat \sigma^2_{Z_j Z_k+}(X_i) = 1\{X_i\geq0\}\frac{M}{M+1}
\left(Z_{j,i} - \sum_{m=1}^M\frac{Z_{j,l^+_m(i)}}{M} \right)
\left(Z_{k,i} - \sum_{m=1}^M\frac{Z_{k,l^+_m(i)}}{M} \right),
\]
\[
\hat \sigma^2_{Z_j Z_k-}(X_i) = 1\{X_i<0\}\frac{M}{M+1}
\left(Z_{j,i} - \sum_{m=1}^M\frac{Z_{j,l^-_m(i)}}{M} \right)
\left(Z_{k,i} - \sum_{m=1}^M\frac{Z_{k,l^-_m(i)}}{M} \right),
\]
\[
\hat \sigma^2_{Z_k+}(X_i) = 1\{X_i\geq0\}\frac{M}{M+1}
\left(Z_{k,i} - \sum_{m=1}^M\frac{Z_{k,l^+_m(i)}}{M} \right)^2,
\]
\[
\hat \sigma^2_{Z_k-}(X_i) = 1\{X_i<0\}\frac{M}{M+1}
\left(Z_{k,i} - \sum_{m=1}^M\frac{Z_{k,l^-_m(i)}}{M} \right)^2,
\]
where $l^+_m(i)$ and $l^-_m(i)$ are the $m$th closest units to unit $i$ among $\{X_i:X_i\geq0\}$ and $\{X_i:X_i<0\}$ respectively, and $M$ denotes the number of neighbors.
\cite{calonicoRobustNonparametricConfidence2014a} shows that, if $\sigma^2_{Z_j Z_k}(x)$ is Lipschitz continuous on $(-\kappa_0,\kappa_0)$, then, for any choice of $M\in\N_+$, 
\[
\hat \Psi_{Z_j Z_k+,l,l}(h_j,h_k) = \Psi_{Z_j Z_k+,l,l}(h_j,h_k)
+ o_p(\min\{h_j^{-1},h_k^{-1}\}),
\]
\[
\hat \Psi_{Z_j Z_k-,l,l}(h_j,h_k) = \Psi_{Z_j Z_k-,l,l}(h_j,h_k)
+ o_p(\min\{h_j^{-1},h_k^{-1}\})
\]
hold, which leads to $\hat V_{Z,l}(h)\pto V_{Z,l}$ combined with Lemma \ref{lem:sa3.v}.

For $\hat V_{f,p}(h_f)$, we employ the jackknife-based standard error estimator from \cite{cattaneoSimpleLocalPolynomial2019}, which may be more robust than plug-in estimations in finite samples.
We define
\[
\hat V_{f+,p}(h_f)=h_f^{-1} e_1' \Gamma_{+,p}^{-1}(h_f)\hat \Psi_{f+,p}(h_f) \Gamma_{+,p}^{-1}(h_f)e_1,
\]
\[
\hat V_{f-,p}(h_f)=h_f^{-1} e_1' \Gamma_{-,p}^{-1}(h_f)\hat \Psi_{f-,p}(h_f) \Gamma_{-,p}^{-1}(h_f)e_1,
\]
with
\[
\begin{split}
\hat \Psi_{f+,p}(h_f) =& \frac{1}{n}\sumin
\left(\frac{1}{n-1}\sum_{j:j\neq i}^n
\hat \phi_{+,p,h_f}^*(X_i, X_j)\right)
\left(\frac{1}{n-1}\sum_{j:j\neq i}^n
\hat \phi_{+,p,h_f}^*(X_i, X_j)\right)'\\
&-\left(\frac{1}{n(n-1)}\sum_{i,j:i < j}^n
\hat \phi_{+,p,h_f}^*(X_i, X_j)\right)
\left(\frac{1}{n(n-1)}\sum_{i,j:i < j}^n
\hat \phi_{+,p,h_f}^*(X_i, X_j)\right)',
\end{split}
\]
\[
\begin{split}
\hat \Psi_{f-,p}(h_f) =& \frac{1}{n}\sumin
\left(\frac{1}{n-1}\sum_{j:j\neq i}^n
\hat \phi_{-,p,h_f}^*(X_i, X_j)\right)
\left(\frac{1}{n-1}\sum_{j:j\neq i}^n
\hat \phi_{-,p,h_f}^*(X_i, X_j)\right)'\\
&-\left(\frac{1}{n(n-1)}\sum_{i,j:i < j}^n
\hat \phi_{-,p,h_f}^*(X_i, X_j)\right)
\left(\frac{1}{n(n-1)}\sum_{i,j:i < j}^n
\hat \phi_{-,p,h_f}^*(X_i, X_j)\right)',
\end{split}
\]
\[
\hat \phi_{+,p,h_f}^*(X_i, X_j) = \hat \psi_{+,p,h_f}^*(X_i, X_j) + \hat \psi_{+,p,h_f}^*(X_j, X_i),
\]
\[
\hat \phi_{-,p,h_f}^*(X_i, X_j) = \hat \psi_{-,p,h_f}^*(X_i, X_j) + \hat \psi_{-,p,h_f}^*(X_j, X_i),
\]
\[
\hat \psi_{+,p,h_f}^*(X_i, X_j) = 1\{X_i \geq 0\}r_p(X_i/h_f)(1\{X_j \leq X_i\} - r_p(X_i)'\hat{\beta}_{f+,p}(h_{f})){K}_{h_f}(X_i),
\]
\[
\hat \psi_{-,p,h_f}^*(X_i, X_j) = 1\{X_i < 0\}r_p(X_i/h_f)(1\{X_j \leq X_i\} - r_p(X_i)'\hat{\beta}_{f-,p}(h_{f})){K}_{h_f}(X_i).
\]
These estimators can be motivated from another representations of $\Delta \hat f_{+,p}(h_f)$ and $\Delta \hat f_{-,p}(h_f)$.
For $\Delta \hat f_{+,p}(h_f)$, one can show that
\[
\Delta \hat f_{+,p}(h_f) = \sqrt{\frac{n}{h_{f}}} e_1' \Gamma_{+,p}^{-1}(h_f) U_{+,p}^*(h_f) + O_p\left(\frac{1}{\sqrt{n h_f}}\right),
\]
where
\[
U_{+,p}^*(h_f) = \frac{1}{n(n-1)}\sum_{i,j:i\neq j}^{n} \psi_{+,p,h_f}^*(X_i, X_j), 
\]
\[
\psi_{+,p,h_f}^*(X_i, X_j) = 1\{X_i \geq 0\}r_p(X_i/h_f)(1\{X_j \leq X_i\} - r_p(X_i)'\beta_{f+,p}){K}_{h_f}(X_i),
\]
and $\hat \Psi_{f+,p}(h_f)$ is constructed to approximate the asymptotic variance of the second-order U-statistic $U_{+,p}^*(h_f)$.

\subsection{Proof of the results}\label{sec:a.4}
The proofs in this section use several auxiliary results (Lemma \ref{lem:sa2}) collected in Appendix \ref{sec:a.5}.
For Lemmas \ref{lem:sa3.b}-\ref{lem:sa3.d}, we only provide proofs for the right-hand-side estimators ($\hat \mu_{Z_k+,l}(h_k)$, $\hat f_{+,p}(h_f)$), and the proofs of the left-hand-side estimators ($\hat \mu_{Z_k-,l}(h_k)$ and $\hat f_{-,p}(h_f)$) are analogous. 
Without loss of generality, we assume that $\kappa h_j < \kappa_0$ for $j=1,\ldots,d,f$ to bound the densities and error variances evaluated at $u h_f$ where $u \in (-\kappa,\kappa)$. 

\begin{proof}[Proof of Lemma \ref{lem:sa3.b}]
For part (a), applying Lemma S.A.3 (B) of \cite{calonicoRobustNonparametricConfidence2014a} to $E[\hat{\mu}_{Z_k+,l}(h_k)|\mX_n]$ for $k=1,\ldots,d$ yields:
\begin{equation}\label{eq:b1.1}
\begin{split}
&\frac{1}{\sqrt{n}}\Delta^B \hat \mu_{Z_k+,l}(h_k)\\
=&(h_k)^{3/2+l}\frac{\mu^{(l+1)}_{Z_k+}}{(l+1)!}
\mB_{+,0,l,l+1}(h_k)+(h_k)^{5/2+l}\frac{\mu^{(l+2)}_{Z_k+}}{(l+2)!}
\mB_{+,0,l,l+2}(h_k) + o_p((h_k)^{5/2+l}). 
\end{split}
\end{equation}

For $\Delta^{B_0} \hat f_{+,p}(h_f)$ in part (b), a derivation analogous to the proof of Lemma 2 in \cite{cattaneoSimpleLocalPolynomial2019} gives
\begin{equation}\label{eq:bf1.1}
\begin{split}
&\frac{1}{\sqrt{n}}\Delta^{B_0} \hat f_{+,p}(h_f)\\
=&h^{p+1/2}_f\frac{f_+^{(p)}}{(p+1)!} \mB_{+,1,p,p+1}(h_f)
+h^{p+3/2}_f\frac{f_+^{(p+1)}}{(p+2)!} \mB_{+,1,p,p+2}(h_f)
+o_p((h_f)^{p+3/2}).
\end{split}
\end{equation}
For $\Delta^{B_1} \hat f_{+,p}(h_f)$ in part (b), we observe that
\[
\begin{split}
&E\left[\frac{1}{n} |X_p(h_f)' W_+(h_f) (\iota - {F})| \right] \\
\leq& \frac{1}{n} \sumin E\left[ |r_p(X_i /h_f)' 1\{X_i \geq 0\} K_{h_f}(X_i) (1 - F(X_i))| \right] \\
\precsim& \int_0^\infty |r_p(u)| K(u) f(h_f u) = O(1),
\end{split}
\]
and using the Markov's inequality,
\[
\frac{1}{n} X_p(h_f)' W_+(h_f) (\iota - {F}) = O_p(1).
\]
Thus, we obtain $\Delta^{B_1} \hat f_{+,p}(h_f) = O_p(\frac{1}{\sqrt{n h_f}})$.

For part (c), we observe the following.
\begin{equation}
\begin{split}
E[|U_{2,+,p}(h_f)|^2]
\precsim& \frac{1}{n(n-1)} E[|\varphi_{+,p,h_f}(X_i,X_j)|^2] \\
\precsim& \frac{1}{n(n-1)} E[|\psi_{+,p,h_f}(X_i,X_j) - E[\psi_{+,p,h_f}(X_i,X_j)|X_j]|^2] \\
&+ \frac{1}{n(n-1)} E[|\psi_{+,p,h_f}(X_j,X_i) - E[\psi_{+,p,h_f}(X_j,X_i)|X_i]|^2 \\
\precsim& \frac{1}{n(n-1)} E[|\psi_{+,p,h_f}(X_i,X_j)|^2].
\end{split}
\end{equation}
A derivation analogous to the proof of Lemma 4 in \cite{cattaneoSimpleLocalPolynomial2019} provides
\begin{equation}
\begin{split}
&\frac{1}{n(n-1)} E[|\psi_{+,p,h_f}(X_i,X_j)|^2]\\
\precsim& \frac{1}{n(n-1)h_f} f_+ (F_+ - F_+^2) 
\int_0^\infty K(u)^2 r_p(u)'r_p(u) du
+ O(\frac{1}{n^{2}})
= O_p(\frac{1}{n^{2} h_f}),
\end{split}
\end{equation}
which implies that
$U_{2,+,p}(h_f) = O_p(\frac{1}{n \sqrt{h_f}})$ from Chebyshev's inequality.
Hence part (c) follows from $\sqrt{\frac{n}{h_f}} U_{2,+,p}(h_f) =  O_p(\frac{1}{\sqrt{n h_f^2}})$.

\end{proof}


\begin{proof}[Proof of Lemma \ref{lem:sa3.v}]
Observe that 
\[
\begin{split}
&\Cov\left(\sqrt{\frac{h_{j}}{n}}X_l(h_j)'W_+(h_j)\eps_{Z_j}, \sqrt{\frac{h_{k}}{n}}X_l(h_k)'W_+(h_k)\eps_{Z_k} | \mX_n \right) \\
=& \frac{h_{jk}}{n}X_l(h_j)'W_+(h_j)\Sigma_{Z_j Z_k}W_+(h_k)X_l(h_k)
= h_{jk} \Psi_{Z_j Z_k+,l,l}(h_j,h_k)
\end{split}
\]
and part (a) follows from Lemma \ref{lem:sa2} (a).

Note that
\[
\begin{split}
&E[\psi_{+,p,h_f}(X_j, X_i)|X_i]\\
=& \int_0^{\infty} r_p(u)(1(X_i \leq h_fu) - F(h_fu)) K(u)f(h_fu)du,
\end{split}
\]
A derivation analogous to the proof of Lemma 3 in \cite{cattaneoSimpleLocalPolynomial2019} provides
\begin{eqnarray*}
&& \Var\left(\int_0^{\infty} r_p(u)(1(X_i \leq h_fu) - F(h_fu)) K(u)f(h_fu)du \right) \\
&=& \int_0^{\infty}\int_0^{\infty} r_p(u)r_p(v)'K(u)K(v)f(h_fu)f(h_fv)
\{F(h_f(u\wedge v))) - F(h_fu)F(h_fv)\}dudv \\
&=& f_+^2 (F_+ - F_+^2) 
\int_0^{\infty}\int_0^{\infty} r_p(u)r_p(v)'K(u)K(v)dudv 
+ h_f f_+^3 \Psi_{f,p} \\
&&+ h_f \{- f_+^3 F_+ + f_+ f'_+ (F_+ - F_+^2) \} 
\int_0^{\infty}\int_0^{\infty} (u+v)r_p(u)r_p(v)'K(u)K(v)dudv 
+ o(1),
\end{eqnarray*}
Hence, part (b) follows from 
\[
\int_0^{\infty}r_p(u)K(u)du = \Gamma_p e_0 \quad\text{and}\quad
\int_0^{\infty}u r_p(u)K(u)du = \Gamma_p e_1.
\]

Observe that
\[
X_l(h_k)'W_+(h_k)\eps_{Z_k} = \sumin 1\{X_i\geq 0\}K_{h_k}(X_i)r_l(X_i/h_k)\eps_{Z_k,i}
\]
and part (c) follows from
\[
\begin{split}
&\Cov\left(\sumin 1\{X_i\geq 0\}K_{h_k}(X_i)r_l(X_i/h_k)\eps_{Z_k,i}, 
\sumin E[\psi_{+,p,h_f}(X_j, X_i)|X_i] \right) \\
=& \sumin E\left[1\{X_i\geq 0\}K_{h_k}(X_i)r_l(X_i/h_k)\eps_{Z_k,i} E[\psi_{+,p,h_f}(X_j, X_i)|X_i]\right]=0
\end{split}
\]
because $E[\eps_{Z_k,i}|\mX_n]=0$.
\end{proof}

\begin{remark}\label{rem:sa3.v.b}
We now outline the proof of Lemma \ref{lem:sa3.v} (b) for the left-hand-side estimator.
A derivation analogous to the proof for the right-side estimator yields
\begin{eqnarray*}
&& \Var\left(E[\psi_{-,p,h_f}(X_j, X_i)|X_i]\right) \\
&=& f_-^2 (F_- - F_-^2) 
H_p(-1)\int_0^{\infty}\int_0^{\infty} r_p(u)r_p(v)'K(u)K(v)dudvH_p(-1) \\
&&+ h_f f_-^3 
H_p(-1)\int_0^{\infty}\int_0^{\infty} \{(-u)\wedge(-v)\}r_p(u)r_p(v)'K(u)K(v)dudv
H_p(-1) \\
&&+ h_f \{- f_-^3 F_- + f_- f'_- (F_- - F_-^2) \} \\
&&\times H_p(-1)
\int_0^{\infty}\int_0^{\infty} -(u+v)r_p(u)r_p(v)'K(u)K(v)dudv
H_p(-1) + o(1),
\end{eqnarray*}
and the stated result follows from
\[
(-u)\wedge(-v) = \frac{1}{2}(-u-v-|u-v|) = u\wedge v - (u+v).
\]
\end{remark}


\begin{proof}[Proof of Lemma \ref{lem:sa3.d}]
We show that
\begin{equation}\label{eq:d24}
\sum_{k=1}^d t_k \tilde \Delta \hat \mu_{Z_k+,l}(h_k) + t_f \tilde \Delta \hat f_{+,p}
\dto N(0,\tilde{t}'V_{+,l,p}\tilde{t}),
\end{equation}
The stated result follows from the Cram\'{e}r-Wold theorem.

We set
\begin{equation}\label{eq:d3}
\xi_{1,n} = \sum_{k=1}^d t_k \Delta^M \hat \mu_{Z_k+,l}(h_k) + t_f \Delta^M \hat f_{+,p}(h_f)
\end{equation}
and
\begin{equation}\label{eq:d4.5}
\xi_{2,n} = \sum_{k=1}^d t_k \tilde \Delta^B \hat \mu_{Z_k+,l}(h_k) + t_f \tilde \Delta^B \hat f_{+,p}(h_f),
\end{equation}
where
\[
\tilde \Delta^B \hat \mu_{Z_k+,l}(h_k) \equiv \Delta^B \hat \mu_{Z_k+,l}(h_k) - \sqrt{n}(h_k)^{3/2+l}B_{Z_k+,l,l+1}(h_k)
\]
and
\[
\tilde \Delta^B \hat f_{+,p}(h_f) \equiv \Delta^B \hat f_{+,p}(h_f) - \sqrt{n}h^{p+1/2}_f B_{f,+,p,p+1}(h_f).
\]
To demonstrate \eqref{eq:d24}, we decompose the left-hand-side as follows.
\[
\sum_{k=1}^d t_k \tilde \Delta \hat \mu_{Z_k+,l}(h_k) + t_f \tilde \Delta \hat f_{+,p}(h_f)
=\xi_{1,n}+\xi_{2,n}.
\]
We show that $\xi_{1,n}\dto N(0,\tilde{t}'V_{+,l,p}\tilde{t})$ and $\xi_{2,n}=o_p(1)$.

First, we demonstrate that $\xi_{2,n}=o_p(1)$.
If $h_{max}\to0$, then from Lemma \ref{lem:sa3.b} (a),
\begin{equation}\label{eq:d5.5}
\sum_{k=1}^d t_k \tilde \Delta^B \hat \mu_{Z_k+,l}(h_k)
=O_p\left(\sqrt{n} \sum_{k=1}^d(h_k)^{5/2+l}\right)
=O_p\left(\sqrt{n} (h_{max})^{5/2+l} \right).
\end{equation}
From Lemma \ref{lem:sa3.v} (b), $E[|\sqrt{\frac{n}{h_{f}}} U_{1,+,p}(h_f)|^2] = O_(\frac{1}{h_f})$, implying \\
$\sqrt{\frac{n}{h_{f}}} U_{1,+,p}(h_f) = O_p(\frac{1}{\sqrt{h_f}})$ from Chebyshev's inequality.
Hence, we have \\
$\Delta^{B_3} \hat f_{+,p}(h_f) = O_p(\frac{1}{n \sqrt{h_f}})$ and
\[
\tilde \Delta^B \hat f_{+,p}(h_f)
= O_p\left(\sqrt{n} h_f^{3/2 + p}\right)
+ O_p\left(\frac{1}{ \sqrt{n h_f^2}}\right).
\]
Therefore, we have $\xi_{2,n}=o_p(1)$ for the bias.

Next, we demonstrate that $\xi_{1,n}\dto N(0,1)$.
Let
\[
\tilde \Delta^M \hat \mu_{Z_k+,l}(h_k) = \sqrt{\frac{h_{k}}{n}} f_+^{-1} e_0' \Gamma_{l}^{-1}X_l(h_k)'W_+(h_k)\eps_{Z_k}
\]
and
\[
\tilde \Delta^M \hat f_{+,p}(h_f) = \sqrt{\frac{n}{h_{f}}} f_+^{-1} e_1' \Gamma_{p}^{-1} U_{1,+,p}(h_f).
\]
Subsequently, we have 
\[
\Delta^M \hat \mu_{Z_k+,l}(h_k) = \tilde \Delta^M \hat \mu_{Z_k+,l}(h_k) + o_p(1)
\]
and
\[
\Delta^M \hat f_{+,p}(h_f) = \tilde \Delta^M \hat f_{+,p}(h_f) + O_p\left(\sqrt{h_f} + \frac{1}{ \sqrt{n h_f^2}}\right).
\]
To demonstrate this, from Lemma S.A.1 in \cite{calonicoRobustNonparametricConfidence2014a}, we obtain
\[
\Gamma_{+,p}(h_f) = f_+ \Gamma_{p} + O_p\left(h_f + \frac{1}{ \sqrt{n h_f}}\right),
\]
Hence, $|\Gamma_{+,p}(h_f) - f_+ \Gamma_{p}|<1$ with a probability approaching one.
Consequently, a well-known result for the matrix inverse yields
\[
|\Gamma_{+,p}^{-1}(h_f) - f_+^{-1} \Gamma_{p}^{-1}|
\leq \frac{|\Gamma_{+,p}(h_f) - f_+ \Gamma_{p}| |f_+^{-1} \Gamma_{p}^{-1}|^2}{1 - |\Gamma_{+,p}(h_f) - f_+ \Gamma_{p}| |f_+^{-1} \Gamma_{p}^{-1}|} = O_p\left(h_f + \frac{1}{ \sqrt{n h_f}}\right)
\]
with a probability approaching one. 
Hence, we obtain 
\[
\Gamma_{+,p}^{-1}(h_f) = f_+^{-1} \Gamma_{p}^{-1} + O_p\left(h_f + \frac{1}{ \sqrt{n h_f}}\right),
\]
and we have $\sqrt{\frac{h_{k}}{n}}X_l(h_k)'W_+(h_k)\eps_{Z_k}=O_p(1)$ and $\sqrt{\frac{n}{h_{f}}}U_{1,+,p}(h_f) = O_p\left(\frac{1}{ \sqrt{h_f}}\right)$ from Lemma \ref{lem:sa3.v} (a) and (b), respectively.
Let
\[
\tilde \xi_{1,n} = \sum_{k=1}^d t_k \tilde \Delta^M \hat \mu_{Z_k+,l}(h_k) + t_f \tilde \Delta^M \hat f_{+,p}(h_f).
\]
Subsequently, from the previous definitions and derivations, $\xi_{1,n} = \tilde \xi_{1,n} + o_p(1)$.

Thus, it remains to show that $\tilde\xi_{1,n}\dto N(0,\tilde{t}'V_{+,l,p}\tilde{t})$.
Notably, $\tilde \xi_{1,n}$ can be represented as $\tilde \xi_{1,n} = 
\sumin \hat \omega_i$ with
\[
\hat \omega_i \equiv f_+^{-1} \left\{ e_0' \Gamma_{l}^{-1} \sum_{k=1}^d t_k \hat\omega_{k,i} \eps_{Z_k,i} + t_f e_1' \Gamma_{p}^{-1} \hat\omega_{f,i} \right\},
\]
where
\[
\hat\omega_{k,i} \equiv \sqrt{\frac{h_k}{n}} 1\{X_i\geq 0\}K_{h_k}(X_i)r_l(X_i/h_k)
\]
and
\[
\hat\omega_{f,i} \equiv \frac{1}{\sqrt{n h_f}} E[\psi_{+,p,h_f}(X_j, X_i)|X_i].
\]
$\{\hat \omega_i\}:1\leq i\leq n$ denotes a triangular array of independent random variables;
We show that 
\begin{equation}\label{eq:d16}
E[\tilde\xi_{1,n}]=0,
\end{equation}
\begin{equation}\label{eq:d19}
\Var(\tilde\xi_{1,n}) = \tilde{t}'V_{+,l,p}\tilde{t} + o(1),
\end{equation}
and
\begin{equation}\label{eq:d20'}
\sumin E\left[\left|\hat \omega_i \right|^4\right] = o(1).
\end{equation}
The Lyapunov Condition---a well-known sufficient condition for the Lindeberg condition---is satisfied by \eqref{eq:d20'}.
Therefore, from \eqref{eq:d16}, \eqref{eq:d19}, and \eqref{eq:d20'}, applying the Lindeberg-Feller central limit theorem yields that $\tilde\xi_{1,n}\dto N(0,\tilde{t}'V_{+,l,p}\tilde{t})$ (see, for example, \cite{durrett2019probability}).

First, from the definition, \eqref{eq:d16} follows from $E[\eps_{Z_k,i}|\mX_n]=0$ and \\$E[\psi_{+,p,h_f}(X_j, X_i)]=0$.
Next, from Lemmas \ref{lem:sa3.v} (a) and (b), we have 
\[
\begin{split}
\Var(\tilde\xi_{1,n}) =& \Var\left(\sum_{k=1}^d t_k \tilde \Delta^M \hat \mu_{Z_k+,l}(h_k) \right) + \Var\left( t_f \tilde \Delta^M \hat f_{+,p}(h_f) \right) \\
=& t'V_{Z+,l}t + t_f^2 V_{f,+,p} + o(1).
\end{split}
\]
Hence, we obtain \eqref{eq:d19}.
Finally, similar to the proof of Lemma S.A.3 (D) in \cite{calonicoRobustNonparametricConfidence2014a}, we obtain 
\begin{align*}
&\sumin E[|\hat\omega_{k,i}\eps_{Z_k,i}|^4]
\precsim \sumin E[|\hat\omega_{k,i}|^4] \\
\precsim& \frac{1}{n h_k} \int_{0}^{\infty} K(u)^4 |r_p(u)|^4 f(h_k u) du 
= O\left(\frac{1}{n h_k} \right).
\end{align*}
The first inequality holds because, for two random variables $W$ and $Y$, $E[|W+Y|^4|X_i=x]\leq8\{E[W^4|X_i=x]+E[Y^4|X_i=x]\}$ holds,
\begin{equation}\label{eq:d21}
E[\eps_{Z_k,i}^4|X_i=x]\leq8\{E[|Z_{k,i}|^4|X_i=x]+\mu(x)^4\}.
\end{equation}
Hence, $E[\eps_{Z_{k,i}}^4|X_i=x]$ is bounded, based on Assumption 1.
Similar to the proof of Lemma 3 in \cite{cattaneoSimpleLocalPolynomial2019}, we obtain 
\begin{align*}
&\sumin E[|\hat\omega_{f,i}|^4]\\
\precsim& \int_{0}^{\infty}\int_{0}^{\infty}\int_{0}^{\infty}\int_{0}^{\infty}
r_p(u_1)r_p(u_2)'r_p(u_3)r_p(u_4)'
\prod_{j=1}^4 \left[K(u_j) f(h_f u_j) \right] du_1du_2du_3du_4 \\
=& O\left(\frac{1}{n h_f^2} \right).
\end{align*}
Observe that
\begin{equation}\label{eq:d20}
\begin{split}
& \sumin E[|\hat\omega_{f,i}|^4]
\precsim \sumin \sum_{k=1}^d E\left[\left| t_k \hat\omega_{k,i}\eps_{Z_k,i}\right|^4\right] + \sumin E\left[\left|t_f \hat\omega_{f,i}\right|^4\right]\\
\precsim&\sum_{k=1}^d \sumin E\left[\left|
\hat\omega_{k,i}\eps_{Z_k,i}\right|^4\right] +  \sum_{i=1}^n E[|\hat\omega_{f,i}|^4]\\
\precsim& O\left(\frac{1}{n h_{min}} \right) + O\left(\frac{1}{n h_f^2}\right).
\end{split}
\end{equation}
The first inequality in \eqref{eq:d20} holds because, for two random variables $W$ and $Y$, $E[|W+Y|^4]\leq8\{E[W^4]+E[Y^4]\}$ holds iteratively.
Therefore, \eqref{eq:d20'} follows.
\end{proof}

\begin{proof}[Proof of Theorem 3.1]
The stated result follows from the same argument given in the proof of Lemma \ref{lem:sa3.d} using the analogs of Lemmas \ref{lem:sa3.b} and \ref{lem:sa3.v}, but is now applied to the estimator $(\hat\tau_{Z_1,l}(h_1),\ldots,\hat\tau_{Z_d,l}(h_d),\hat\tau_{f,p}(h_f))'$. 
\end{proof}

\begin{proof}[Proof of Proposition 3.1]
Let $c_{l, p}(\alpha)$ be the $\alpha$-quantile of $\| N_{d+1}(0,V_{l,p}^*) \|_{2}^2$.
Theorem 3.1 implies that $\tilde \chi^2_{l,p}(h) \dto \| N_{d+1}(0,V_{l,p}^*) \|_{2}^2$ under $H_0$.
To prove Proposition 3.1 under $H_0$, it suffices to show that $\hat c_{l,p}(\alpha) \pto c_{l,p}(\alpha)$ because this implies that $\tilde \chi^2_{l,p}(h) - \hat c_{l,p}(\alpha) \dto \| N_{d+1}(0,V_{l,p}^*) \|_{2}^2 - c_{l,p}(\alpha)$ under $H_0$.
Let $F(\cdot; V)$ be the distribution function of $\| N_{d+1}(0,V) \|_{2}^2$.
Suppose that $V_m$ and $V$ are positive definite matrices for $m$ and $V_m \to V$.
Subsequently, for all $t \in \mathbb{R}$, it follows from the dominated convergence theorem that
\begin{eqnarray*}
F(t;V_m) &=& \int \cdots \int 1 \left\{ u_1^2 + \cdots + u_{d+1}^2 \leq t \right\} \phi(u_1, \ldots , u_{d+1};V_m)  du_1 \cdots du_{d+1} \nonumber \\
&\to & \int \cdots \int 1 \left\{ u_1^2 + \cdots + u_{d+1}^2 \leq t \right\} \phi(u_1, \ldots , u_{d+1};V) du_1 \cdots du_{d+1} = F(t;V), \nonumber
\end{eqnarray*}
where $\phi(u_1, \ldots , u_{d+1};V)$ denotes the joint density function of $N_{d+1}(0,V)$.
As $F(t;V)$ is continuous with respect to $t$, $F(\cdot;V_m)$ converges uniformly to $F(\cdot;V)$, according to Polya's theorem.
As discussed in Section 3.9.4.2, in \cite{vaart1996weak}, the inverse map is continuous.
Thus, $F^{-1}(\alpha;V_m) \to F^{-1}(\alpha;V)$ for any $\alpha \in (0,1)$.
Hence, $F^{-1}(t;V)$ is continuous with respect to $V$, and we obtain $\hat c_{l,p}(\alpha) \pto c_{l,p}(\alpha)$ from the continuous-mapping theorem.

To prove Proposition 3.1 under $H_1$, suppose that, without loss of generality, $\tau_{Z_k} \neq 0$.
The proof for the case $\tau_f \neq 0$ follows from an analogous argument.
Observe that
\[
P\left( \frac{n h_{k}\hat \tau_{Z_k,l}^2(h_k)}{\hat V_{Z_k,l}(h_k)} \geq \hat c_{l,p}(\alpha) \right) 
\leq P(\tilde \chi^2_{l,p}(h) \geq \hat c_{l,p}(\alpha))
\]
and
\[
\begin{split}
&P\left( \frac{n h_{k}\hat \tau_{Z_k,l}^2(h_k)}{\hat V_{Z_k,l}(h_k)} \geq \hat c_{l,p}(\alpha) \right) \\
=& P\left( \frac{n h_{k} (\hat \tau_{Z_k,l}(h_k) - \tau_{Z_k})^2}{\hat V_{Z_k,l}(h_k)} - \hat c_{l,p}(\alpha) 
\geq - \frac{n h_{k} (2 \hat \tau_{Z_k,l}(h_k) - \tau_{Z_k}) \tau_{Z_k}}{\hat V_{Z_k,l}(h_k)} \right).
\end{split}
\]
Because we have
\[
\frac{n h_{k} (\hat \tau_{Z_k,l}(h_k) - \tau_{Z_k})^2}{\hat V_{Z_k,l}(h_K)} - \hat c_{l,p}(\alpha)
\dto || N(0, 1) ||_2^2 - c_{l,p}(\alpha),
\]
\[
\frac{(2 \hat \tau_{Z_k,l}(h_k) - \tau_{Z_k}) \tau_{Z_k}}{\hat V_{Z_k,l}(h_k)} \pto \frac{2 \tau_{Z_k}^2}{V_{Z_k,l}},
\]
and $n h_{k} \to \infty$, we obtain
\[
\lim_{n\to \infty} P\left( \frac{n h_{k}\hat \tau_{Z_k,l}^2(h_k)}{\hat V_{Z_k,l}(h_k)} \geq \hat c_{l,p}(\alpha) \right) = 1.
\]
Hence, the stated result under $H_1$ follows.
\end{proof}

\begin{proof}[Proof of Proposition 3.2]
Let $m_{l, p}(\alpha)$ be the $\alpha$-quantile of $\| N_{d+1}(0,V_{l,p}^*) \|_{\infty}^2$.
Then, $\hat m_{l,p}(\alpha) \pto m_{l,p}(\alpha)$ follows from the same argument as the proof of  Proposition 3.1.
Hence, the stated result under $H_0$ follows.

To prove Proposition 3.2 under $H_1$, suppose that, without loss of generality, $\tau_{Z_k} \neq 0$.
The proof for the case $\tau_f \neq 0$ follows from an analogous argument.
Observe that
\[
P\left( \frac{n h_{k}\hat \tau_{Z_k,l}^2(h_k)}{\hat V_{Z_k,l}(h_k)} \geq \hat m_{l,p}(\alpha) \right) 
\leq P(\hat M_{l,p}(h) \geq \hat m_{l,p}(\alpha)).
\]
By an analogous argument as in the proof of Proposition 3.1, we obtain
\[
\lim_{n\to \infty} P\left( \frac{n h_{k}\hat \tau_{Z_k,l}^2(h_k)}{\hat V_{Z_k,l}(h_k)} \geq \hat m_{l,p}(\alpha) \right) = 1.
\]
Hence, the stated result under $H_1$ follows.
\end{proof}

\begin{proof}[Proof of Proposition 5.1]
We first prove the statement under $H_0$.
First, suppose that $\max\{\tau_f,\tau_{Z_1},\ldots,\tau_{Z_d}\} \geq \varepsilon$.
Then, for the case of $\max\{\tau_f,\tau_{Z_1},\ldots,\tau_{Z_d}\} = \tau_{Z_k}$, observe that
\[
\hat \psi_{1,l,p}(h) \geq \frac{\sqrt{n}\sqrt{h_{k}}(\hat\tau_{Z_k,l}(h_k)-\varepsilon)}{\hat V_{Z_k,l}(h_k)}
\geq \frac{\sqrt{n}\sqrt{h_{k}}(\hat\tau_{Z_k,l}(h_k)-\tau_{Z_k})}{\hat V_{Z_k,l}(h_k)}.
\]
Hence, we have
\[
P(\hat \psi_{1,l,p}(h) \leq z_{\alpha/2})
\leq P\left( \frac{\sqrt{n}\sqrt{h_{k}}(\hat\tau_{Z_k,l}(h_k)-\tau_{Z_k})}{\hat V_{Z_k,l}(h_k)} \leq z_{\alpha/2} \right) \to \frac{\alpha}{2},
\]
The same result holds for the case of $\max\{\tau_f,\tau_{Z_1},\ldots,\tau_{Z_d}\} = \tau_f$.

Next, suppose that $\min\{\tau_f,\tau_{Z_1},\ldots,\tau_{Z_d}\} \leq -\varepsilon$.
Then, for the case of $\min\{\tau_f,\tau_{Z_1},\ldots,\tau_{Z_d}\} = \tau_{Z_k}$, observe that
\[
\hat \psi_{2,l,p}(h) \leq \frac{\sqrt{n}\sqrt{h_{k}}(\hat\tau_{Z_k,l}(h_k)+\varepsilon)}{\hat V_{Z_k,l}(h_k)}
\leq \frac{\sqrt{n}\sqrt{h_{k}}(\hat\tau_{Z_k,l}(h_k)-\tau_{Z_k})}{\hat V_{Z_k,l}(h_k)}.
\]
Hence, we have
\[
P(\hat \psi_{2,l,p}(h) \geq z_{1 - \alpha/2})
\leq P\left( \frac{\sqrt{n}\sqrt{h_{k}}(\hat\tau_{Z_k,l}(h_k) - \tau_{Z_k})}{\hat V_{Z_k,l}(h_k)} \geq z_{1 - \alpha/2} \right) \to \frac{\alpha}{2},
\]
The same result holds for the case of $\min\{\tau_f,\tau_{Z_1},\ldots,\tau_{Z_d}\} = \tau_f$.

Therefore, under $H_0$, we have
\[
\limsup_{n\to\infty} P\left( \hat \psi_{1,l,p}(h) \leq z_{\alpha/2} \text{ or }
\hat \psi_{2,l,p}(h) \geq z_{1 - \alpha/2} \right) \leq \alpha.
\]

We next prove the statement under $H_1$.
Observe that
\[
\begin{split}
&P\left( \frac{\sqrt{n}\sqrt{h_{k}}(\hat\tau_{Z_k,l}(h_k) - \varepsilon)}{\hat V_{Z_k,l}(h_k)} \geq z_{\alpha/2} \right) \\
=& P\left( \frac{\sqrt{n}\sqrt{h_{k}}(\hat\tau_{Z_k,l}(h_k) - \tau_{Z_k})}{\hat V_{Z_k,l}(h_k)} - z_{\alpha/2} \geq \frac{\sqrt{n}\sqrt{h_{k}}(\varepsilon - \tau_{Z_k})}{\hat V_{Z_k,l}(h_k)} \right).
\end{split}
\]
Because we have
\[
\frac{\sqrt{n}\sqrt{h_{k}}(\hat\tau_{Z_k,l}(h_k) - \tau_{Z_k})}{\hat V_{Z_k,l}(h_k)} - z_{\alpha/2}
\dto N(0,1) - z_{\alpha/2},
\]
\[
\frac{\varepsilon - \tau_{Z_k}}{\hat V_{Z_k,l}(h_k)} 
\pto \frac{\varepsilon - \tau_{Z_k}}{V_{Z_k,l}},
\]
and $n h_k \to \infty$, we obtain
\[
\lim_{n\to\infty} P\left( \frac{\sqrt{n}\sqrt{h_{k}}(\hat\tau_{Z_k,l}(h_k) - \varepsilon)}{\hat V_{Z_k,l}(h_k)} \geq z_{\alpha/2} \right) = 0.
\]
Similar result holds for $\hat\tau_{f,p+1}(h_f)$.
Therefore, we obtain
\[
\begin{split}
&P(\hat \psi_{1,l,p}(h) \geq z_{\alpha/2}) \\
\leq& \sum_{k=1}^d P\left( \frac{\sqrt{n}\sqrt{h_{k}}(\hat\tau_{Z_k,l}(h_k) - \varepsilon)}{\hat V_{Z_k,l}(h_k)} \geq z_{\alpha/2} \right)
+ P\left( \frac{\sqrt{n}\sqrt{h_{f}}(\hat\tau_{f,p}(h_f)-\varepsilon)}{\hat V_{f,p}(h_f)} \geq z_{\alpha/2} \right) \\
\to& 0,
\end{split}
\]
which implies
\[
\lim_{n\to\infty} P\left( \hat \psi_{1,l,p}(h) \geq z_{\alpha/2} \text{ and }
\hat \psi_{2,l,p}(h) \leq z_{1 - \alpha/2} \right)  = 0.
\]

\end{proof}

\begin{proof}[Proof of Theorem \ref{thm:1'}]
Under the same assumption as in Theorem \ref{thm:1'}, the asymptotic distribution for the local polynomial estimators, where the order is increased instead of eliminating the bias term, is as follows.
\[
[ \Delta \hat \mu_{Z_1+,l+1}(h_1),\ldots, \Delta \hat \mu_{Z_d+,l+1}(h_d),\Delta \hat f_{+,p+1}(h_f)]'
\dto N_{d+1}(0,V_{+,l+1,p+1}), 
\]
\[
[ \Delta \hat \mu_{Z_1-,l+1}(h_1),\ldots, \Delta \hat \mu_{Z_d-,l+1}(h_d),\Delta \hat f_{-,p+1}(h_f)]'
\dto N_{d+1}(0,V_{-,l+1,p+1}).
\]
This result and the stated result follow from a similar argument given in the proof of Lemma \ref{lem:sa3.d} using the analogs of Lemmas \ref{lem:sa3.b} and \ref{lem:sa3.v}.
\end{proof}


\subsection{Auxiliary Lemma}\label{sec:a.5}
The following lemma establishes convergence in the probability of the sample matrix $\Psi_{Z_j Z_k+,p,q}(h_j,h_k)$ to its population counterpart and characterizes this limit. 
\begin{lemma}\label{lem:sa2}
Suppose that Assumptions 1-2 hold.
If $h_{jk}\to0$ and $n h_{jk}\to\infty$, then
\begin{equation}\label{eq:20}
h_{jk}\Psi_{Z_j Z_k+,p,q}(h_j,h_k)=
\tilde\Psi_{Z_j Z_k+,p,q}(h_j,h_k)+o_p(1),
\end{equation}
\begin{equation}
h_{jk}\Psi_{Z_j Z_k-,p,q}(h_j,h_k)=
H_p(-1)\tilde\Psi_{Z_j Z_k-,p,q}(h_j,h_k)H_q(-1)+o_p(1),
\end{equation}
and
\begin{equation}\label{eq:21}
\tilde\Psi_{Z_j Z_k+,p,q}(h_j,h_k)=
\sigma^2_{Z_j Z_k+}f_+\Psi_{jk,p,q}+o(1),
\end{equation}
\begin{equation}
\tilde\Psi_{Z_j Z_k-,p,q}(h_j,h_k)=
\sigma^2_{Z_j Z_k-}f_-\Psi_{jk,p,q}+o(1).
\end{equation}

\end{lemma}

\begin{remark}\label{rem:0}
In the proof of Lemma \ref{lem:sa2}, we use the compactness of the support of $K(\cdot)$ in the derivation of \eqref{eq:21}.
We obtain \eqref{eq:21} without assuming compactness in support of $K(\cdot)$.
If $\rho_{jk}\to c_{jk}\in(0,\infty)$, \eqref{eq:21} follows:
continuity of $K(u)$, $r_p(u)$, $\sigma^4_{Z_j Z_k}(u)$, and $f(u)$.
For the other cases, we observe the following.
\begin{equation}\label{eq:24}
\begin{split}
&\tilde \Psi_{Z_j Z_k+,p,q}(h_j,h_k)\\
=&\frac{1}{h_{jk}}\int_0^\infty K\left(\frac{x}{h_j}\right)K\left(\frac{x}{h_k}\right)
r_p\left(\frac{x}{h_j}\right)r_q\left(\frac{x}{h_k}\right)'
\sigma^2_{Z_j Z_k}(x)f(x)dx\\
=&\frac{m_{jk}}{h_{jk}}\int_0^\infty K\left(\frac{m_{jk}u}{h_j}\right)
K\left(\frac{m_{jk}u}{h_k}\right)
r_p\left(\frac{m_{jk}u}{h_j}\right)
r_q\left(\frac{m_{jk}u}{h_k}\right)'
\sigma^2_{Z_j Z_k}(um_{jk})f(um_{jk})du.
\end{split}
\end{equation}
Suppose that $\rho_{jk,n}\to 0$. Subsequently, $h_{k,n}<h_{j,n}$; thus, $m_{jk,n}/h_{jk,n}=\rho_{jk,n}^{1/2}$ holds for a sufficiently large $n$.
Hence, $m_{jk,n}/h_{jk,n}\to0$, and \eqref{eq:21} follows from \eqref{eq:24}. The case of $\rho_{jk,n}^{-1}\to 0$ follows from an analogous argument.
\end{remark}

\begin{proof}[Proof of Lemma \ref{lem:sa2}]
First, for $\Psi_{Z_j Z_k+,p,q}(h_j,h_k)$, a change in the variables yields
\begin{equation}\label{eq:22}
\begin{split}
&E[h_{jk}\Psi_{Z_j Z_k+,p,q}(h_j,h_k)]\\
=&E\left[
\frac{1}{nh_{jk}}\sumin 1(X_i\geq 0)K\left(\frac{X_i}{h_j}\right)K\left(\frac{X_i}{h_k}\right)
r_p\left(\frac{X_i}{h_j}\right)r_q\left(\frac{X_i}{h_k}\right)'
\sigma^2_{Z_j Z_k}(X_i)
\right]\\
=&\frac{1}{h_{jk}}\int_0^\infty K\left(\frac{x}{h_j}\right)K\left(\frac{x}{h_k}\right)
r_p\left(\frac{x}{h_j}\right)r_q\left(\frac{x}{h_k}\right)'
\sigma^2_{Z_j Z_k}(x)f(x)dx\\
=&\int_0^\infty K\left(\frac{h_{jk}u}{h_j}\right)
K\left(\frac{h_{jk}u}{h_k}\right)
r_p\left(\frac{h_{jk}u}{h_j}\right)
r_q\left(\frac{h_{jk}u}{h_k}\right)'
\sigma^2_{Z_j Z_k}(uh_{jk})f(uh_{jk})du\\
=&\tilde \Psi_{Z_j Z_k+,p,q}(h_j,h_k).
\end{split}
\end{equation}
Let 
\begin{equation}\label{eq:22.1}
\psi_{j,k,i}=1(X_i\geq 0)K\left(\frac{X_i}{h_j}\right)K\left(\frac{X_i}{h_k}\right)
r_p\left(\frac{X_i}{h_j}\right)r_q\left(\frac{X_i}{h_k}\right)'
\sigma^2_{Z_j Z_k}(X_i),
\end{equation}
where $\psi_{j,k,i}=\{\psi_{j,k,l,m,i}\}_{l=1,\ldots,(p+1)}^
{m=1,\ldots,(q+1)}$ is a $(p+1)\times(q+1)$-matrix.
Additionally, let $\psi_{j,k,l,m}=\sumin \psi_{j,k,l,m,i}$ and
$\psi_{j,k}=\sumin \psi_{j,k,i}$.
Observe that 
\begin{equation}\label{eq:22.2}
\begin{split}
h_{jk}^2 & E[|\Psi_{Z_j Z_k+,p,q}(h_j,h_k)-
E[\Psi_{Z_j Z_k+,p,q}(h_j,h_k)]|^2]\\
=&\frac{1}{n^2h_{jk}^2}E[|\psi_{j,k}-E[\psi_{j,k}]|^2]
=\frac{1}{n^2h_{jk}^2}\sum_{l=1}^{p+1}\sum_{m=1}^{q+1}
\Var(\psi_{j,k,l,m})\\
=&\frac{1}{n^2h_{jk}^2}\sum_{l=1}^{p+1}\sum_{m=1}^{q+1}
\sumin \Var(\psi_{j,k,l,m,i})\\
\leq&\frac{1}{n^2h_{jk}^2}\sum_{l=1}^{p+1}\sum_{m=1}^{q+1}
\sumin E[\psi_{j,k,l,m,i}^2]
=\frac{1}{n^2h_{jk}^2}\sumin E[|\psi_{j,k,i}|^2]
\end{split}
\end{equation}
and
\begin{equation}\label{eq:23}
\begin{split}
\frac{1}{n^2h_{jk}^2}\sumin E[|\psi_{j,k,i}|^2] &\\
=\frac{1}{nh_{jk}^2}\int_0^\infty K\left(\frac{x}{h_j}\right)^2 K&\left(\frac{x}{h_k}\right)^2
\left|r_p\left(\frac{x}{h_j}\right)\right|^2\left|r_q\left(\frac{x}{h_k}\right)\right|^2
\sigma^4_{Z_j Z_k}(x)f(x)dx\\
\leq\frac{1}{nm_{jk}^2}\int_0^\infty K\left(\frac{x}{h_j}\right)^2 K&\left(\frac{x}{h_k}\right)^2
 \left|r_p\left(\frac{x}{h_j}\right)\right|^2\left|r_q\left(\frac{x}{h_k}\right)\right|^2
\sigma^4_{Z_j Z_k}(x)f(x)dx\\
=\frac{1}{nm_{jk}}\int_0^\infty K\left(\frac{m_{jk}u}{h_j}\right)^2
&K\left(\frac{m_{jk}u}{h_k}\right)^2 \\
& \left|r_p\left(\frac{m_{jk}u}{h_j}\right)\right|^2
\left|r_q\left(\frac{m_{jk}u}{h_k}\right)\right|^2
\sigma^4_{Z_j Z_k}(um_{jk})f(um_{jk})du\\
=O\left(\frac{1}{nm_{jk}}\right)=o(1).&
\end{split}
\end{equation}
Thus, \eqref{eq:20} follows the Markov inequality.
Second, \eqref{eq:21} follows by
continuity of $K(u)$, $r_p(u)$, $\sigma^4_{Z_j Z_k}(u)$, and $f(u)$, and the compactness of the support of $K(\cdot)$.
\end{proof}


\section{Details on the search criterion} \label{sec:a_crit}
For the meta-analysis, we analyzed RD studies using diagnostic tests. There are two widely cited methodological papers, (\citealp{mccraryManipulationRunningVariable2008} and \citealp{leeRandomizedExperimentsNonrandom2008}), and two widely cited survey papers, (\citealp{Imbens_Lemieux_2008} and \citealp{leeRegressionDiscontinuityDesigns2010}). We collected the citations of $2,697$ unique papers on November 5, 2021, from the Web of Science. \footnote{We used the Web of Science to limit our focus on published papers.} Among $2,697$ papers, we limited our focus to publications from the top five journals in economics (\textit{American Economic Review, Econometrica, Journal of Political Economy, Quarterly Journal of Economics, Review of Economic Studies} in alphabetical order.). Among $98$ in the top five publications, $60$ papers reported at least one diagnostic test to validate their empirical analyses of RD designs. Furthermore, we removed one paper which meant to report the implausibility of an existing design by demonstrating the rejected diagnostic tests. We excluded surveys, theoretical contributions, and other uses of similar tests in manipulation detection or kink designs. Furthermore, we limited our focus to the density and balance or placebo tests in their \textit{standard} procedures, excluding placebo or balance tests for predicted variables from covariates. We identified five studies that incorporated a joint test for multiple testing problems. However, we did not include these joint tests because all but one study did not incorporate the nonparametric nature of the RD estimates. The only considered study \citep{Fort_Ichino_Zanella_2020} used \cite{canayApproximatePermutationTests2018}.

From these $59$ papers, we collected the balance, placebo, or density test results that appeared to be their \textit{main} specifications. In practice, many researchers have run multiple versions of these tests using different bandwidths, kernels, and specifications. We collected the total number of tests separately, but our numerical analysis was limited to the main specifications. We computed p-values from the reported statistics when only the test statistics were reported.

\section{Details on the simulation data generating process}\label{sec:c}

We conducted $3,000$ replications to generate a random sample. 
\[
\{(X^+_i,X^-_i,\tilde Z_{1,i},\ldots,\tilde Z_{d,i},U_i)':i=1,\ldots,n\}
\]
with size $n=500, 1000$, $X^+_i\sim tN(0,0.12^2;[0,1])$ with $tN(\mu,\sigma^2;[0,1])$ denotes a truncated normal distribution with mean $\mu$ and variance $\sigma^2$, lying within the interval $[0,1]$, $X^-_i\deq -X^+_i$. 
$U_i\sim Unif[0,1]$ with $Unif[0,1]$ denoting a uniform distribution on the interval $[0,1]$, $(\tilde Z_{1i},\ldots,\tilde Z_{di})'\sim N_d(0, \tilde \Sigma)$ with $\tilde \Sigma$ denoting a correlation matrix where each $(j,k)$ entry is $1>\rho\geq0$ for $j\neq k$,
and $(X^+_i,X^-_i)$, $(\tilde Z_{1,i},\ldots,\tilde Z_{d,i})$, and $U_i$ are mutually independent.
The running variable $X_i$ is generated as follows.
\[
X_i = (1-M_i)X^-_i + M_i X^+_i,
\]
where $M_i = 1\{U_i\leq \bar{p}\}$ with $\bar{p}\geq0.5$.
For pre-treatment covariates $Z_i=(Z_{1,i},\ldots,Z_{d,i})'$, we consider two specifications: first, one of the $d$ covariates sees a jump, and second, all the $d$ covariates see a jump, but each jump size is divided by $d$.
For the first specification, each $Z_{k,i}, k=1,\ldots,d$ is generated as follows.
\[
Z_{k,i} = 
\begin{cases}
\lambda(X_i) + \tilde Z_{k,i} & \text{for }k=1,\ldots,d-1\\
\lambda(X_i) + a M_i + \tilde Z_{k,i} & \text{for }k=d\text{ with }a\geq0,
\end{cases}
\]
where the functional form of $\lambda(x)$ is defined as
\begin{equation*}
 \lambda(x) = 0.48 + 
 \begin{cases}
 0.84x - 3.00 x^2 + 7.99 x^3 - 9.01 x^4 + 3.56 x^5, & \mbox{ if } x \geq 0\\
 1.27x + 7.18 x^2 + 20.21 x^3 + 21.54 x^4 + 7.33 x^5, & \mbox{ otherwise}.
 \end{cases}
\end{equation*}
For the second specification, each $Z_{k,i}, k=1,\ldots,d$ is generated as follows.
\[
Z_{k,i} = \lambda(X_i) + a/d M_i + \tilde Z_{k,i}.
\]
Let $f^*(x)$ be the density of $X^+_i$, 
\[
f^*(x) = \frac{\phi(x/0.12)}{0.12(\Phi(1/0.12) - \Phi(0))},
\]
where $\phi(x)$ and $\Phi(x)$ are the density and cumulative distribution functions, respectively, of the standard normal random variable.
Using $f^*(x)$, the density of $X_i$ can be expressed as
\[
f(x) =
\begin{cases}
f^*(x)(1-\bar{p}) & \text{if }x<0\\
f^*(x)\bar{p} & \text{if }x>0,
\end{cases}
\]
and the discontinuity level of $f(x)$ at $x=0$ is $\tau_f = f^*(0)(2\bar{p}-1)$.
Hence, we have $\tau_f = 0$ if and only if $\bar{p}=0.5$, and $\tau_f > 0$ when $\bar{p} > 0.5$.

\section{Additional figures}\label{sec:a_table}
\subsection{Additional power plots}
Appendix Figures 1-6 present the empirical power of the (ii) Bonferroni correction, (iv) Max test, and (v) sWald test under $p=0.575$ and $\tau_{Z_d}(=a)=0,0.5,1,1.5,2$ for different dimensions, and $d=3,5$ for the Monte Carlo experiment in Section 4. In Appendix Figures 1–4, each figure presents the results for different correlations, $\rho = 0, 0.3$, which are weaker than those in Figures 5-8 from the main text.

Appendix Figures 1 and 2 present the results when one of the $d$ covariates jumps, and Appendix Figures 3 and 4 present the results when all $d$ covariates jump. However, each jump size is divided by $d$ for $d=3,5$.
Although the difference between the Max test and Bonferroni correction is smaller, the simulation results show that the proposed joint testing methods exhibit power improvements similar to those shown in Section 4. 
Moreover, the power improvements in the sWald test are more significant when the correlation is weaker. 

Appendix Figures 5 and 6 compare the results for positive and negative correlations when $d=3$ and $|\rho| = 0.9$.
In Appendix Figures 5 (b), and 6 (b), two of the three covariates have the same pairwise correlation coefficient of $0.9$, and the remaining one has the pairwise correlation coefficient of $-0.9$.
Appendix Figure 5 presents the results when one of the three covariates jumps, and Appendix Figure 6 presents the results when all three covariates jump. Nevertheless, each jump size is divided by three.
Although the sWald test is conservative when the correlation is strong, 
the Max test had more power than the Bonferroni correction in both cases. 
This suggests that the power improvements in the Max test are robust against negative correlations.

\begin{figure}[H]
 \centering
 \begin{minipage}[b]{0.85\hsize}
 \includegraphics[width=\hsize]{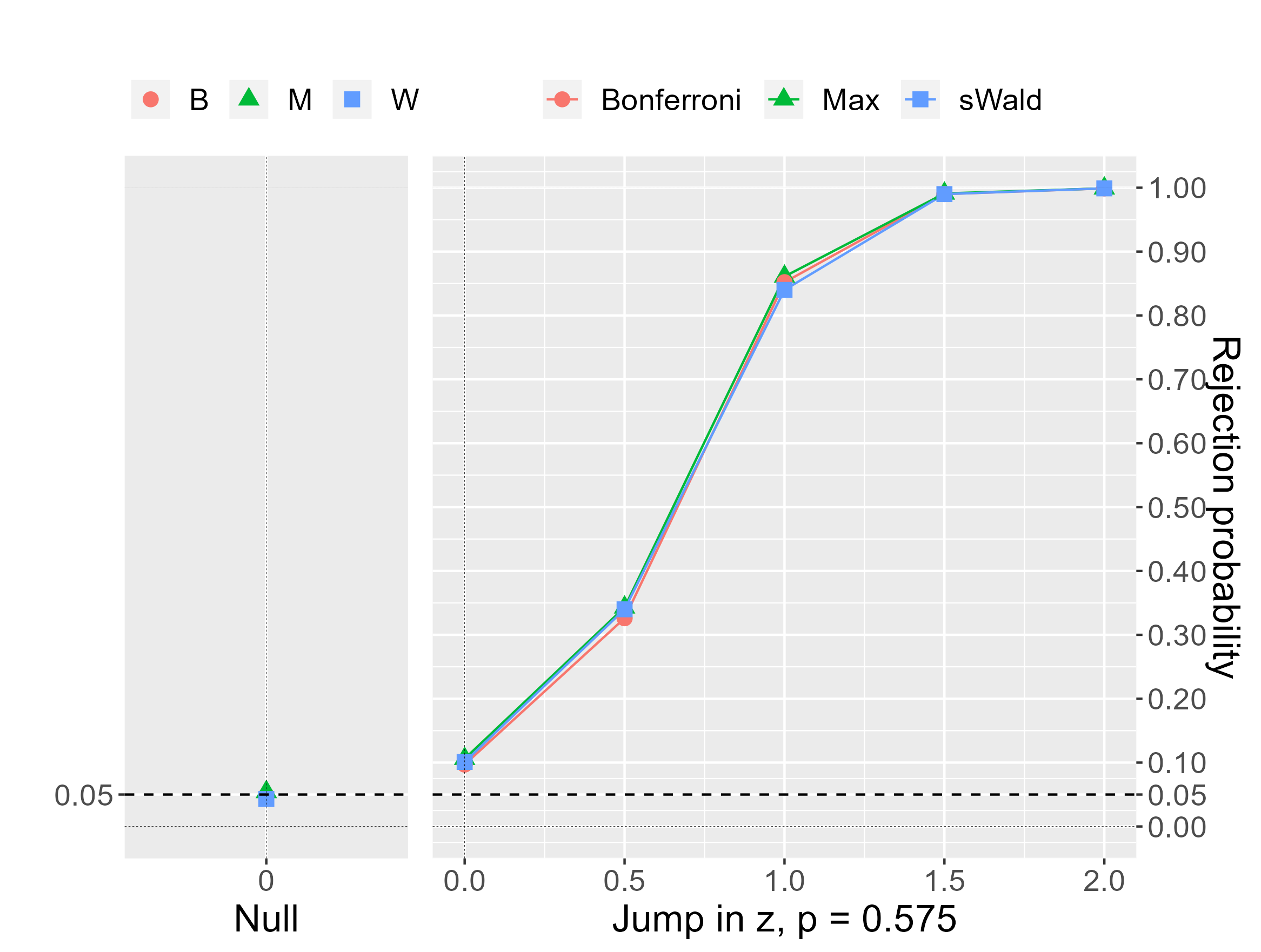}
 \subcaption{Covariates have the same pairwise correlation coefficient of $0$}
 \end{minipage} \\
 \begin{minipage}[b]{0.85\hsize}
 \includegraphics[width=\hsize]{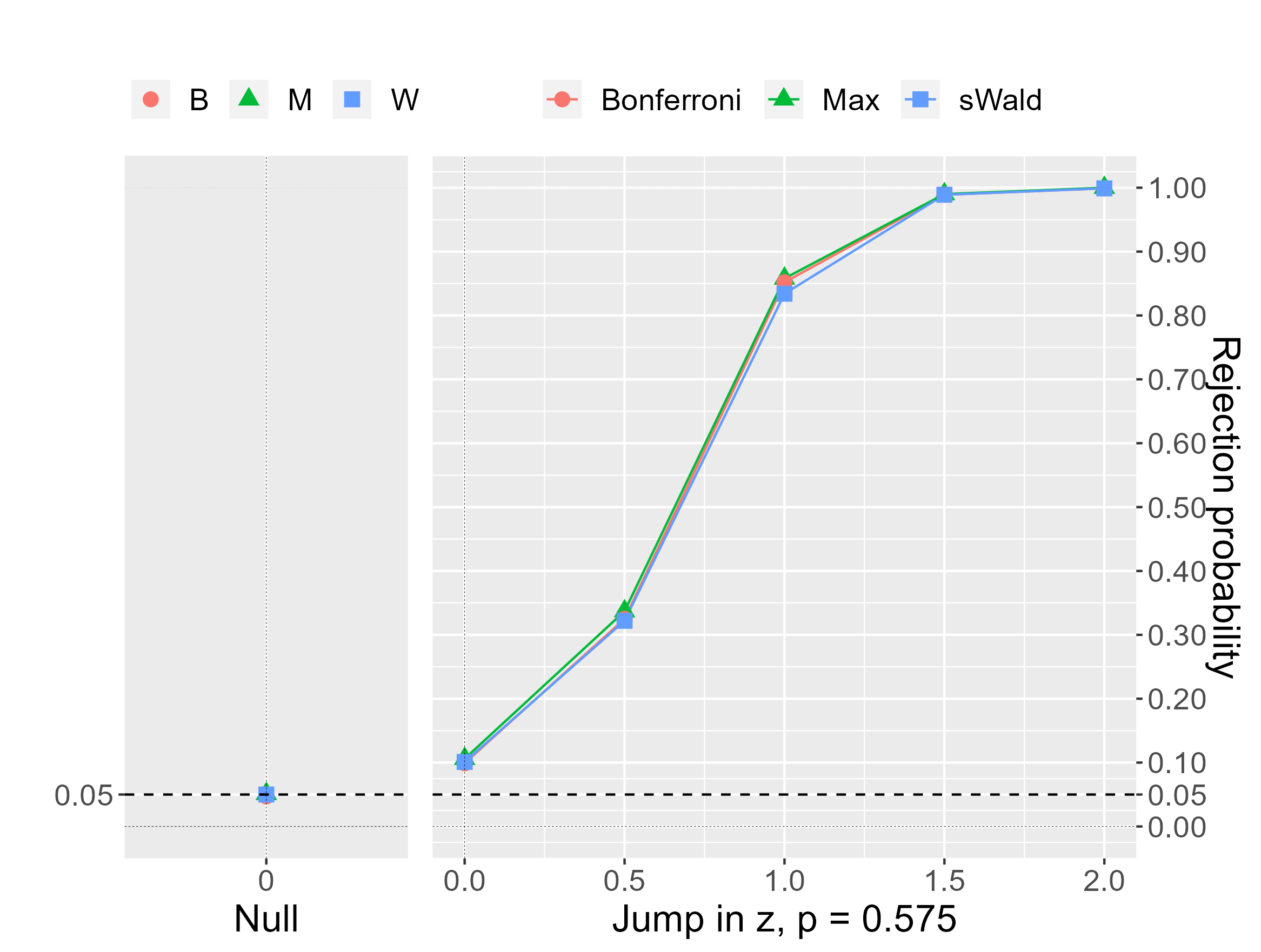}
 \subcaption{Covariates have the same pairwise correlation coefficient of $0.3$}
 \end{minipage} 
 \caption*{Appendix Figure 1: A jump in the three covariates with a density, $n = 1000$}\label{fig:5.3.last.small} 
\end{figure}

\begin{figure}[H]
 \centering
 \begin{minipage}[H]{0.85\hsize}
 \includegraphics[width=\hsize]{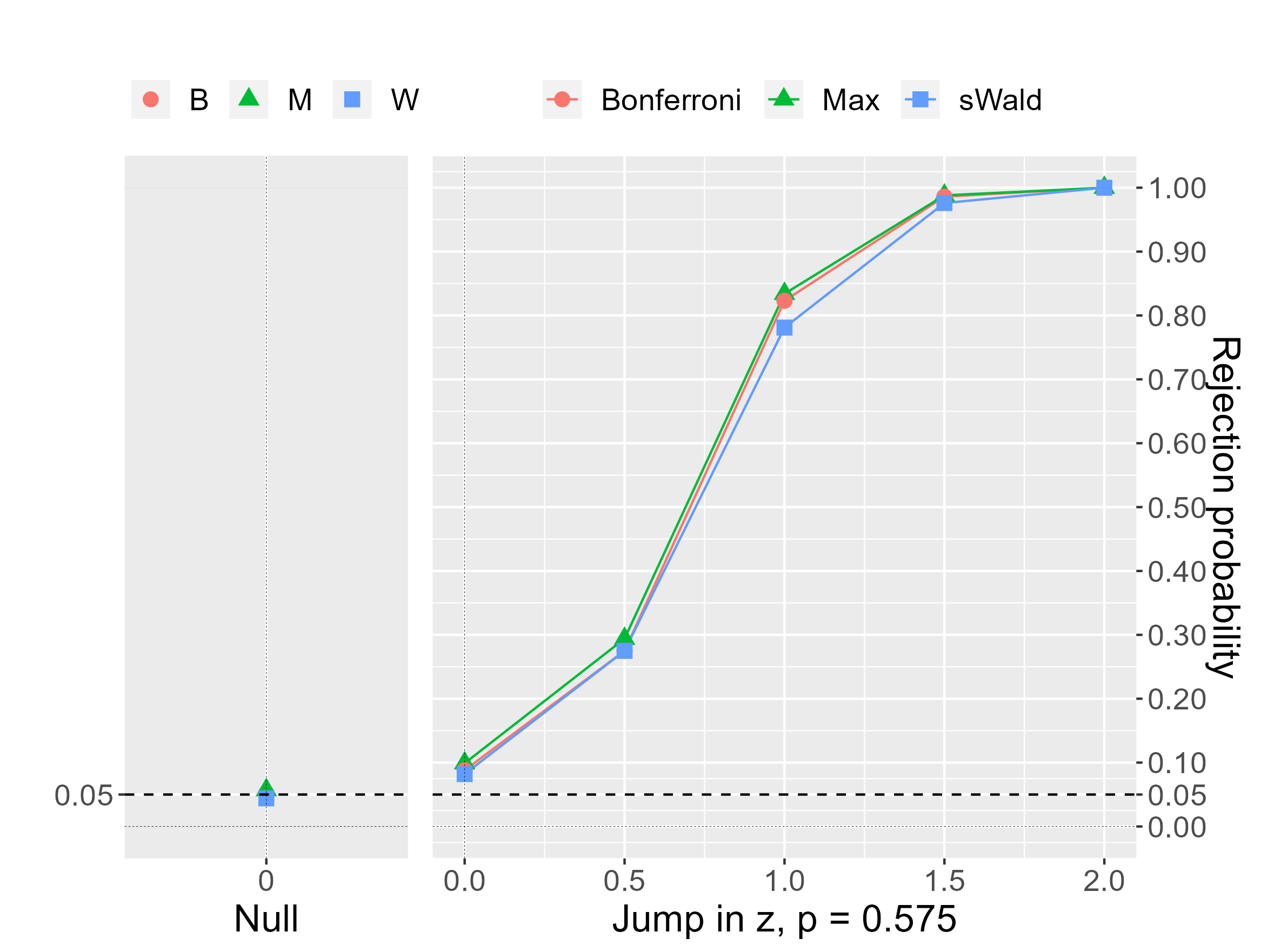}
 \subcaption{Covariates have the same pairwise correlation coefficient of $0$}
 \end{minipage} \\
 \begin{minipage}[b]{0.85\hsize}
 \includegraphics[width=\hsize]{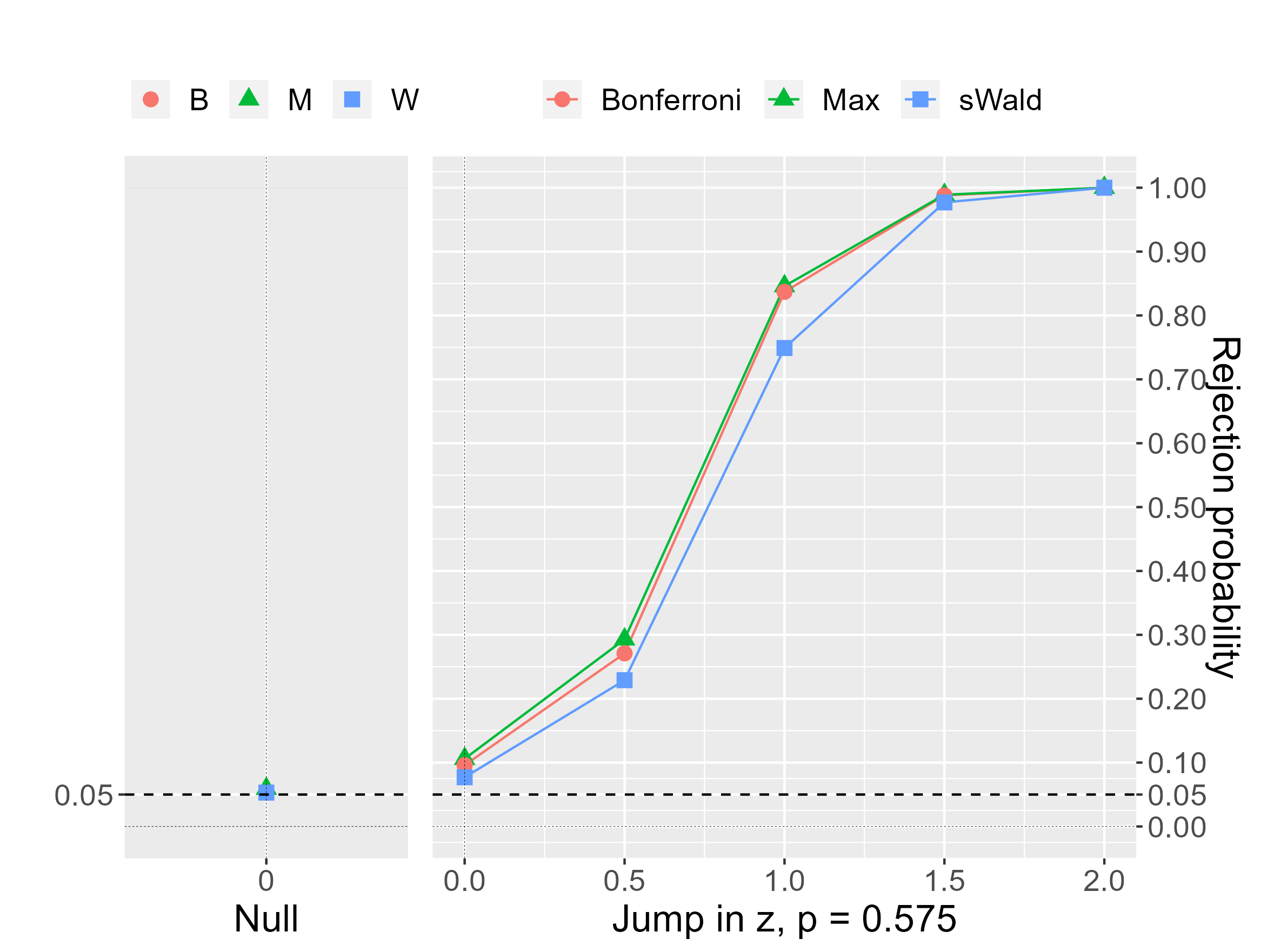}
 \subcaption{Covariates have the same pairwise correlation coefficient of $0.3$}
 \end{minipage} 
 \caption*{Appendix Figure 2: A jump in the five covariates with a density, $n = 1000$}\label{fig:5.5.last.small} 
\end{figure}

\begin{figure}[H]
 \centering
 \begin{minipage}[b]{0.85\hsize}
 \includegraphics[width=\hsize]{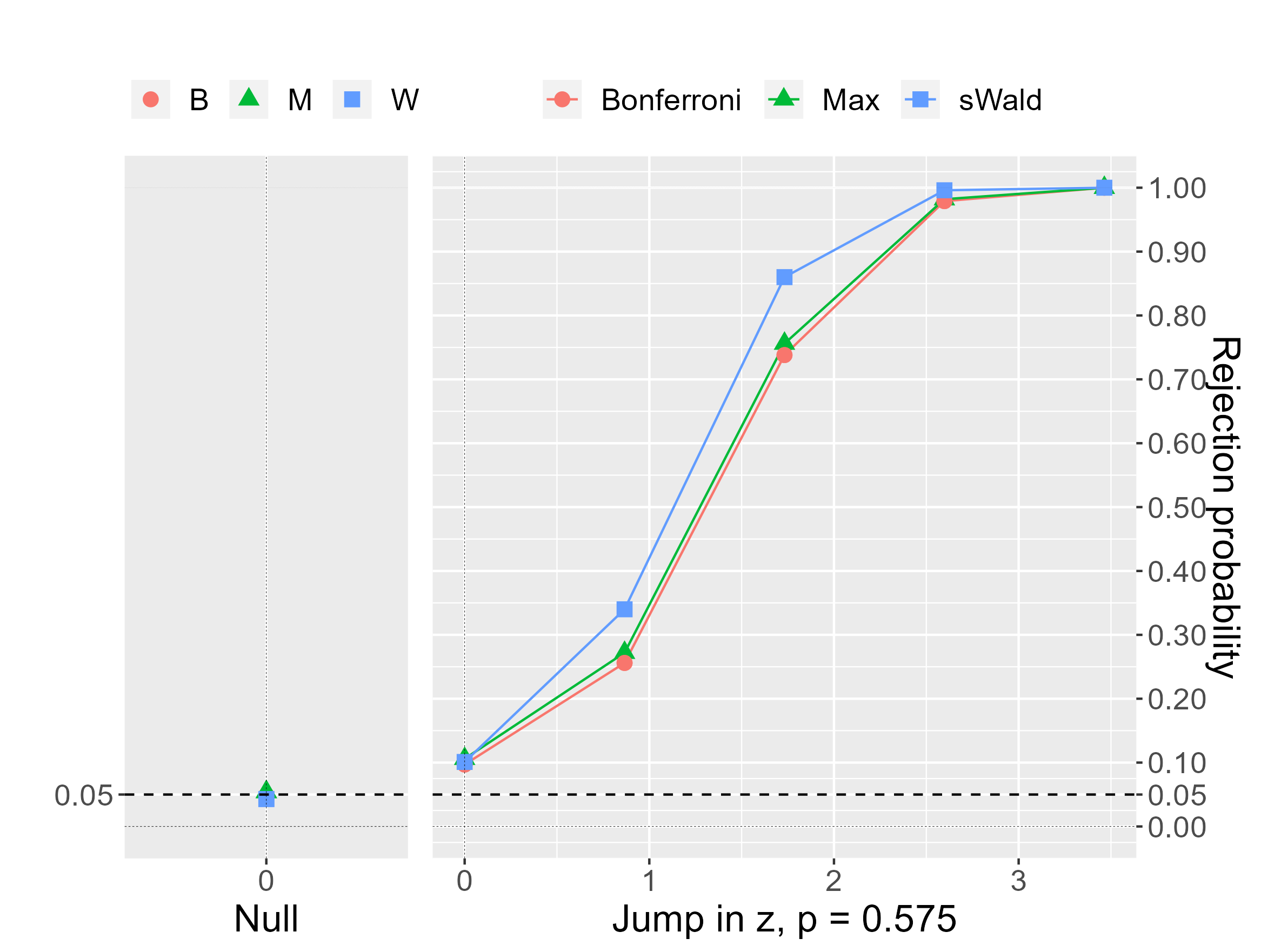}
 \subcaption{Covariates have the same pairwise correlation coefficient of $0$}
 \end{minipage} \\
 \begin{minipage}[b]{0.85\hsize}
 \includegraphics[width=\hsize]{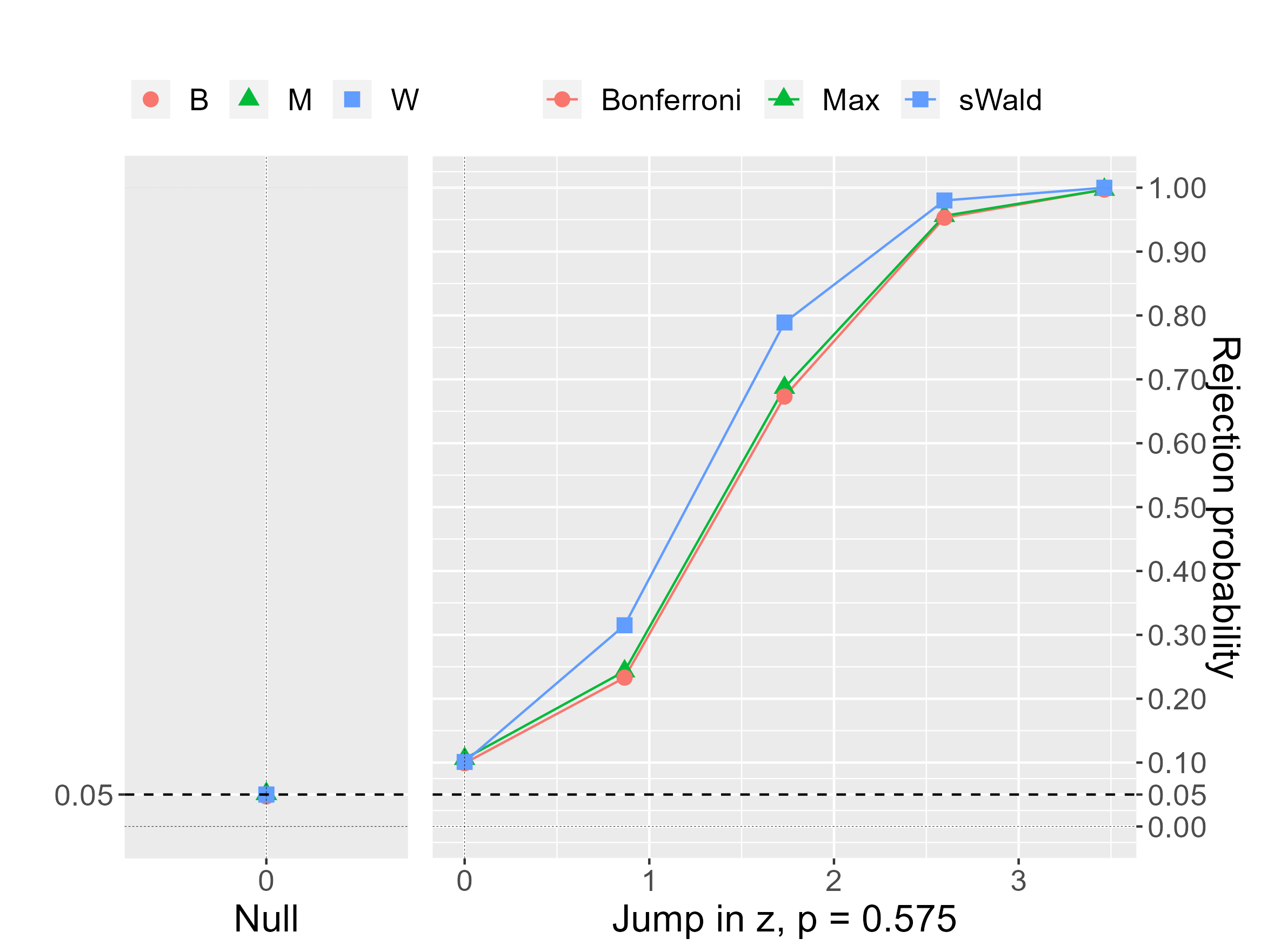}
 \subcaption{Covariates have the same pairwise correlation coefficient of $0.3$}
 \end{minipage} 
 \caption*{Appendix Figure 3: $1/3$ of jumps in all the three covariates with a density, $n = 1000$.}\label{fig:5.3.all.small} 
\end{figure}

\begin{figure}[H]
 \centering
 \begin{minipage}[b]{0.85\hsize}
 \includegraphics[width=\hsize]{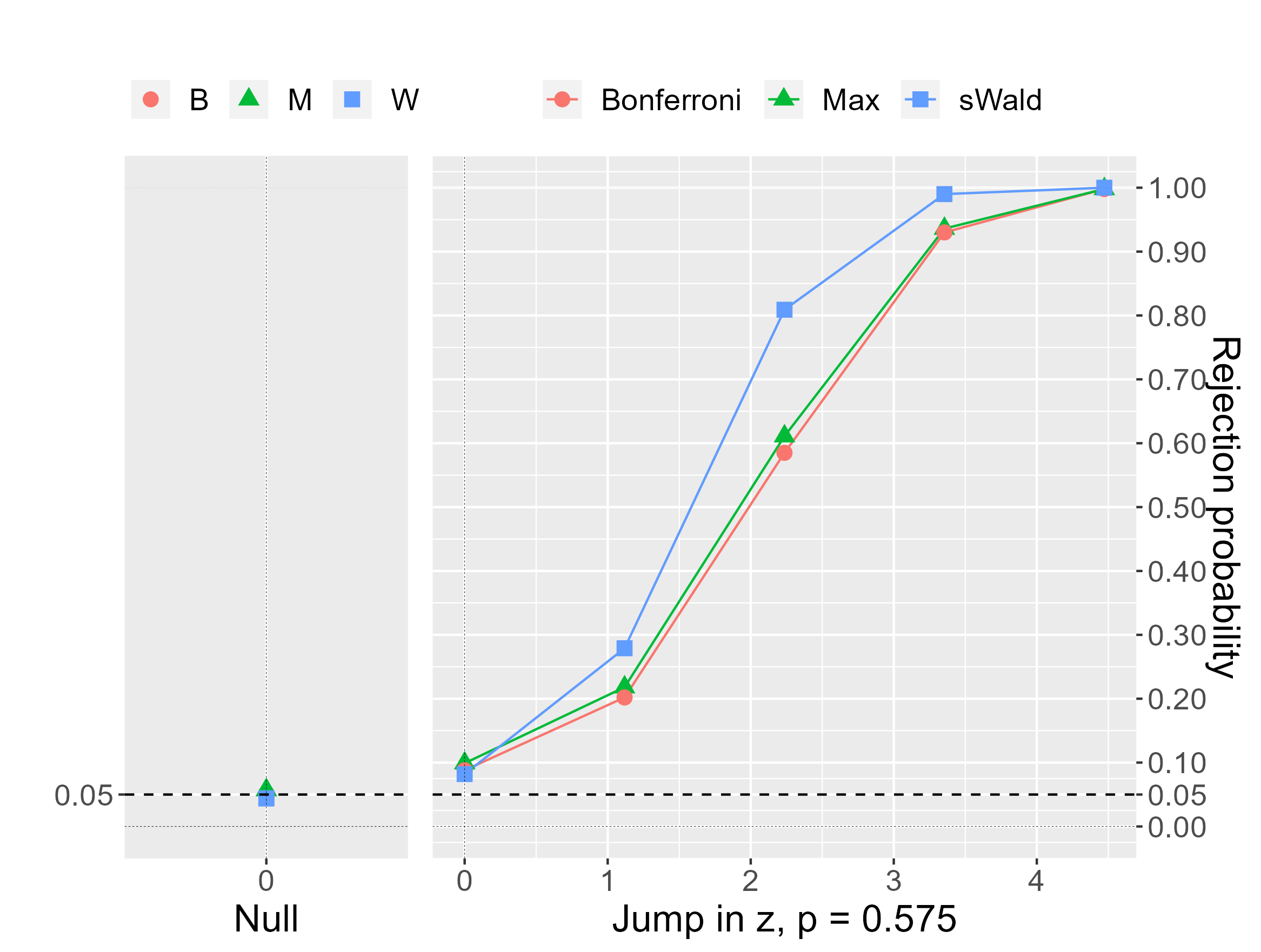}
 \subcaption{Covariates have the same pairwise correlation coefficient of $0$}
 \end{minipage} \\
 \begin{minipage}[b]{0.85\hsize}
 \includegraphics[width=\hsize]{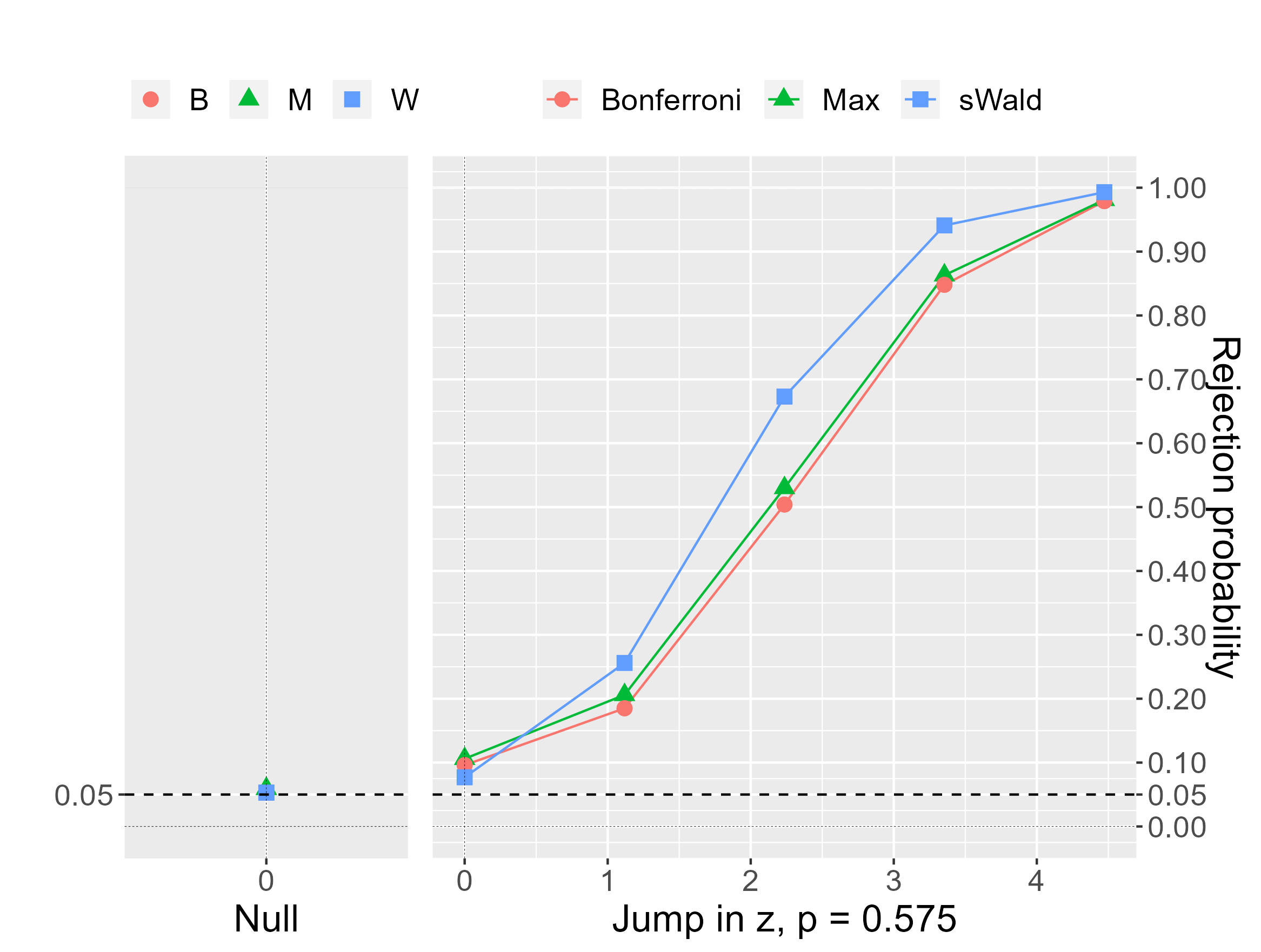}
 \subcaption{Covariates have the same pairwise correlation coefficient of $0.3$}
 \end{minipage} 
 \caption*{Appendix Figure 4: $1/5$ of jumps in all the five covariates with a density, $n = 1000$.}\label{fig:5.5.all.small} 
\end{figure}

\begin{figure}[H]
 \centering
 \begin{minipage}[b]{0.85\hsize}
 \includegraphics[width=\hsize]{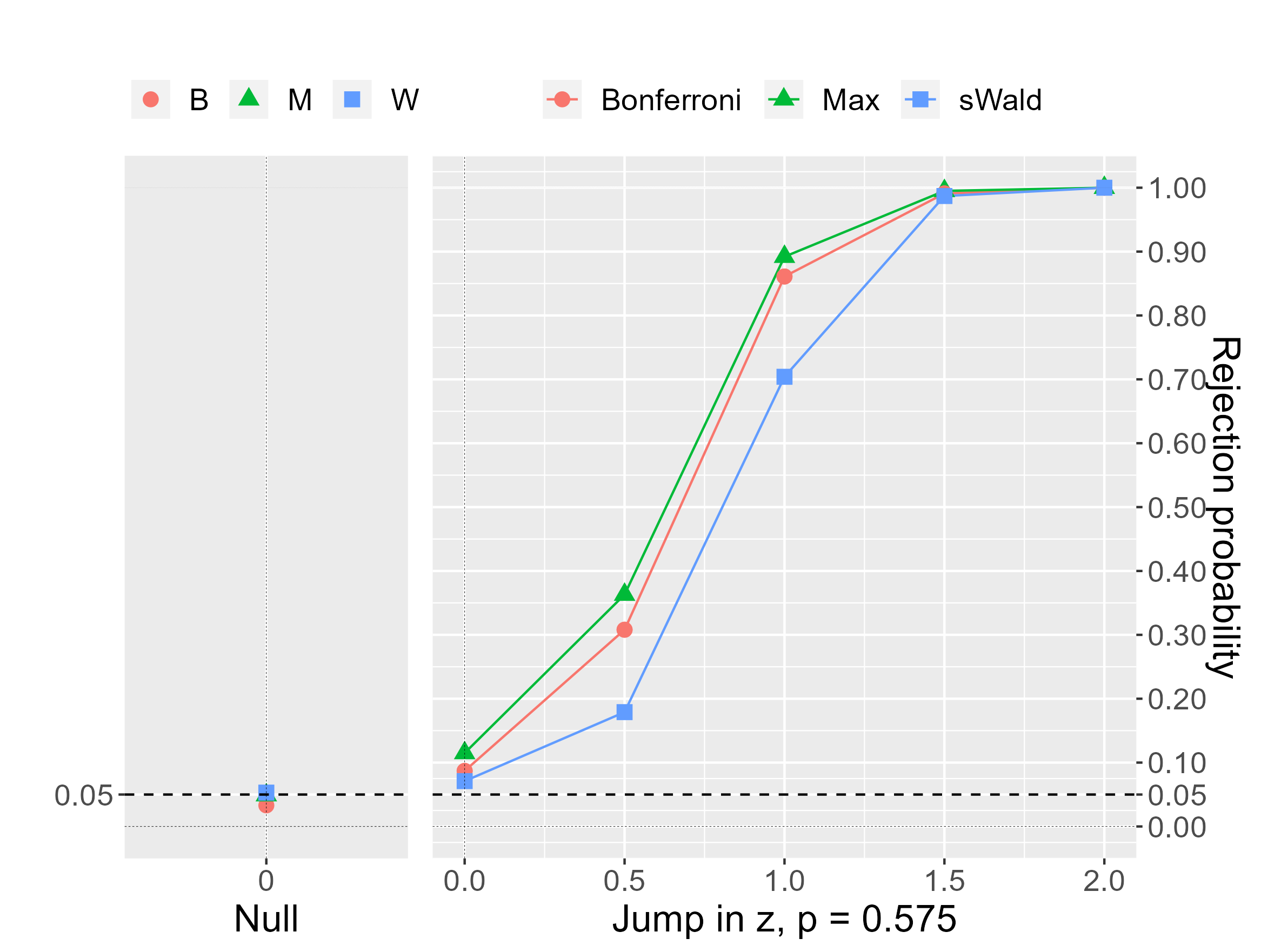}
 \subcaption{The remaining one has the same pairwise correlation coefficient of $0.9$.}
 \end{minipage} \\
 \begin{minipage}[b]{0.85\hsize}
 \includegraphics[width=\hsize]{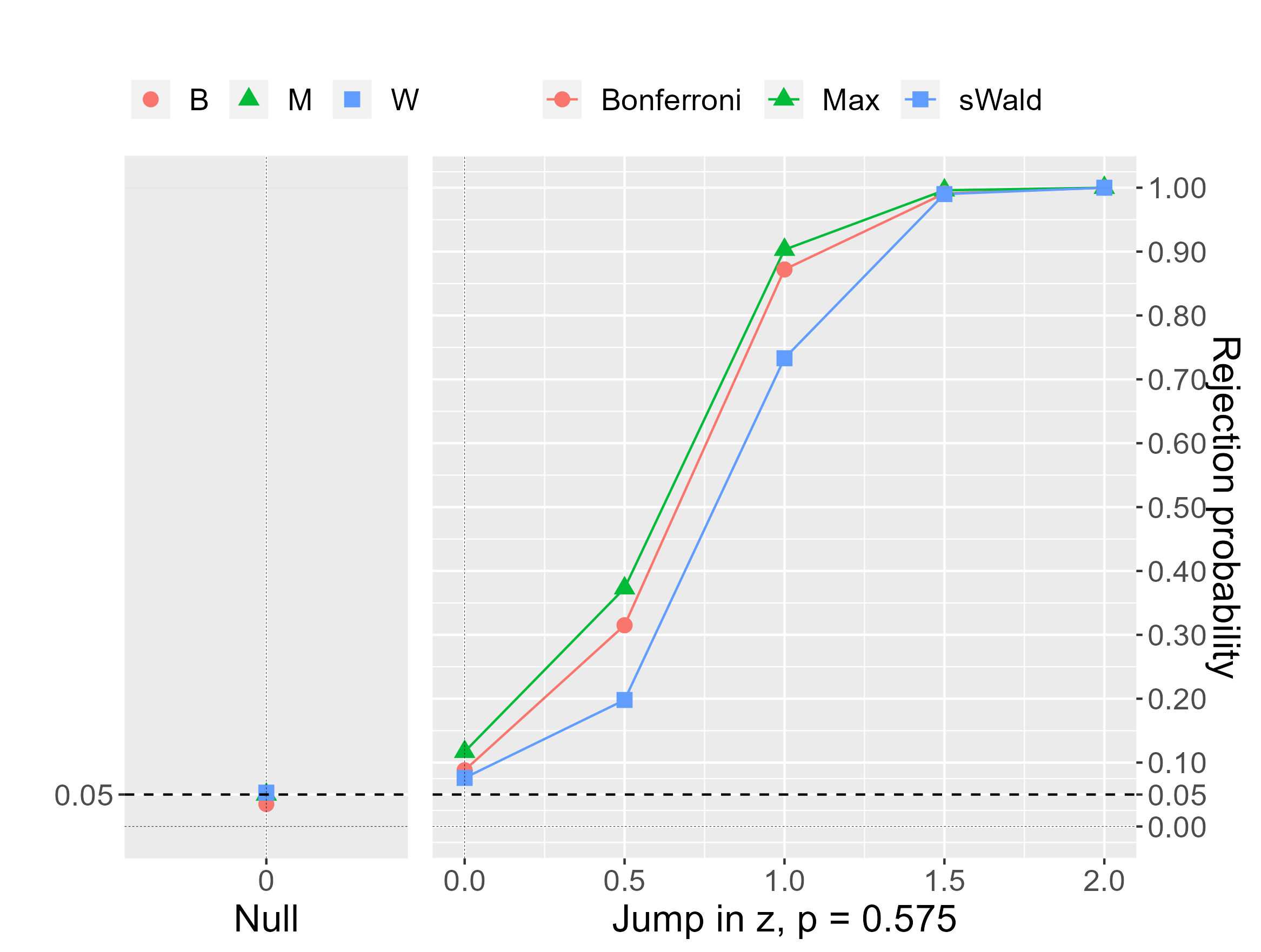}
 \subcaption{The remaining one has the pairwise correlation coefficient of $-0.9$.} 
 \end{minipage} 
 \caption*{Appendix Figure 5: A jump in the three covariates with a density, $n = 1000$. Two out of three covariates have the same pairwise correlation coefficient of $0.9$}\label{fig:5.3.last.neg} 
\end{figure}

\begin{figure}[H]
 \centering
 \begin{minipage}[b]{0.85\hsize}
 \includegraphics[width=\hsize]{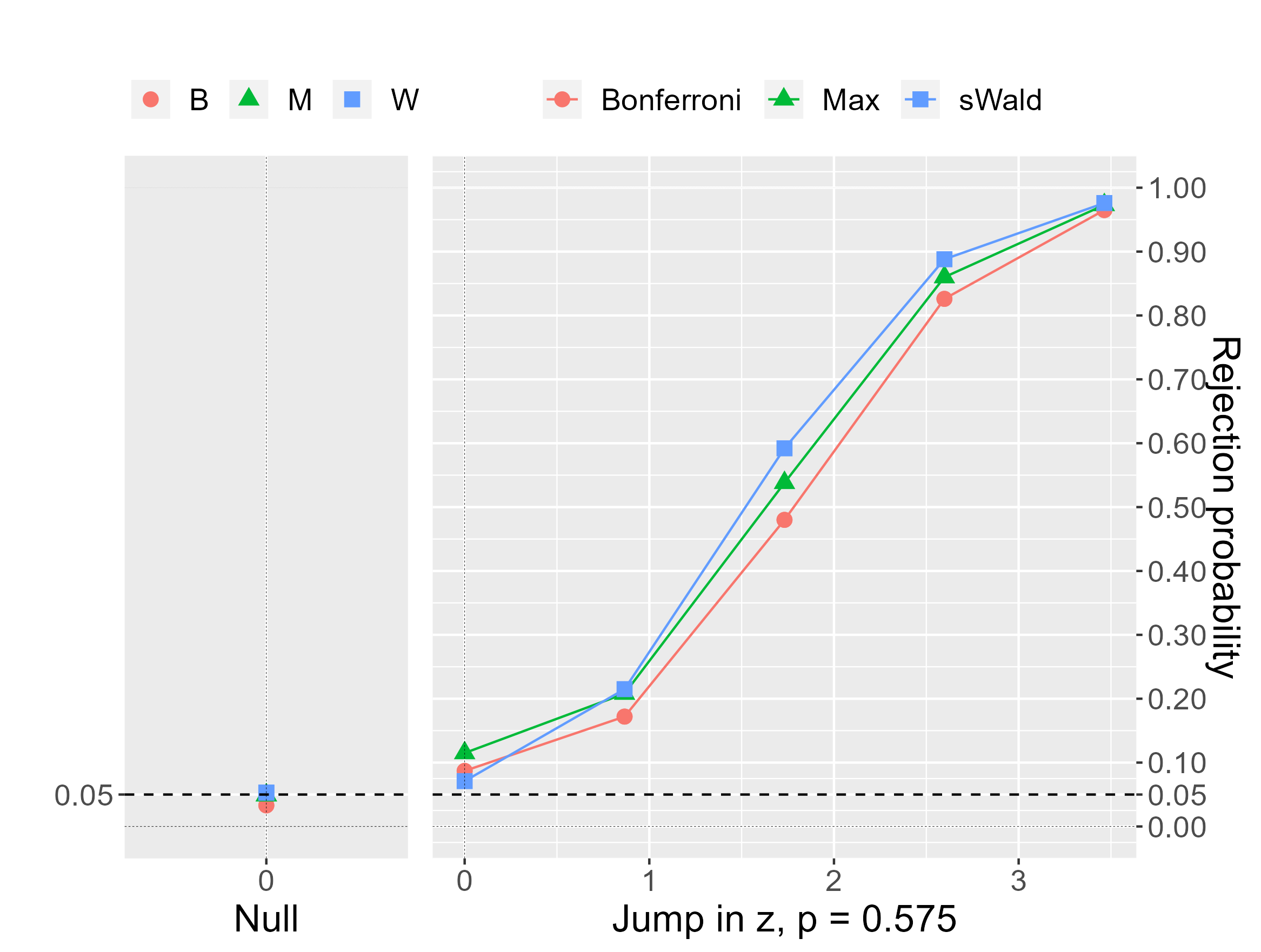}
 \subcaption{The remaining one has the same pairwise correlation coefficient of $0.9$.}
 \end{minipage} \\
 \begin{minipage}[b]{0.85\hsize}
 \includegraphics[width=\hsize]{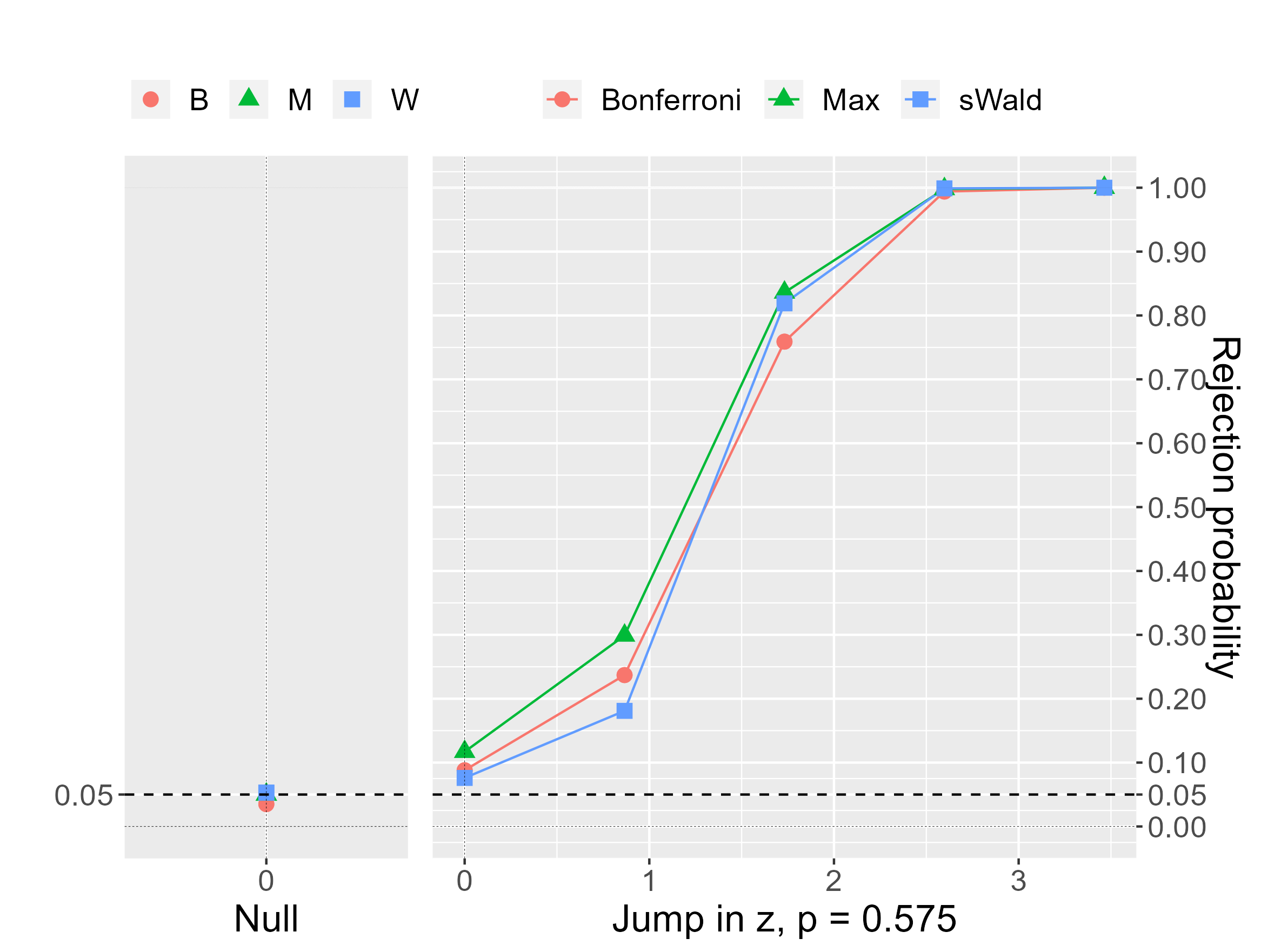}
 \subcaption{The remaining one has the pairwise correlation coefficient of $-0.9$.} 
 \end{minipage} 
 \caption*{Appendix Figure 6: $1/3$ of jumps in all the three covariates with a density, $n = 1000$. Two out of three covariates have the same pairwise correlation coefficient of $0.9$}\label{fig:5.3.all.neg} 
\end{figure}

\subsection{Additional meta-analysis plots}

\begin{figure}[H]
    \centering
 \caption*{Appendix Figure 7: Histogram of p-values from rejected tests.}
    \includegraphics[width=\hsize]{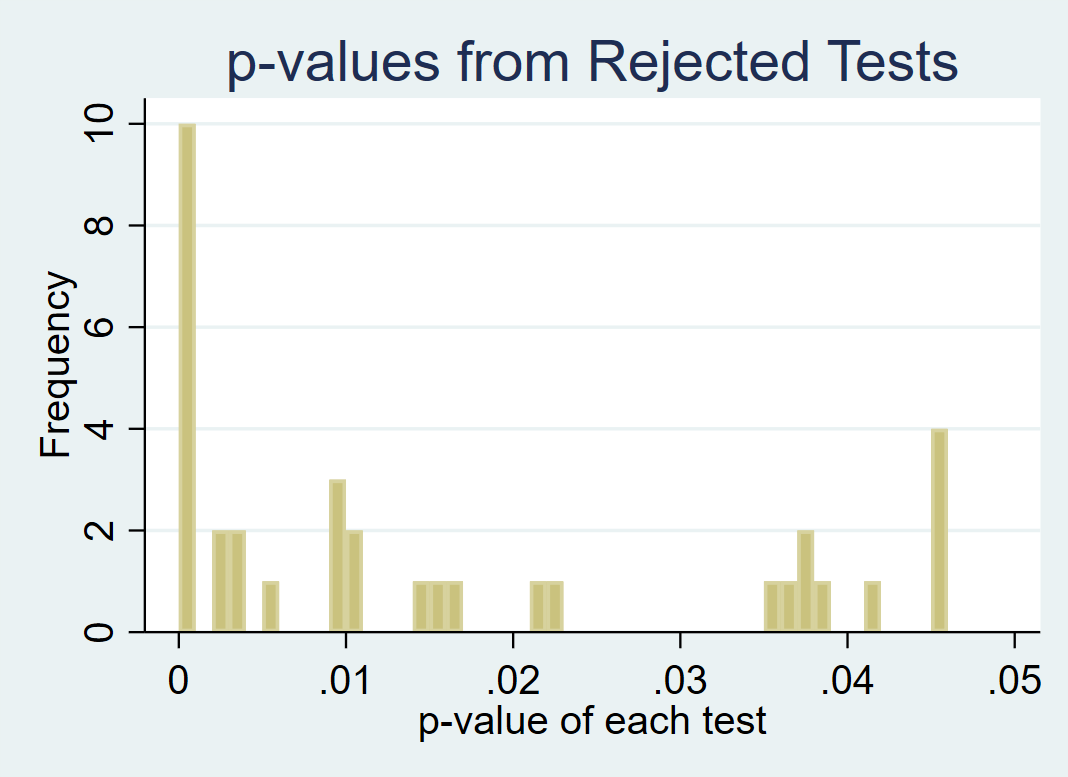}
    \label{fig:figure_E_1}
     \begin{minipage}{380pt}
    {\fontsize{10pt}{10pt}\selectfont\smallskip\textit{Notes}: The distribution of the $34$ tests are fairly uniform over $0$ to $0.05$. In the histogram of Appendix Figure 7, there are two masses due to rounding: the mass at $0$ is rounding to $0$ in the original reported p-value in the papers; the mass near $0.05$ is at $0.0455$ which is the p-value corresponding the computed t-statistics of $2$ by dividing rounded numbers such as $0.02$ over $0.01$.}
    \end{minipage}
\end{figure}


\begin{figure}[H]
    \centering
    \includegraphics[width=\hsize]{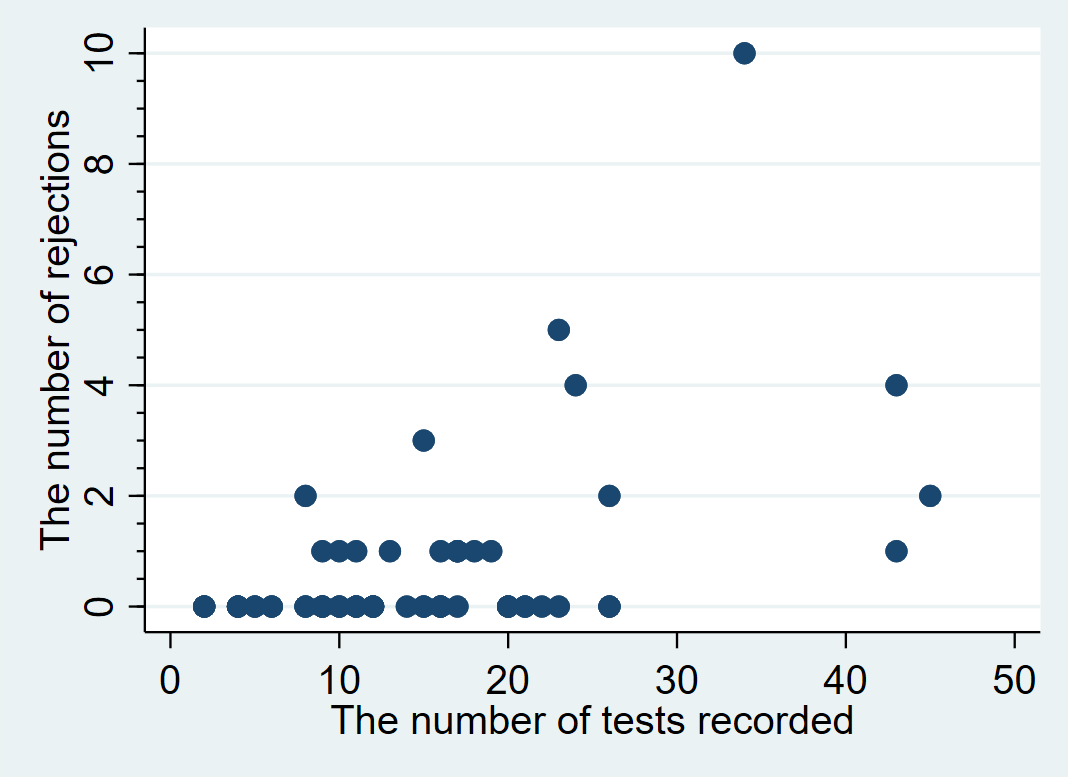}
 \caption*{Appendix Figure 8: Scatter plot of the number of tests and the number of rejected tests.}
    \label{fig:figure_E_3}
\end{figure}

\subsection{Additional power plots to compare with rwolf2 for balance tests only}

\begin{figure}[H]
    \centering
     \begin{minipage}[b]{0.8\hsize}
    \includegraphics[width=0.9\hsize]{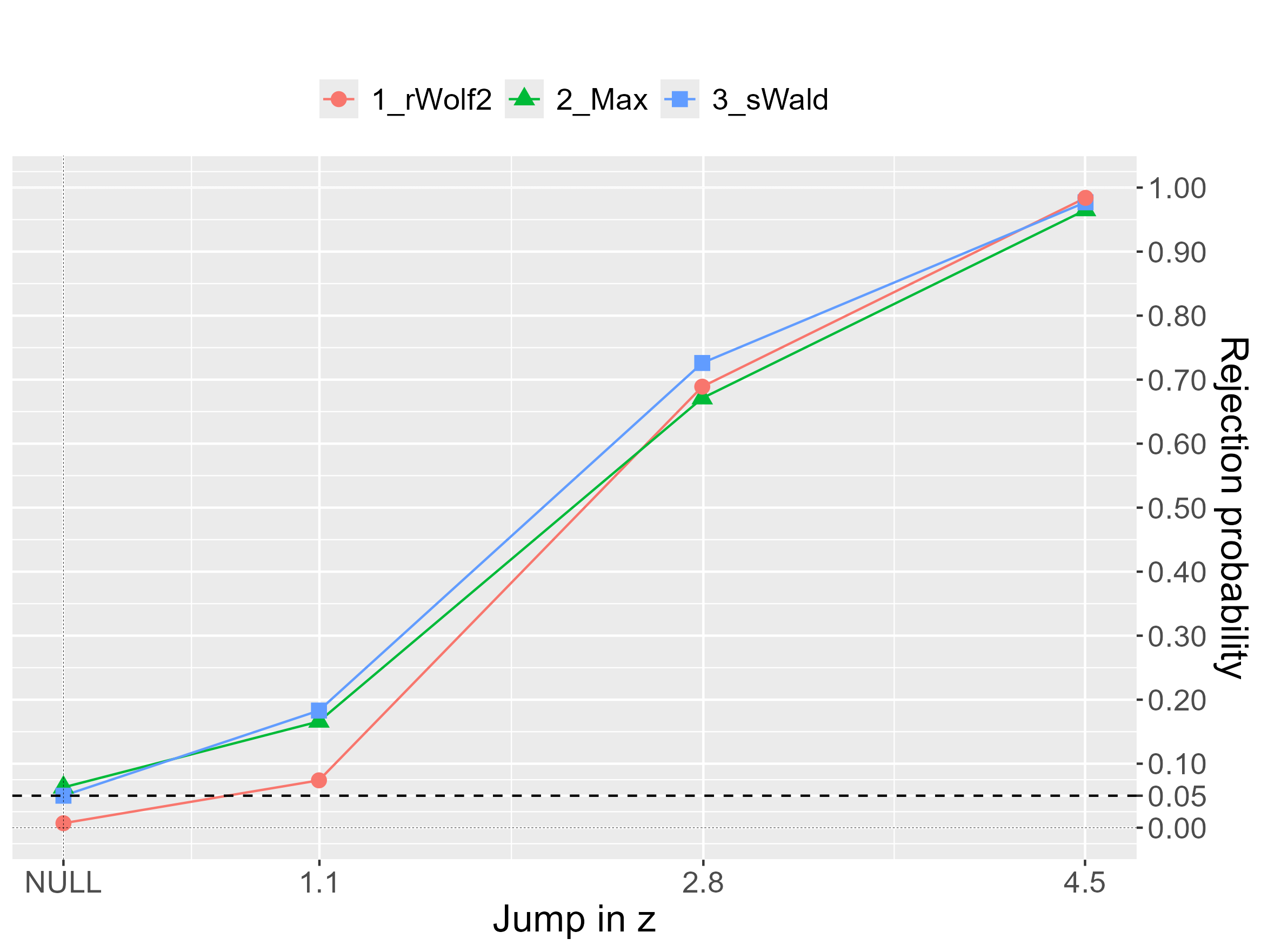}
 \caption*{Appendix Figure 9: The unified balance test results compared with rwolf2. $1/5$ of jumps in all the five covariates. Five covariates have the pairwise correlation coefficient of $0.5$}    \label{fig:rwolf_05_all}
    \end{minipage} \\
    \begin{minipage}[b]{0.8\hsize}
    \includegraphics[width=0.9\hsize]{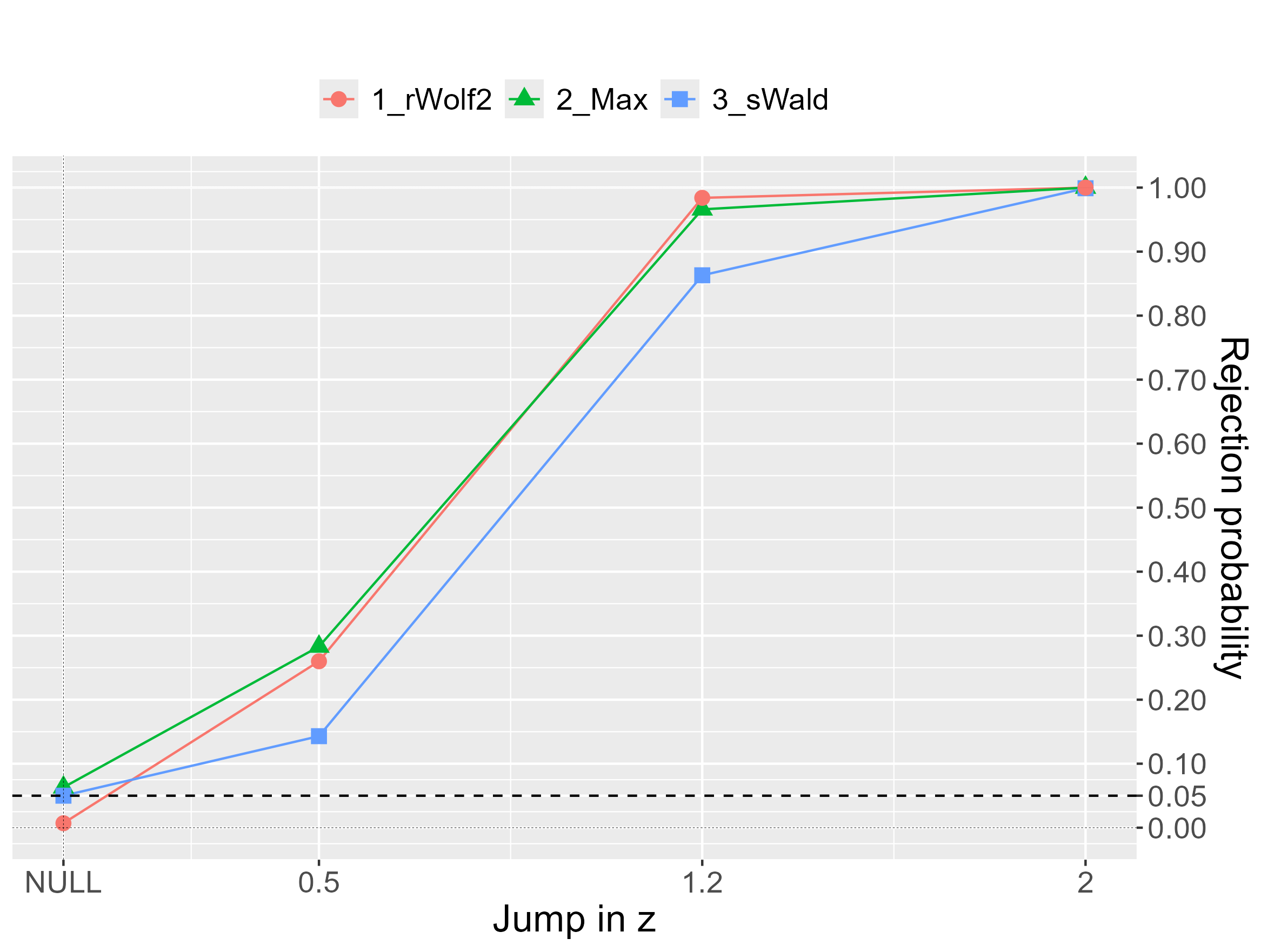}
 \caption*{Appendix Figure 10: The unified balance test results compared with rwolf2. One of five covariates have the jump. Five covariates have the pairwise correlation coefficient of $0.5$}
    \label{fig:rwolf_05_last}
    \end{minipage}
\end{figure}

\begin{figure}[H]
    \centering
     \begin{minipage}[b]{0.8\hsize}
    \includegraphics[width=0.9\hsize]{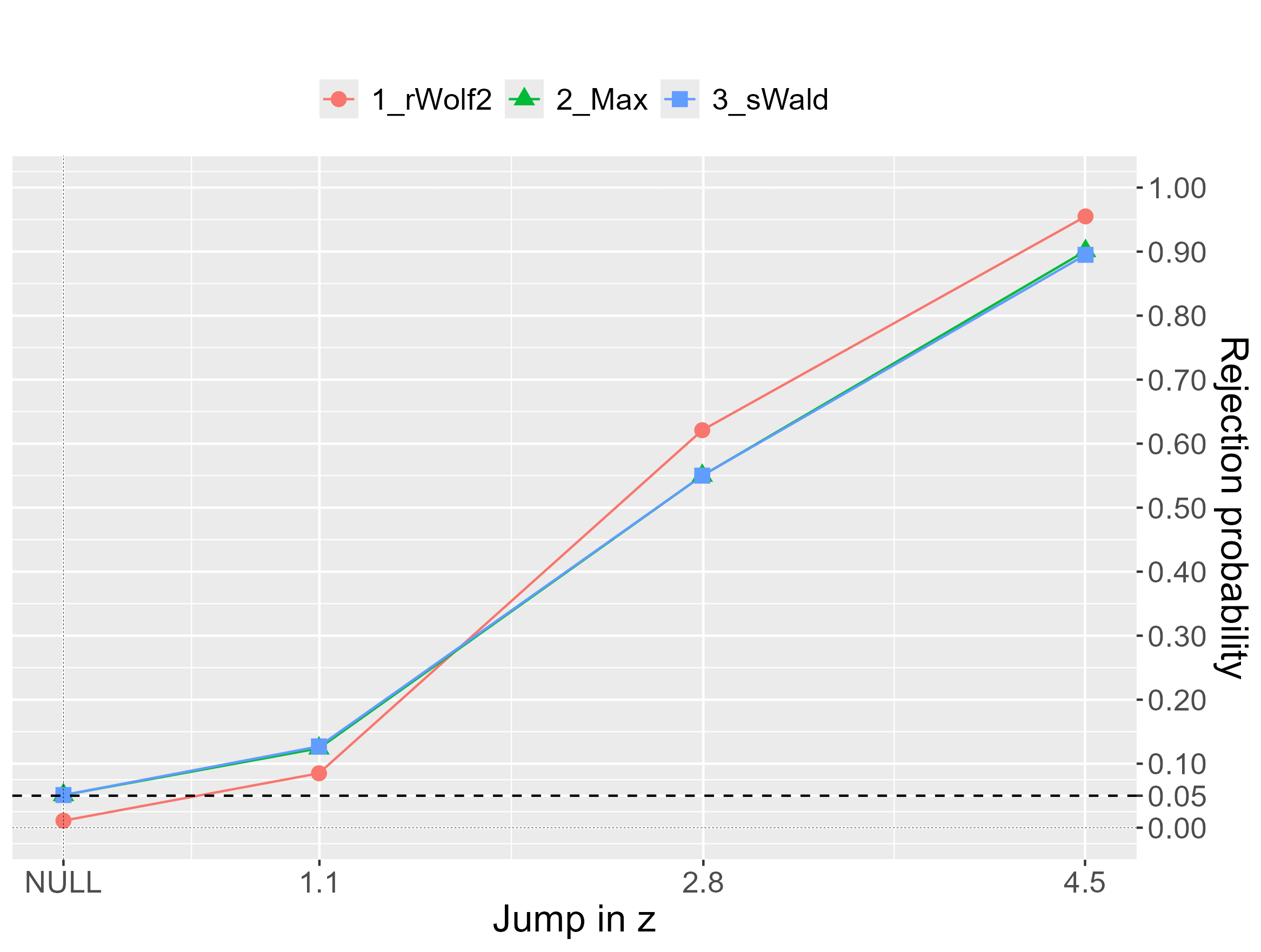}
 \caption*{Appendix Figure 11: The unified balance test results compared with rwolf2. $1/5$ of jumps in all the five covariates. Five covariates have the pairwise correlation coefficient of $0.9$}    \label{fig:rwolf_05_all}
    \end{minipage} \\
    \begin{minipage}[b]{0.8\hsize}
    \includegraphics[width=0.9\hsize]{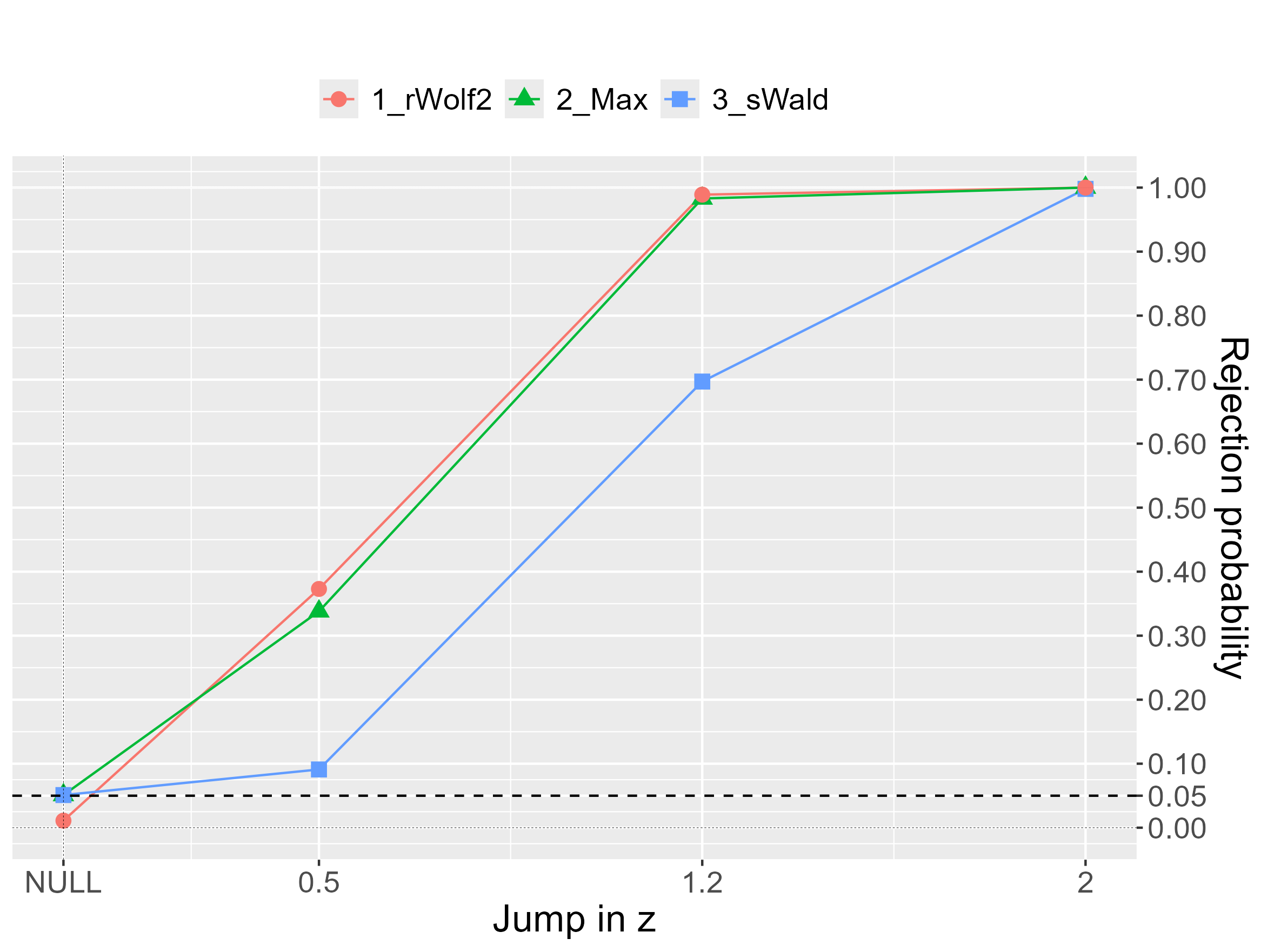}
 \caption*{Appendix Figure 12: The unified balance test results compared with rwolf2. One of five covariates have the jump. Five covariates have the pairwise correlation coefficient of $0.9$}
    \label{fig:rwolf_05_last}
    \end{minipage}
\end{figure}

\section{Further discussions for the pretest analysis} \label{sec:a-pretest-general}

Given the general case of the pretest analysis, we ran a few numerical analysis based on a simulated toy dataset and calibrated datasets from two empirical studies. 

First, we consider a null RD design with a sample size of $3000$ and four covariates. In the toy dataset, we generate the outcome using the following model
\[
 Y_i = \epsilon_i + Z_i' \beta
\]
where $\epsilon_i \sim N(0,1)$, $Z_i \sim N(0,I_4)$ and $\beta = [-0.05,0.1,0.2,-0.1]'$. Second, we use the dataset from two empirical studies, \cite{pons_tricaud_2018} and \cite{Asher_Novosad_2020}. For \cite{pons_tricaud_2018}, we use the five covariates for the balance test in the original study to evaluate the impact on the first stage binary treatment status. For \cite{Asher_Novosad_2020}, we use the six covariates for the balance test in the original study to evaluate the impact on the first stage binary treatment status.

For three datasets, we conduct the following analysis as in  \cite{roth2022pretest}. First, we set $\delta_Z$ so that the max test has $80\%$ power. Second, run \textit{rdrobust} with covariates option and obtain the treatment effect estimate $\hat{\tau}$ and the weights on the covariates $\hat{\beta}$. Finally, compute the mean $\delta_Y$ according to the linear model with $\hat{\beta}$ as the unconditional bias, and the conditional bias which is the mean $\delta_Y$ among simulated samples which accepts the null for the data generating process of $\delta_Z$ for the $80\%$ power.

Appendix Figure 13 summarizes the results. In all three results, the unconditional bias is reduced towards the original estimate by conditioning on the draws which pass the max test. Hence, pretesting is bias-reducing at least for these cases.

\begin{figure}
    \centering
    \includegraphics[width=1\linewidth]{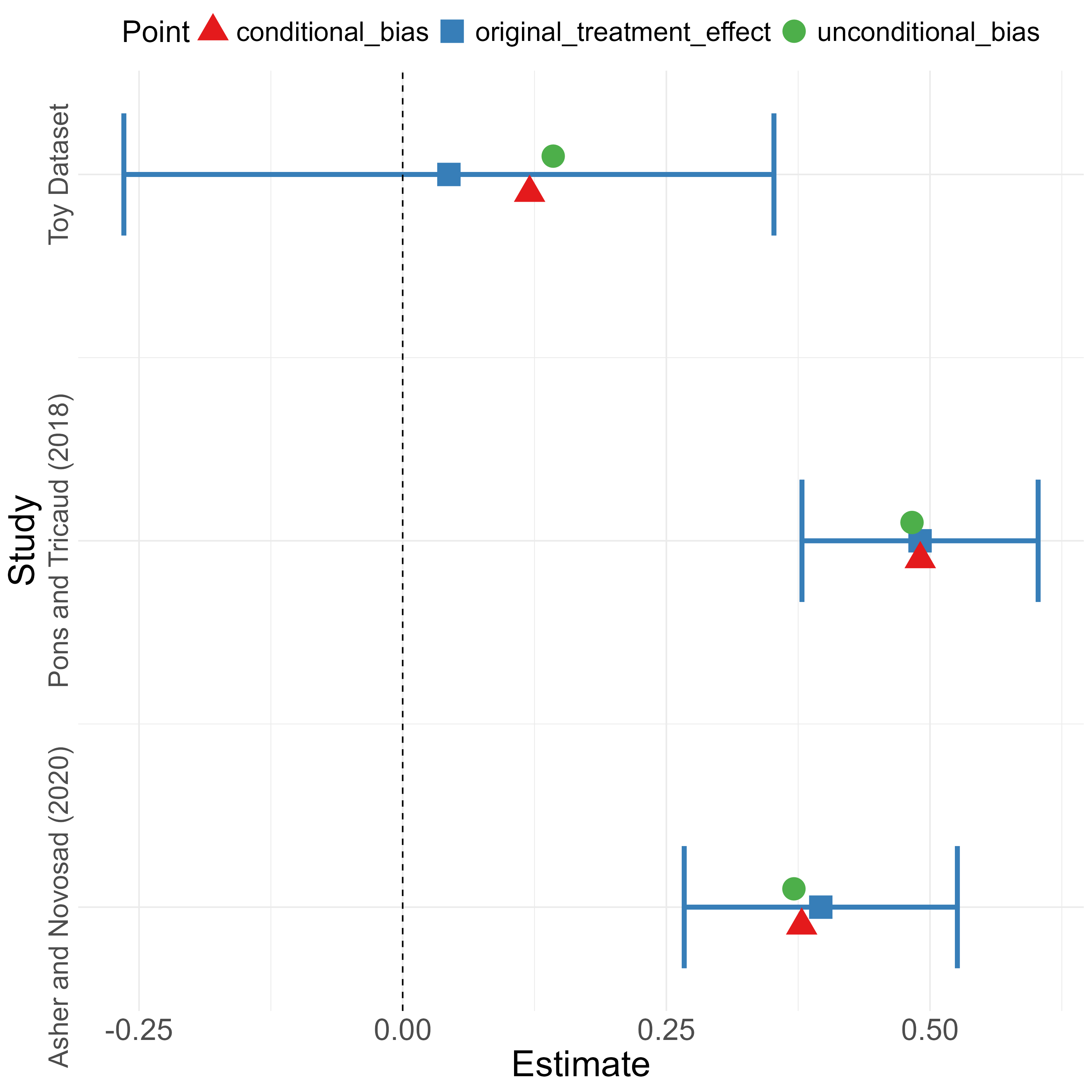}
 \caption*{Appendix Figure 13: Pretest numerical results}
    \label{fig:result_pretest}
     \begin{minipage}{380pt}
    {\fontsize{10pt}{10pt}\selectfont\smallskip\textit{Notes}: The figure compares the original estimates and confidence intervals with two points of the unconditional bias and (pretest) conditional bias plus the point estimates from a toy dataset and two empirical applications. The square points denote the original (bias corrected) point estimate, the range denote the confidence interval, the circle points denote the unconditional bias, the triangular points denote the conditional bias. In all three cases, conditioning reduces the bias towards the original point estimates.}
    \end{minipage}
\end{figure}

\FloatBarrier


\bibliography{reference}  

\end{appendix}
\end{document}